\documentclass[a4paper]{article}
\usepackage[margin=1.5in]{geometry}
\usepackage{tikz}
\usetikzlibrary{positioning,chains,fit,shapes,calc}
\usepackage{amsmath,amsthm,amssymb}
\usepackage{todonotes}
\usetikzlibrary{automata, positioning}
\usepackage[mode=buildnew]{standalone}
\usepackage[most]{tcolorbox}
\usepackage{hyperref} 
\usepackage{subcaption}

\newtheorem{theorem}{Theorem}%
\newtheorem{lemma}[theorem]{Lemma}
\newtheorem{corollary}[theorem]{Corollary}
\newtheorem{claim}[theorem]{Claim}
\newtheorem{definition}[theorem]{Definition}
\newtheorem{example}[theorem]{Example}
\newtheorem{proposition}[theorem]{Proposition}
\newtheorem{fact}{Fact}

\newenvironment{claimproof}[1]{\par\noindent\underline{Proof:}\space#1}{\hfill $\blacksquare$}

\newcommand{\dom}{\mathsf{dom}}

\newcommand{\quant}{\mathsf{quant}}
\newcommand{\points}{\mathsf{points}}

\newcommand{\act}{\mathsf{active}}
\newcommand{\QN}{\mathsf{QN}}
\newcommand{\QD}{\mathsf{QD}}
\newcommand{\type}{\mathsf{type}}
\newcommand{\LHS}{\mathsf{LHS}}
\newcommand{\RHS}{\mathsf{RHS}}
\newcommand{\good}{\mathsf{good}}

\newcommand{\struc}[1]{\mathcal{#1}}
\newcommand{\strucA}{\mathcal A}
\newcommand{\strucB}{\mathcal B}
\newcommand{\strucC}{\mathcal C}
\newcommand{\strucD}{\mathcal D}
\newcommand{\strucX}{\mathcal X}
\newcommand{\strucY}{\mathcal Y}

\newcommand{\strucSA}{S(\mathcal A)}
\newcommand{\strucSB}{S(\mathcal B)}
\newcommand{\strucSX}{S(\mathcal X)}

\newcommand{\assA}{\alpha}
\newcommand{\assB}{\beta}
\newcommand{\assC}{\gamma}
\newcommand{\assD}{\delta}

\newcommand{\lstrucAa}{\strucA^{\assA}}
\newcommand{\lstrucAb}{\strucA^{\assB}}
\newcommand{\lstrucAc}{\strucA^{\assC}}

\newcommand{\lstrucBb}{\strucB^{\assB}}
\newcommand{\lstrucBc}{\strucB^{\assC}}
\newcommand{\lstrucBd}{\strucB^{\assD}}

\newcommand{\lstrucCa}{\strucC^{\assA}}

\newcommand{\lstrucCc}{\strucC^{\assC}}

\newcommand{\lstrucDb}{\strucD^{\assB}}

\newcommand{\lstrucDd}{\strucD^{\assD}}

\newcommand{\lstrucXa}{\strucX^{\assA}}

\newcommand{\lstrucYb}{\strucY^{\assB}}

\newcommand{\lstrucSAa}{S(\strucA)^{\assA}}
\newcommand{\lstrucSBb}{S(\strucB)^{\assB}}
\newcommand{\lstrucSBc}{S(\strucB)^{\assC}}

\newcommand{\lstruc}[2]{\struc{#1}^{#2}}

\newcommand{\classA}{\mathfrak A}
\newcommand{\classB}{\mathfrak B}
\newcommand{\classC}{\mathfrak C}
\newcommand{\classD}{\mathfrak D}
\newcommand{\classX}{\mathfrak X}
\newcommand{\classY}{\mathfrak Y}

\newcommand{\logic}{\mathcal L}
\newcommand{\klogic}{\mathcal{L}^k}

\newcommand{\exklogic}{\exists^{+}\klogic}

\newcommand{\pexists}{\exists^{+}}

\newtcolorbox{mybox}{
nobeforeafter,
colframe=black,
colback=gray!10, 
boxrule=1.5pt,
left =0pt,
right =0pt,
top = 0pt,
bottom =0pt,
boxsep=0.1mm,
enhanced,
}

 \newcommand{\confORfull}[2]{#2} 

\usepackage{subfiles}

\title{From Quantifier Depth to Quantifier Number: Separating Structures 
with \texorpdfstring{$k$}{k} variables}

\date{}

\author{
  Harry Vinall-Smeeth\\
  Technische Universit\"at Ilmenau, Germany \\
  \texttt{harry.vinall-smeeth@tu-ilmenau.de}
}

\begin{document}

\maketitle

\begin{abstract}
Given two $n$-element structures, $\strucA$ and $\strucB$, which can be distinguished by a sentence of $k$-variable first-order logic ($\klogic$), what is the minimum $f(n)$ such that there is guaranteed to be a sentence $\phi \in \klogic$ with at most $f(n)$ quantifiers, such that $\strucA \models \phi$ but $\strucB \not \models \phi$? We will present various results related to this question obtained by using the recently introduced QVT games \cite{carmosino2023finer}. In particular, we show that when we limit the number of variables, there can be an exponential gap between the quantifier depth and the quantifier number needed to separate two structures. \confORfull{}{Through the lens of this question, we will highlight some difficulties that arise in analysing the QVT game and some techniques which can help to overcome them.} As a consequence, we show that $\mathcal{L}^{k+1}$ is exponentially more succinct than $\klogic$. We also show, in the  setting of the existential-positive fragment, how to lift quantifier depth lower bounds to quantifier number lower bounds. This leads to almost tight bounds.
\end{abstract}

\section{Introduction}

The classic combinatorial game in finite model theory is the Ehrenfeucht–Fraïssé (EF) game \cite{ehrenfeucht1961application, fraisse1954some}. This is played on a pair of structures $[\strucA, \strucB]$ by two players, Spoiler and Duplicator. Spoiler aims to expose the differences between the two structures while Duplicator tries to hide them. The EF game captures the  quantifier depth ($\QD$) needed to separate $\strucA$ and $\strucB$ in the sense that Spoiler can win the game in $r$-rounds iff there is a first-order (FO) sentence $\phi$ with quantifier depth $r$ such that $\strucA \models \phi$ and $\strucB \not \models \phi$. The $\QD$ needed to separate $\strucA$ and $\strucB$ is a measure of how different the structures are. 

Recently another measure, quantifier number ($\QN$), has received substantial attention \cite{carmosino2023finer,fagin2021, fagin2022number}. The original motivation for studying $\QN$ comes from Immerman \cite{immerman1980, immerman1981number}, who connected $\QN$ to complexity theory \confORfull{and}{. The idea is that if we understand how many quantifiers are needed to express certain properties, this could enable us to separate complexity classes. Moreover, Immerman} provided a combinatorial game which captures the $\QN$ needed to separate two sets of structures.\confORfull{}{\footnote{Originally called the \emph{separability game}.}} 

\confORfull{}{The idea of using combinatorial games to separate complexity classes has, so far, failed to live up to its early promise. By far the most studied is the aforementioned EF game and in this domain the outlook doesn't look good \cite{DBLP:journals/eccc/ChenF13, DBLP:conf/birthday/ChenF15}. But \confORfull{a}{studying games concerning $\QN$ is relatively untrodden ground. A}s Immerman put it: `Little is known about how to play the [multi-structural] game.' ... `We urge others to study it, hoping that the
[multi-structural] game may become a viable tool for ascertaining some of the lower bounds which are ``well believed'' but have so far escaped proof.'}

Another prominent object in finite model theory is $\klogic$, the $k$-variable fragment of FO, see \confORfull{\cite{grohe1998finite} for a survey}{\cite{dawar1999finite} and \cite{grohe1998finite} for surveys.} \confORfull{}{This consists of those formulas of FO which use at most $k$ distinct variables.} These logics have some nice properties. For example, while the model checking problem is $\mathsf{PSPACE}$-complete for FO \cite{vardi1982complexity}, it is $\mathsf{PTIME}$-complete for $\klogic$ \cite{immerman1982upper, vardi1995complexity}; given $\phi \in \klogic$ and an $n$-element structure $\strucA$, checking whether $\strucA \models \phi$ can be done in time $O(|\phi| \cdot k \cdot n^k)$\confORfull{}{, by a simple bottom-up algorithm \cite{vardi1995complexity}}.  

This work studies $\QN$ in $\klogic$\confORfull{}{; we should note that this is closely related to formula size, see Section~\ref{ss: formula}}. To be more precise we ask: given two $n$-element structures, $\strucA$ and $\strucB$, that can be distinguished by $\klogic$ what is the minimum $f(n)$ such that there is guaranteed to be a sentence $\phi \in \klogic$ with\confORfull{}{ at most} $f(n)$ quantifiers, such that $\strucA \models \phi$ but $\strucB \not \models \phi$? Simultaneously restricting the number of variables and the number of quantifiers has a natural connection to complexity theory \cite{immerman1980}, we outline this in \confORfull{Appendix~\ref{a: complexity}}{Section~\ref{s: prelims}}. 

The aforementioned question is trivial in FO where we can characterise isomorphism by a sentence with $n$ quantifiers. But if we remain in $\klogic$ and replace the role of $\QN$  with $\QD$ then we get a question which has received substantial attention \cite{berkholz2016, cai1992optimal,furer2001weisfeiler,grohe1999equivalence, grohe2023iteration}. \confORfull{}{Recently an almost tight $n^{\Omega(k)}$ lower bound was obtained by Grohe, Lichter and Neuen in this setting \cite{grohe2023iteration}.} In this line of work pebble games---a variant of the EF game which simultaneously captures $\QD$ and number of variables---are crucial. An analogue to such games also exists in our setting: the quantifier-variable tree (QVT) game, recently introduced in \cite{carmosino2023finer}. To the best of our knowledge, the following is the first lower bound obtained using the QVT game. 

\begin{theorem} \label{thm: genlb}
For every integer $k \ge 3$ there exists $\varepsilon >0$ and $n_0 \in \mathbb{N}$, such that for all $n>n_0$ there exists a pair of $k$-ary structures $\strucA, \strucB$ with $|A| = |B| = n$ that can be distinguished in $\klogic$, but which agree on every sentence of $\klogic$ with less than $2^{\varepsilon n}$ quantifiers.
\end{theorem}

The $\QN$ needed to separate structures is closely related to the \emph{succinctness} of logics. In particular, we are able to leverage the machinery built to prove Theorem~\ref{thm: genlb}, to show an exponential succinctness gap between $\klogic$ and $\mathcal{L}^{k+1}$ answering an open question mentioned in \cite{berkholz2019compiling}.

\begin{theorem} \label{thm: suc}
For every integer $k \ge 3$ there exists a constant $\varepsilon >0$ such that for all $n \in \mathbb{N}$ there is a sentence $\phi \in \mathcal{L}^{k+1}$,  with $|\phi| \ge n$ such that $\phi$ has an equivalent sentence in $\klogic$ but such that any such sentence has size at least $2^{\varepsilon \cdot |\phi|}$.
\end{theorem} 

One motivation for studying the trade-off between number of 
variables and succinctness goes as follows. Since model checking is in $\mathsf{PTIME}$ for $\klogic$ one way to evaluate FO formulas is to first rewrite them to an equivalent formula using fewer variables and then evaluate this formula. Theorem~\ref{thm: suc} shows that we need to be careful applying this strategy: even if we only reduce the number of variables by  one we may get an exponential formula size blow-up.

The rest of the paper is structured as follows. We start in Section~\ref{s: game}, by formally introducing the QVT game. As a first application, we show an upper bound for our main question and also make explicit the connection between $\QN$ and formula size. In Section~\ref{ss: exgame}, we introduce two variations of the game: one  corresponding to the  existential-positive fragment of $\klogic$ and a simplified version which does not capture $\QN$, but which can be used to prove lower bounds. In Section~\ref{s: lb}, we prove Theorems~\ref{thm: genlb} and \ref{thm: suc}. Our main idea is as follows. In the $\QD$ setting studying structures constructed via XOR formulas has proved fruitful \cite{berkholz2016}. We give a generalisation of this framework by considering `formulas' consisting of linear constraints over elements of an abelian group $G$.\confORfull{}{ This allows us to circumvent a surprising technical hurdle that emerges when trying to use structures constructed from XOR formulas to prove lower bounds.} Finally, in Section~\ref{s: exlb}, we tackle a restricted version of our main question and obtain almost tight bounds in this setting: concretely we investigate the existential-positive fragment of $\klogic$, $\exklogic$ and prove the following.

\begin{theorem} \label{thm: exlb}
For every integer $k \ge 4$ there exists $\varepsilon >0$ and $n_0 \in \mathbb{N}$, such that for all $n>n_0$ there exists a pair of $k$-ary structures $\strucA, \strucB$ with $|A| = |B| = n$ 
that can be distinguished in $\exklogic$, but which agree on every sentence of\, $\exklogic$ with less than $2^{\varepsilon n^{2k-4}}$ quantifiers. 
\end{theorem}

We hope that the central idea in the proof of Theorem~\ref{thm: exlb}---that of lifting $\QD$ lowing bounds to $\QN$ lower bounds---can be applied in other settings. More generally, throughout the paper, we develop a range of tools for `taming' the QVT game. Extra details and omitted proofs are given in the appendix.

\confORfull{\paragraph{Related Work}}{\subparagraph{Related Work}}
This work contributes to a line of work which investigates $\QN$ as a complexity measure. This research strand was initiated by Immerman \cite{immerman1980}, who introduced a game characterising the number of quantifiers needed to separate two sets of structures; he called this the \emph{separability} game. \confORfull{R}{However, to the best of our knowledge, no further work was done in this area until r}ecently\confORfull{}{ when} Fagin, Lenchner, Regan and Vyas rediscovered the game \cite{fagin2021}; they named it the \emph{multi-structural} (MS) game. This game differs from the EF game in two key aspects. Firstly, the game is played on \emph{sets} of structures and secondly, Duplicator is given the power to make copies of structures. This has some novel consequences for how the game is played \cite{carmosino2023finer, fagin2021, fagin2022number}\confORfull{.}{; to provide some context we introduce these games in Section~\ref{s: prelims}.} 

The $k$-QVT game is similar to the \emph{formula size game} used by Adler and Immerman \cite{adler2003n} in the context of \emph{temporal logic}. In fact, the QVT game can be viewed as a version of their game where we only `count' moves corresponding to quantifiers; both games are placed within a common framework in \cite{carmosino2023finer}.

The structures used in Section~\ref{s: lb} are based on systems of linear constraints over the group $(\mathbb{Z}_2^k, +, 0)$, generalising a construction of Berkholz and Nordstr{\"o}m \cite{berkholz2016}. It is also possible to generalise their framework in other ways, see for example \cite{DBLP:conf/icalp/BerkholzG15, berkholz2017linear}. The idea of using XOR formulas to construct hard instances has a long\confORfull{}{ and diverse} history. Notably, in the 1980s Immerman discovered a beautiful construction which, given an XOR formula, produces two graphs which cannot easily be distinguished, in the sense that any separating sentence must either have many variables or high quantifier depth \cite{immerman1981number}. This, and related constructions, have led to hardness results \confORfull{for}{in the context of} \emph{counting logics} \cite{cai1992optimal, furer2001weisfeiler,  grohe2023iteration}. 

The use of XOR formulas as a source of hard instances is also a prominent idea in \emph{proof complexity}, see e.g., \cite{ ben2002hard, urquhart1987hard}. In fact, on a conceptual level, our construction in Section~\ref{s: exlb} has similarities to the technique of \emph{XORification}, which emerged from this discipline \cite{ben2002size,ben1999short,  buresh2007complexity}. \confORfull{Moreover, the lower bound game introduced in Section~\ref{ss:  exgame} is, on a high level, similar to the Prover-Delay game from \cite{ben2004near}.}
{The idea is to take a formula which is hard with respect to some complexity measure and replace every variable with an XOR clause to produce a  formula which is very hard with respect to some, potentially different, complexity measure. Our construction is similar in spirit. We start off with two structures which are hard to separate with respect to $\QD$ and we produce two structures which are very hard to separate with respect to $\QN$. Both constructions also make use of XOR gadgets, albeit in a different way. One final connection with proof complexity: the lower bound game introduced in Section~\ref{ss:  exgame} is, on a high level, similar to the Prover-Delayer game from \cite{ben2004near} introduced in the context of separating tree-like and general resolution.}

Theorem~\ref{thm: suc} concerns the \emph{succinctness} of logics. In finite model theory, two logics are often compared by their expressiveness: $\logic_1$ is at least as expressive as $\logic_2$ iff every property expressible in $\logic_2$ is also expressible in $\logic_1$. But expressiveness does not tell the whole story: it might be that a property is expressible in both $\logic_1$ and $\logic_2$ but can be expressed much more succinctly in one than the other. In particular, we say that $\logic_1$ is \emph{exponentially more succinct} than $\logic_2$ if there is a sentence $\phi \in \logic_1$ which is equivalent to a sentence of $\logic_2$ but such that every such sentence has size exponential in $|\phi|$. Results concerning succinctness are surprisingly scattered \cite{DBLP:conf/aiml/HellaV16, DBLP:journals/tocl/EickmeyerEH17, berkholz2019compiling}. Most relevant for us is the result from Grohe and Schweikardt \cite{DBLP:journals/lmcs/GroheS05}, who showed that $\mathcal{L}^4$ is exponentially more succinct than $\mathcal{L}^3$ over linear orders. To the best of our knowledge, it was open even whether FO is exponentially more succinct than $\logic^{k}$ for $k\ge 4$. With Theorem~\ref{thm: suc} we show that $\logic^{k+1}$ is exponentially more succinct than $\klogic$ for $k \ge 4$. We should note that our result holds over the class of all structures and only for a particular choice of signature, which includes higher arity relations. 

Theorem~\ref{thm: exlb} concerns the existential-positive fragment of FO. From a database theoretic perspective, this can be viewed as those relational algebra queries using only Select, Project, Join and Union which is semantically equivalent to the class of unions of conjunctive queries (UCQs). This class of queries is prominent in the literature \cite{carmeli2021enumeration, chen2014complexity, chen2016counting}, as  often queries asked by users in database systems are of this form. \confORfull{}{In this setting, the number of variables used is an important measure of how difficult the query is to evaluate. }

\section{Preliminaries} \label{s: prelims}

\confORfull{\paragraph{(Pebbled) Structures}}{\subparagraph{(Pebbled) Structures}}
 We assume throughout that all structures are finite,
relational and have a finite signature. We also assume knowledge of basic notions and notation from logic, see e.g. \cite{libkin2004elements}.

For a 
structure $\strucA$, we write $A$ for $\dom(\strucA)$. A 
\emph{pebbled} structure $\lstrucAa$ is a structure $\strucA$ 
along a partial assignment $\assA$ of variables to elements of 
$A$. In the context of games, we will also refer to pebbled structures as \emph{boards}. If a variable $x \not \in \dom(\assA)$, it will be notationally convenient to write $\assA(x) := \emptyset$. If $\assA = \emptyset$ we will sometimes write $\strucA$ for $
\lstrucAa$. $\lstrucAa$ \emph{models} a formula $\phi$, if $(\strucA, \assA) \models \phi$, we write $\lstrucAa \models \phi$; $\lstrucAa$ and $\lstrucBb$ \emph{agree} on a formula $\phi$, if either both $\lstrucAa$ and $\lstrucBb$ model $\phi$, or neither do. Similarly two sets of pebbled structures $\classA, \classB$ agree on a formula if every $\lstrucAa \in \classA$ and $\lstrucBb \in \classB$ agree on the formula. If for every $\lstrucAa \in \classA$, $\lstrucAa \models \phi$ we write $\classA \models \phi$.

We say that two assignments are \emph{compatible} if they have the same domain. Two pebbled structures over a common signature $\lstrucAa$ and $\lstrucBb$ are compatible if $\assA$ and $\assB$ are compatible. \confORfull{We extend this to sets of structures in the natural way.}{Finally we say that two collections of pebbled structures $\classA, \classB$ are compatible, if every $\lstrucXa, \lstrucYb \in \classA \cup \classB$ are compatible.} In this work, whenever we simultaneously consider two (sets of) structures we implicitly assume they are compatible.

We say that an element $a \in A$ is pebbled by $x$, a 
variable, if $\assA(x)=a$. If $\dom(\assA) \subseteq \{x_i \, \mid \, i\in [k]\}$ we say that $\lstrucAa$ is $k$-pebbled and that $\assA$ is a $k$-assignment. Given a $k$-pebbled structure $\lstrucAa$ we can think of our structure as coming with a set of pebbles labelled with elements of $[k]$ where pebble $i$ corresponds to variable $x_i$. Reflecting this we will write $\assA(i)$ for $\assA(x_i)$. 

An important construction for us will be that of \emph{moving a pebble}. 
Formally given a pebbled structure $\lstrucAa$ and a pebble $i$, we form a new pebbled structure $
\lstrucAb$ where 
$\assB$ and $\assA$ agree on the image of each variable in 
$\dom(\assA)$ except possibly $v_i$ and $\dom(\assB) = \dom(\assA) 
\cup \{x_i\}$. If $\assB(i)=a$ we write $\assA(i \to a)$ for this 
assignment, and say that $\lstruc{\strucA}{\assA(i \to a)}$ is formed from $\lstrucAa$ by moving pebble $i$ to $a$.

\confORfull{\paragraph{Partial Isomorphism/Homomorphism}}{\subparagraph{Partial Isomorphism/Homomorphism}}
Let $\lstrucAa$ and $\lstrucBb$ be pebbled $\sigma$-structures. Then $\lstrucAa$ and $\lstrucBb$ are partially isomorphic if:
\begin{itemize}
\item for every $i,j \in [k]$, $\assA(i) = \assA(j)$ iff $\assB(i) = \assB(j)$ and
\item for every $m$-ary relation symbol $R \in \sigma$, and every sequence $(i_1, \dots, i_m) \in [k]^m$
\[ R^{\strucA}(\assA(i_1), \dots, \assA(i_m)) \textrm{ iff } R^{\strucB}(\assB(i_1), \dots, \assB(i_m))
\]
\end{itemize}
Note that the map $\sigma$ which maps $\assA(i) \to \assB(i)$ is a partial isomorphism iff $\lstrucAa$ and $\lstrucBb$ are partially isomorphic. We call this $\sigma$ the \emph{canonical partial map} between $\lstrucAa$ and $\lstrucBb$. If the `only if' direction of the above points holds we say that $\lstrucAa$ and $\lstrucBb$ are partially homomorphic and say that $\sigma$ is a partial homomorphism.  Note, that this does \emph{not} imply that $\lstrucBb$ and $\lstrucAa$ are partially homomorphic. 

\confORfull{\paragraph{Finite Variable Logics and Pebble Games}}{\subparagraph{Finite Variable Logics and Pebble Games}}
Fix some integer $k$. Then $\klogic$ is the fragment of 
FO which only uses variables $x_1, \dots, x_k$. We will also be 
interested in the existential-positive fragment which 
we denote by $\exklogic$. To be precise this is the fragment of 
$\klogic$ which does not use the connectives $\forall$ and $\neg$. \confORfull{}{We will also briefly touch on the logic $\exists \klogic$ which expands $\exklogic$ by allowing negation, but only at the atomic level.} If $\phi$ is a formula of FO, $\klogic$ or $\exklogic$ we write $\quant(\phi)$ to denote the number of quantifiers occurring in $\phi$ and $|\phi|$ to denote the total number of symbols in $\phi$.

Given two $k$-pebbled structures $\lstrucAa, \lstrucBb$ 
the \emph{$k$-pebble game} 
starting from position $[\lstrucAa, \lstrucBb]$ is played by two 
players, Spoiler and Duplicator, as follows. If $
\lstrucAa$ and $\lstrucBb$ are not partially isomorphic
Spoiler wins in zero rounds. If Spoiler does not win in $r-1$ rounds, then in round $r$ play starts from the end position of round $r-1$, say $[\lstrucAc, \lstrucBd]
$, where $\assC$ and $\assD$ are $k$-assignments. Then Spoiler 
picks a pebble $i \in [k]$ and either moves $i$ on $\lstrucAc$ or on $
\lstrucBd$. Duplicator responds by moving $i$ on the 
other structure; we get a position $[\lstruc{A}{\assC(i \to 
x)}, \lstruc{B}{\assD(i \to y)}]$. If these two boards are \emph{not} partially isomorphic Spoiler wins in 
$r$ rounds. Otherwise, the game continues in round $r+1$ from this 
position. This game characterises the quantifier depth needed to separate $\lstrucAa$ from $\lstrucBb$ in the following sense.

\begin{theorem}[\cite{immerman1982upper}] \label{thm: pebble}
Spoiler wins the $k$-pebble game from position $[\lstrucAa,\, \lstrucBb]$ in $r$ rounds if and only if there is a formula $\phi \in \klogic$ with quantifier depth $r$ such that $\lstrucAa \models \phi$ and $\lstrucBb \not \models \phi$.
\end{theorem}

\confORfull{}{If we do not place a restriction on the number of pebbles we recover the EF game.} \confORfull{If we stipulate that Spoiler may only play on the left-hand structure and that Spoiler wins at the end of a round whenever the left hand structure is not partially homomorphic to the right hand structure, we get a game characterising quantifier depth in $\exklogic$: the $\pexists k$-pebble game. If instead the restriction on the number of variables is removed we get the EF game.}
{It is easy to see that if we restrict Spoiler to always playing on the left-hand structure in the $k$-pebble game then we get a game characterising quantifier depth in $\exists \klogic$. If we also change the winning condition so that Spoiler wins at the end of a round whenever the left hand structure is not partially homomorphic to the right hand structure we instead get a game characterising quantifier depth in $\exklogic$. See e.g.\! \cite[Chapter 11]{libkin2004elements} for an introduction to pebble games.}

\confORfull{}{\subparagraph{Quantifier Number in \texorpdfstring{$\klogic$}{L(k)} as a Complexity Measure}
For a formula $\phi$, we write $\quant(\phi)$ to denote the number of quantifiers in $\phi$. It is well known that, in $\klogic$, $\QN$ is very closely linked to formula size; we make this connection explicit in Section~\ref{ss: formula}. To motivate the study of $\QN$ in $\klogic$, we now give one example from Immerman's PhD thesis \cite{immerman1980}, connecting this measure to complexity theory. 

Let $\QN^k[\log(n)]$ be those properties $\mathbf{P}$, such that there is a \emph{uniform} sequence $\{\phi_n\}_{n \in \mathbb{N}}$ of $\klogic$-sentences such that:
\begin{enumerate}
\item $\phi_n$ expresses $\mathbf{P}$ on any structure with at most $n$ elements and
\item $\quant(\phi_n) = O(\log(n))$.
\end{enumerate}
Then Immerman showed in \cite{immerman1980} that $\mathsf{NL} \subseteq \bigcup_{k\in \mathbb{N}} \QN^k[\log(n)]$ over all \emph{ordered} finite structures. This gives us a method of separating $\mathsf{NL}$ from, for instance, $\mathsf{NP}$. Take some $\mathsf{NP}$ complete problem; this induces some property $\mathbf{P}$. Then if $\mathbf{P} \not \in \QN^k[\log(n)]$ for every $k$ it follows that $\mathsf{NL} \neq \mathsf{NP}$. 

The question, then, is: how we can prove statements like $\mathbf{P} \not 
\in \QN^k[\log(n)]$? Well we can, in principle, use the $k$-QVT game 
\cite{carmosino2023finer}. The problem is that this game is complicated and 
difficult to analyse. A natural suggestion would be to instead study $\QD$ 
as this provides a lower bound for $\QN$ and, as we have seen, can be 
characterised by a simpler combinatorial game. Unfortunately, this cannot 
work because, as Immerman showed \cite{immerman1980}, we can capture 
isomorphism on ordered finite structures with only logarithmic quantifier 
depth and three variables. It seems, therefore, that we really do need to 
study $\QN$. In this work, we provide---to the best of our knowledge---the 
first lower bounds which use the $k$-QVT games. While we only do this over 
\emph{unordered} structures we believe that the techniques we develop go 
some way to taming this game.

Finally, we should note that the above example is not the only connection 
between $\QN$ in $\klogic$ to complexity theory. Immerman's thesis also 
showed, more generally, that this measure is closely related to space 
complexity and that there is a close connection to computations on 
\emph{alternating Turing machines}.}

\confORfull{\paragraph{The Multi-Structural Game}}{\subparagraph{The Multi-Structural Game}} 
We now introduce the MS game, as a stepping stone to the QVT game.\confORfull{}{ This game is similar to the EF game, except we play on two sets of structures and Duplicator has the ability to make copies of structures.} It is played in rounds by two players, Spoiler and Duplicator, 
starting from position $[\classA_0, \classB_0]$, where $\classA_0, \classB_0$ are sets of pebbled structures. Initially, if no $\strucA 
\in \classA_0$ is partially isomorphic to any $\strucB \in \classB_0$ we 
say that Spoiler wins after 0 rounds. Otherwise, for $r>0$, round $r$ starts from position $[\classA_{r-1}, \classB_{r-1}]$. Then Spoiler 
chooses to either play on $\classA_{r-1}$ or $\classB_{r-1}$. If they play 
on $\classA_{r-1}$, then they take a fresh pebble $i$ and place it on an element of $\lstrucAa$ for each $\lstrucAa \in \classA_{r-1}$, 
to form $\classA_{r}$. Duplicator then 
replies by making as many copies as they would like of structures in $
\classB_{r-1}$ and placing $i$ on an element of each of these structures 
to form $\classB_{r}$. Spoiler wins 
after $r$ rounds if no structure in $\classA_{r}$ is partially isomorphic to a structure in $\classB_{r}$. If Spoiler chooses to play on $\classB_{r-1}$ play proceeds in the same way but with the roles of $\classA_{r-1}$ and $\classB_{r-1}$ swapped.

The following example, which is given in \cite{fagin2021}, nicely illustrates \confORfull{how Duplicator's extra power can help them survive for longer.}{how the extra power given to Duplicator makes it harder for Spoiler to win.}

\begin{example} \label{ex: 3vs2}
Let $\strucA$ be a linear order of size three and $\strucB$ a linear order of size two. Formally, we have a single binary relation $<$ and \confORfull{$\strucA := \{\{a_1, a_2, a_3\}, \{(a_1, a_2), (a_1, a_3), (a_2, a_3)\}\}$, $\strucB := \{\{b_1, b_2\}, \{(b_1, b_2)\}\}$.}{
\begin{align*}
\strucA :=& \{\{a_1, a_2, a_3\}, \{(a_1, a_2), (a_1, a_3), (a_2, a_3)\}\} \\
\strucB :=& \{\{b_1, b_2\}, \{(b_1, b_2)\}\}
\end{align*}} 
Then\confORfull{}{ it is easy to see that} Spoiler can win the EF game in two rounds: they first play $a_2$ and Duplicator must reply with $b_1$ or $b_2$. If they play $b_1$, then in the next round Spoiler plays $a_1$ and Duplicator has no reply as there is no $b\in B$ with $b< b_1$. Similarly, if Duplicator plays $b_2$ then Spoiler  wins in the next round by playing $a_3$. \confORfull{Now suppose that in the MS game Spoiler plays $a_2$ in the first round. The difference is that now Duplicator can make two copies of $\strucB$ and play $b_0$ on one board and $b_1$ on the other, see Figure~\ref{fig: 3vs2}. Now Spoiler \emph{cannot} win in the next round. For example, if they play $a_1$ then Duplicator will survive on the board they played $b_2$ on in the first round. In fact, one can show that Duplicator requires three rounds to win the MS game on $\strucA, \strucB$.
}{What is the difference in the MS game? Well here suppose again that Spoiler plays $a_2$ in the first round. The difference is that Duplicator can make two copies of $\strucB$ and play $b_0$ on one board and $b_1$ on the other, see Figure~\ref{fig: 3vs2}. Now Spoiler \emph{cannot} win in the next round. For example, if they play $a_1$ then Duplicator will survive on the board they played $b_2$ on in the first round. In fact, one can show that Spoiler requires three rounds to win the MS game on $\strucA, \strucB$ \cite{fagin2021}.}
\end{example}
\confORfull{
\begin{figure} 
\centering
\includegraphics[width=\linewidth]{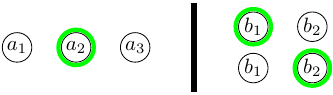}
\caption{Position of the MS game on $\strucA, \strucB$ after one round.}
\Description[The position of the MS game on a linear order of length 3 versus a linear order of size two after one round.]{The position of the MS game on a linear order of length 3 versus a linear order of size two after one round. On the left the middle element of the linear order of size 3 is circled in green to represent the fact that this is pebbled in the first round. On the right there are two copies of the linear order of size two: one with the first element circled in green and the other with the last element circled in green.}
\label{fig: 3vs2}
\end{figure}
}
{\begin{figure} 
\centering
\includegraphics[width=0.6\textwidth]{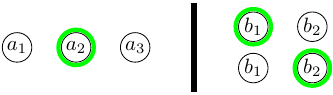}
\caption{Position of the MS game on $\strucA, \strucB$ after one round, the green circle indicates the pebbled element on each board.}
\label{fig: 3vs2}
\end{figure}
}
\section{The \texorpdfstring{$k$}{k}-QVT game} \label{s: game}

In this section, we introduce the $k$-QVT game. \confORfull{We then use this game to give an easy upper bound for our main question and observe a connection to formula size.}
{We give a different formulation of the game to \cite{carmosino2023finer} and so our first task is to prove that the two formulations are equivalent. Afterwards, we give an easy upper bound for our main question; the idea here is that we can utilise the Spoiler winning strategy in the $k$-pebble game. Finally, we observe a connection to formula size.}

We next give an intuitive description of the game. A naive attempt to define a game capturing $\QN$ in $\klogic$ would be to amend the MS game such that each player has only $k$ pebbles available to them, call this the $k$-MS game. This approach is exactly how we generalise EF games to account for the number of variables. However, it is easy to see this is insufficient in our case.  It is shown in \cite{carmosino2023finer} that $\strucA$ and $\strucB$ from Example~\ref{ex: 3vs2} can be separated by a FO sentence with $2$ variables and $3$ quantifiers and yet Duplicator has a 3-round winning strategy for the $2$-MS game on these structures.

So what do we need to add? The problem arises because when we reuse variables we can no longer assume that our formulas are in PNF, which means we need to think about how to treat conjuncts and disjuncts. \confORfull{For instance the the separating sentence in the above counterexample is}{To see this let us look at the separating sentence in the above counterexample:}
\begin{equation} \label{eq: 2-var}
\exists v_1  \left(\exists v_2 \,(\,v_1 < v_2\,) \wedge \exists v_2 \,(\,v_2 < v_1\,) \right).
\end{equation}
\confORfull{}{If we did not need to worry about re-using variables we could re-write this as the equivalent sentence:
\begin{equation}
\exists v_1  \exists v_2  \exists v_3 \,(\, v_1 < v_2 \wedge v_3 < v_1 \, ).
\end{equation}
 But this sentence has three variables which is problematic. To deal with conjuncts and disjuncts we allow Spoiler to `split' the game into subgames. For example suppose we have some position $[\{\lstrucAa\}, \{\lstrucBb, \lstrucBc\}]$ then it may be that $\lstrucAa \models \phi_1 \wedge \phi_2$, $\lstrucBb \models \phi_1 \wedge \neg \phi_2$, $\lstrucBc \models \neg \phi_1 \wedge \phi_2$. In such a case we allow Spoiler to perform a split of the game into two new games, one in position $[\{\lstrucAa\}, \{\lstrucBb\}]$ and one in position $[\{\lstrucAa\}, \{\lstrucBc\}]$. This means the game is actually played on a tree where each node corresponds to some subgame.} \confORfull{To deal with conjuncts and disjuncts we allow Spoiler to `split' the game into subgames. We next formalise this.}{It turns out that this is exactly the ingredient we need; we next formalise this.}

\subsection{Game Description}

Let $\classA$, $\classB$ be sets of $k$-pebbled structures. The $k$-QVT game starting from position $[\classA, \classB ]$ is played by two players, Spoiler and Duplicator, and is initialised by a triple $(\mathcal{T}^1, \chi^1, c^1)$ where: $\mathcal{T}^1$ is a tree consisting of a single node, $t_0$; $\chi^1$ is a labelling function which labels $t_0$ by $[\classA,\classB]$ and $c^1 = t_0$. Before the first round, we delete every $\lstrucAa 
\in \classA$ which is not partially 
isomorphic to any $\lstrucBb \in \classB$. Similarly, we delete every $\lstrucBb 
\in \classB$ not partially 
isomorphic to any $\lstrucAa \in \classA$. If after this process $\chi^1(t_0) = [\emptyset, \emptyset]$ we say that Spoiler wins in $0$ rounds.  For $r>0$, if Spoiler does not win in round $r-1$ then round $r$ is played from the position reached at the end of the previous round, i.e., a triple $(\mathcal{T}^r,\, \chi^r,\, c^r)$ \confORfull{where: $\mathcal{T}^r$ is a tree rooted at $t_0$; $\chi^r$ is a labelling function which assigns each $t\in V(\mathcal{T}^r)$ a pair of sets of $k$-pebbled structures and
$c^r \in V(\mathcal{T}^r)$, we call this the \emph{current node}.}
{where: \begin{enumerate}
\item $\mathcal{T}^r$ is a tree rooted at $t_0$;
\item $\chi^r$ is a labelling function which assigns each $t\in T$ a pair of sets of $k$-pebbled structures, and
\item $c^r \in V(\mathcal{T}^r)$, we call this the \emph{current node}.
\end{enumerate}}    

Before continuing let us introduce some terminology. If $\chi^r(t) = [\classC, \classD]$ then we say that $\classC$ is the left-hand side ($\LHS$) and $\classD$ is the right-hand side ($\RHS$). If $\lstrucXa \in \LHS$ and $\lstrucYb \in \RHS$ we say that $\lstrucXa$ is on the other side to $\lstrucYb$ and, vice-versa. By a slight abuse of notation, if $\chi^r(t) = 
[\classC, \classD]$ we will write $\lstrucXa \in \chi^r(t)$, 
to mean $\lstrucXa \in \classC \cup \classD$. Finally, if $t$ is not a leaf or if $\chi^r(t) = [\emptyset, \emptyset]$ we say that $t$ is $\emph{closed}$, otherwise we say $t$ is \emph{open}. 

At the start of each round Spoiler may choose whether to \emph{split} or \emph{pebble}. In both cases `the action' occurs at $c^r$; suppose this is labelled by $[\classC, \classD]$. 

A split move works as follows. Spoiler chooses $\classX \in \{\classC,\classD\}$ and a sequence of disjoint sets 
$\classX_1, \dots, \classX_j$, such that $\classX = 
\bigcup_{\ell=1}^j \classX_{\ell}$. We then edit $\mathcal{T}^r$ by adding $j$ children of $t$:
$t_1, \dots, t_j$ and extend $\chi^r$ to these new nodes by setting $\chi^r(t_{\ell}) 
= (\classC_{\ell}, \classD)$, if $\classX =\classC$ and $
\chi^r(t_{\ell})=(\classC, \classD_{\ell})$ otherwise. The current node is then set to be $t_1$. After this round $r$ does not end, rather Spoiler can again choose to split or pebble starting from the amended position $(\mathcal{T}^r, \chi^r, c^r)$.

A pebble move is the same as a move in the MS game from position $[\classC, \classD]$ except now Spoiler cannot always choose a fresh pebble: they must move a pebble $i\in[k]$. 
This results in a fresh pair $[\hat{\classC}, \hat{\classD}]$.
Then every $\lstrucXa \in \hat{\classC} \cup \hat{\classD}$ which is not partially isomorphic to a structure on the other side is deleted. We set $\mathcal{T}^{r+1} = \mathcal{T}^r$, $\chi_{r+1}(u) = \chi(u)$ for $u\neq c^r$ and $\chi^{r+1}(c^r) = [\hat{\classC},\hat{\classD}]$. If $\chi(c^r) \neq [\emptyset, \emptyset]$, we set $c^{r+1} = c^r$ and end the round. If $\chi(c^r) = [\emptyset, \emptyset]$ and an open node exists we choose one arbitrarily and set it as the current node. Otherwise, Spoiler wins in $r$ rounds. 

\begin{example}[Linear order of size 3 vs 2] \label{ex: 3vs2tree}
We now show how Spoiler can win the $2$-QVT game in three rounds starting from position $[\{\strucA\}, \{\strucB\}]$, where $\strucA$, $\strucB$ are the structures from Example~\ref{ex: 3vs2}.

In the first round, Spoiler makes a pebble move: they move pebble one to $a_2$. Duplicator then makes both possible responses. Next Spoiler makes a split move. This closes the root of the tree and creates two new children $t_1$ and $t_2$, one for each $\RHS$ board. Suppose $t_1$ contains the board where Duplicator replied with $b_1$. Then Spoiler can close $t_1$ in a single round by moving pebble two to $a_1$. Similarly, Spoiler can close $t_2$ by moving pebble two to $a_3$. Therefore, Spoiler wins in 3 rounds. Note that the above strategy mimics the structure of Equation~(\ref{eq: 2-var}).  
\end{example}

\confORfull
{
}{
\begin{figure} 
\centering
\includegraphics[width=\textwidth]{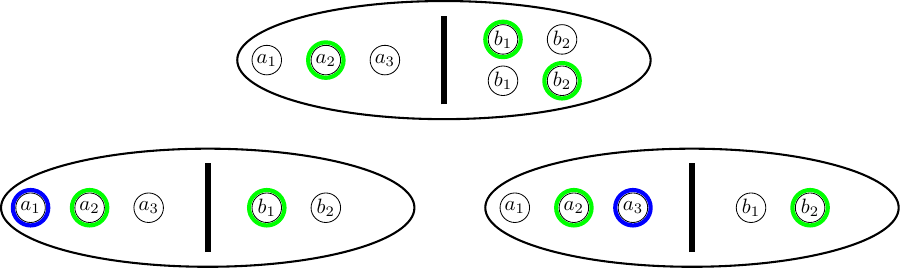}
\caption{The 3-round Spoiler win described in Example~\ref{ex: 3vs2tree}.}
\label{fig: 3vs2_QVT}
\end{figure}
}

\confORfull{}{\subsection{Remarks on the Game}}

Duplicator always has the option to play every possible move when it is their turn; call this the obvious strategy. Clearly, Duplicator can survive $r$ rounds iff they can survive $r$ rounds by playing the obvious strategy\confORfull{}{, see \cite{carmosino2023finer} for a formal proof of this}. We could, therefore, have made our game single player, however for ease of intuition and analysis we retain a second player. If Duplicator moves $p \to x$ on some board $\lstrucAa$ and $\mathcal{A}^{\assA(p\to x)}$ is not deleted at the end of the round we say that $x$ is a \emph{valid} response. It will be convenient at times to assume Spoiler only makes valid responses. \confORfull{Our version of the game is slightly different to that from \cite{carmosino2023finer}; see Appendix~\ref{a: game} for a comparison. The next theorem establishes the correctness of the game and therefore that it is equivalent to the formulation in \cite{carmosino2023finer}.}
{

Before moving on to prove the correctness of our game we should compare our presentation to that given in \cite{carmosino2023finer}. There are a few key differences. Firstly, in their game structures are not deleted at the end of each round. Instead at any time Spoiler may \emph{close} a node by providing an atomic formula which separates the two sides. If every node of the tree is closed Spoiler then wins the game. This may seem like a major difference between the two games but in fact it is not. If for example at the end of some round, a node is labelled by $[\classA, \classB]$ and some $\classB_1 \subseteq \classB$ would be deleted in our version of the game (i.e. because structures in $\classB_1$ are not partially isomorphic to any structures in $\classA$), then in the original formulation of the game Spoiler can perform a split creating two new nodes, $t_1, t_2$ labelled by $(\classA, \classB_1)$ and $(\classA, \classB_2 \setminus \classB_1)$ respectively. Then $\classA, \classB_1$ disagree on a Boolean combination of atomic formulas and so by performing further appropriate splits Duplicator can ensure that all leaves in the subtree rooted at $t_2$ are closed without playing any further pebble moves. Secondly, their game is not symmetric and they introduce so-called `swap' moves which swap the roles of the left- and right-hand side structures. Spoiler can then only play on the $\LHS$. 

Overall, it is slightly easier to prove Theorem~\ref{thm:eqiv} for the version of the game presented in \cite{carmosino2023finer}. Moreover, this version of the game is effectively a version of the formula-size games introduced from \cite{adler2003n} and so fit into a broader framework of games, again see \cite{carmosino2023finer} for details. In our view, the version of the game presented here is easier to work with for proving lower bounds since it strips back all elements which are not directly relevant to $\QN$. Nevertheless, playing our version of the game does yield some insight into the size of a formula needed to separate two structures, as we will see at the end of this section.}

\confORfull{}{\subsection{Proof of Correctness}}

\begin{theorem} \label{thm:eqiv} 
Let $\classA, \classB$ be sets of $k$-pebbled structures. Then Spoiler can win the $k$-QVT game from position $[\classA, \classB]$ in $r$ rounds if and only if there exists a $\klogic$-formula $\phi$ that $\classA$ and $\classB$ disagree on with $\quant(\phi) = r$.
\end{theorem}

 \begin{proof}
We assume throughout w.l.o.g., that Duplicator plays the obvious strategy. First suppose that $\classA$ and $\classB$ disagree on a formula $\phi$, which we may assume uses only the connectives $\neg,\, \wedge,\, \exists$.  We show by induction on the length of $\phi$, that Spoiler has a $\quant(\phi)$-round winning strategy in the $k$-QVT game from position $[\classA, \classB]$. If $\phi$ is an atomic formula that $\classA$ and $
 \classB$ disagree on, then trivially Spoiler wins after $0$ 
 rounds. For the induction step, suppose the equivalence is 
 true for every formula of length at most $m$ and that $\phi$ is a non-atomic $\klogic$-formula of length $m+1$. 

 First, say $\phi \equiv \exists x_i \, \theta$, with $\quant(\theta) = r-1$. Spoiler's
first move is to select pebble $i$ and for each $\lstrucAa \in \classA$, place the pebble on an element $a$ such that $\lstruc{\strucA}{\assA(i \to a)} \models \theta$, noting that this is always possible since $\lstrucAa \models \exists x_i \, \theta$. Duplicator can then
duplicate structures and move pebble $i$ on each of these copies. Then every pebbled structure which is not partially isomorphic to any structure on the other side is deleted. Denote the resulting sets of structures by $\hat{\classA}$ and $\hat{\classB}$. Since $\classB \models \neg\phi$, pebble $i$ is not on a witness for $\theta$ for any $\lstrucBb \in \hat{\classB}$. Therefore $\hat{\classA}$ and $\hat{\classB}$ disagree on the formula $\theta$. As $\theta$ is of length $m$ and has $r-1$ quantifiers Spoiler has a $(r-1)$-round winning
strategy from this position and, therefore, a $r$-round winning strategy overall. 

Second, suppose $\phi \equiv \theta_1 \wedge \theta_2$, $\quant(\theta_j) =k_j$. Let $\classB_1 := \{
\lstrucBb \in \classB \, \mid \, \lstrucBb\not \models \theta_1 \}$
and $\classB_2 := \classB \setminus \classB_1.$  Then Spoiler performs a split, creating two new nodes $t_1, t_2$ such that $\chi(t_i)=
[\classA, \classB_i]$. By the
induction hypothesis, Spoiler has a $k_i$-round winning strategy for the game starting from position $\chi(t_i)$ and so a $k_1 + k_2 = \quant(\phi)$-round winning strategy overall.

 Finally, suppose $\phi \equiv \neg \theta$. Since $\classB \models \theta$ and $\classA \models \neg\theta$, by the induction hypothesis Spoiler has a $r$ round winning strategy from position $[\classB , \classA]$, where $r =\quant(\theta)$. By playing this exact strategy in the $[\classA, \classB]$ game Spoiler also wins this game in $r$ rounds.

For the other direction we induct on $r$, the number of rounds Spoiler needs to win the game. First, for a pebbled structure $\lstrucXa$ and $\classX$  a class of compatible pebbled structures we define, 
\begin{align*}
 T_{\lstrucXa} &:= \{ \phi \in \klogic \, \mid \, \lstrucXa \models \phi \textit{, $\phi$ is atomic or $\phi \equiv \neg \theta$, $\theta$ is atomic} \} \\
 \phi_{\lstrucXa} &\equiv \bigwedge_{\psi \in T_{\lstrucXa}} \psi \\
 T_{\classX} &:= \{ \phi_{\lstrucXa} \, \mid \, \lstrucXa \in \classX\} \\
\phi_{\classX} &\equiv \bigvee_{\psi \in T_{\classX}} \psi 
\end{align*}
Note that even if $\classX$ is infinite, the set $T_{\classX}$ is finite and so $\phi_{\classX} \in \klogic$. Clearly $\lstrucXa \models \phi_{\classX}$ for every $\lstrucXa \in \classX$. 
Moreover, if $\lstrucYb$ is a structure which is not partially isomorphic to any $
\lstrucXa \in \classX$ then $\lstrucYb \not \models \phi_{\classX}$.
Given this, for the base case, we set $\phi = \phi_{\classA}$.

 Now suppose that whenever Spoiler wins in $r$-rounds from a position $[\classX, \classY]$ there is a formula $\phi \in \klogic$, with $\quant(\phi)=r$, which $\classX$
and $\classY$ disagree on. Suppose Spoiler wins the ($r+1$)-round $k$-QVT
game from position $[\classA, \classB]$ and that in the first round Spoiler performs a split of
  $\classB$ into $(\classB_i)_{i \in [\ell]}$. Say Spoiler has a winning strategy in at
  most $r$ moves from position $[\classA, \classB_i]$ for every $i$. Then, by
  the induction hypothesis, there are formulas $\phi_1, \dots, \phi_{\ell}$ such
  that $\classA$ and $\classB_i$ disagree on $\phi_i$ and $\quant(\phi_i)$ is equal to
  the number of rounds in the Spoiler winning strategy from position $[\classA, \classB_i]$. Then 
  \[ \phi \equiv \neg \phi_{\classB} \vee \left(\phi_{\classA} \wedge \bigwedge_{i=1}^j \phi_i \right) 
 \]
 gives us what is required. To see this first let $\lstrucAa \in \classA$. If 
 $\lstrucAa$ is deleted before the split, then $\lstrucAa$ is not partially 
 isomorphic to any $\lstrucBb \in \classB$. Therefore $\lstrucAa \models \neg 
 \phi_{\classB}$. Otherwise $\lstrucAa$ is not deleted and so $\lstrucAa 
 \models \phi_i$ for each $i$ by the induction hypothesis. Either way $
 \lstrucAa \models \phi$. On the other hand, if $\lstrucBb \in \classB$ then 
 clearly $\lstrucBb \not \models \neg \phi_{\lstrucBb}$. If $\lstrucBb$ is 
 deleted before the split then $\lstrucBb \not \models \phi_{\classA}$, 
 otherwise $\lstrucBb \in \classB_i$ for some $i$ and so $\lstrucBb \not 
 \models \phi_i$, by assumption. So $\classA$ and $\classB$ really do disagree 
 on $\phi$. Finally, note that $\quant(\phi) = \sum_{i=1}^{\ell} 
 \quant(\phi_i) = r+1$.
  
   Otherwise, there is an $i$ such that Spoiler needs $r+1$
  rounds to win from position $[\classA, \classB_i]$. But then all other sub-games are winnable in $0$
  rounds since the game from position $[\classA, \classB]$ is winnable in $r+1$ rounds by assumption. So before the game starts we should have deleted $\classB \setminus \classB_i$ from $\classB$, a contradiction. Furthermore, if Spoiler performs a split of $\classA$ instead of $\classB$ we can construct $\phi$ exactly as above, except with the roles of $\classA$ and $\classB$ reversed, and negate the resulting formula.

Now suppose Spoiler does not perform any splits in the first round and that in the playing stage they choose $\classA$ and move pebble $i$. Then Spoiler has a $r$-move winning strategy from
position $\chi^2(t_0)= [\hat{\classA}, \hat{\classB}]$. Therefore, by the induction hypothesis, there is some $\theta \in \klogic$ with $\quant(\theta)=r$ which $\hat{\classA}$ and $\hat{\classB}$ disagree on. Let $\classB'$ be the class of structures derived from $\classB$ that occurs \emph{after} Spoiler and Duplicator have moved, but \emph{before} any structures are deleted. Let 
\[ 
\phi \equiv \neg \phi_{\classB} \vee \left(\phi_{\classA} \wedge  \exists x_i \left(\neg \phi_{\classB'} \vee (\phi_{\hat{\classA}} \wedge \theta )\right)\right). 
\]
Clearly $\quant(\phi) = r+1$. We will show that $\classA, \classB$ disagree on this formula. 
By the same argument as in the base case, we see that any structure in $\classA$ which is 
deleted before Spoiler moves models $\phi$. So suppose $\lstrucAa \in \classA$ is not deleted 
initially. Then Spoiler moves pebble $i \to a$ on $\lstrucAa$, so we get some structure $\lstruc{\strucA}
{\assA(i\to a)}$. It may be that this structure is deleted before the end of 
the turn. In this case, it is not partially isomorphic to any $\lstrucBb \in \classB'$. Therefore $
\lstruc{\strucA}{\assA(i\to a)} \models \neg \phi_{\classB'}$. But then, taking $a$ as the 
witness for the existential quantification, it follows that $\lstrucAa \models \exists x_i \neg 
\phi_{\classB'}$, so $\lstrucAa \models \phi$. But if $
\lstruc{\strucA}{\assA(i\to a)}$ is not deleted, then it is in $\hat{\classA}$ and so by the 
induction hypothesis models $\theta$. Therefore, again taking $a$ as the witness to the existential 
quantification we get that $\lstrucAa \models \phi$. 

We now claim that no $\lstrucBb \in \classB$ models $ \phi$. Again the case where $\lstrucBb$ is deleted before the turn starts is similar to the base case. Now suppose $\lstrucBb$ is not deleted. Then Duplicator plays every possible move on $\lstrucBb$. Suppose for a contradiction that
\[
 \lstrucBb \models \exists x_i \left(\neg \phi_{\classB'} \vee (\phi_{\hat{\classA}} \wedge \theta)\right),
 \]
then there is some $b \in B$ such that $\lstruc{\strucB}{\assB(i \to b)} \models \neg \phi_{\classB'} \vee (\phi_{\hat{\classA}} \wedge \theta)$. Since $\lstruc{\strucB}{\assB(i \to b)}\in \classB'$, it follows that $\lstruc{\strucB}{\assB(i \to b)} \models \phi_{\hat{\classA}} \wedge \theta$. Suppose $\lstruc{\strucB}{\assB(i \to b)}$ is deleted before the end of the round. Then $\lstruc{\strucB}{\assB(i \to b)} \not \models \phi_{\hat{\classA}}$. On the other hand, if $\lstruc{\strucB}{\assB(i \to b)}$ is not deleted, then $\lstruc{\strucB}{\assB(i \to b)} \in \hat{\classB}$ and so by the induction hypothesis does not model $\theta$, a contradiction. The claim follows.

Finally suppose that in the first round Spoiler makes no splits but this time selects $\classB$. By the symmetry of the game, we have that Spoiler wins in $(r+1)$ rounds from position $[\classB, \classA]$. Therefore there is a formula $\psi$ with $\quant(\psi) = r+1$ such that $\classB \models \psi$ but $\classA \not \models \psi$. Then taking $\phi \equiv \neg \psi$ we get the desired formula. 
\end{proof}
 
\subsection{Upper Bounds}
\confORfull{C}{We want to deploy this game to help us with our main question. To this end c}onsider the $k$-QVT game starting from position $[\{\strucA\},\{\strucB\}]$, where $\strucA$ and $\strucB$ are $n$-element structures. Then by Theorem~\ref{thm:eqiv}, Spoiler can win the game in $r$ rounds iff there is a sentence with $r$ quantifiers which $\strucA$ and $\strucB$ disagree on. Immediately we can see the following.

\begin{lemma} \label{lem: ub}
 If two $n$-element structures can be distinguished by a sentence of
 $\mathcal{L}^k$ with quantifier rank $r$, then there is a sentence $\phi$ which the structures disagree on with 
 \[ 
  \quant(\phi) \le \frac{n^{r}-1}{n-1}.
 \] 
\end{lemma}
\begin{proof}
Let $\strucA$ and $\strucB$ be two $n$-element structures which disagree on some sentence with quantifier rank $r$. Then by Theorem~\ref{thm: pebble}, Spoiler has an $r$-round winning strategy in the $k$-pebble game from position $[\strucA, \strucB]$; we use this to construct a winning strategy for the $k$-QVT game from position $[\{\strucA\}, \{\strucB\}]$.

This works as follows. In round $s$, Spoiler only pebbles if $\chi(c^s)$ is of the form $[\{\lstrucAa\}, \{\lstrucBb\}]$. Otherwise, they split a set that didn't consist of a single structure into subsets each of which only contains one structure. Then whenever they pebble, they are playing on a node with a label of the form $[\{\lstrucAa\}, \{\lstrucBb\}]$. They can therefore play as in an optimal strategy of the $k$-pebble game starting from position $[\lstrucAa, \lstrucBb]$. After such a move one side consists of a single structure and the other of at most $n$ structures, since Duplicator has $n$ possible responses to any Spoiler move. In the next round the larger class will then be split into singletons and the current node will gain at most $n$ children. Since Spoiler wins the $k$-pebble game in at most $r$ rounds the height of the final tree is bounded by $r-1$. Therefore, in the worst case for Spoiler we get a complete $n$-ary tree of height $r-1$ where exactly one pebble move occurs at each node. By Theorem~\ref{thm:eqiv} the result follows.
\end{proof}

The advantage of Lemma~\ref{lem: ub} is that we can turn any quantifier depth upper bound into a quantifier number upper bound; we will see a similar phenomenon in the context of lower bounds in Section~\ref{s: exlb}. By combining Lemma~\ref{lem: ub} with the $n^{k-1}$ upper bound from \cite{immerman1990describing} we obtain the following concrete bound.\confORfull{}{\footnote{We should note here that the recent upper bound from \cite{grohe2023iteration} does not apply to our case, as this upper bound is for $k$-variable logic \emph{with counting quantifiers}.}} 

\begin{corollary} \label{cor: ub}
If two $n$-element structures can be distinguished by 
 $\mathcal{L}^k$, then there is a sentence $\phi$ which the structures disagree on with 
 \[ 
  \quant(\phi) \le \frac{n^{n^{k-1}}-1}{n-1}.
 \] 
\end{corollary}

\confORfull{}{Since the $k$-QVT game concerns sets of structures we now extend Lemma~\ref{lem: ub} to this setting.

\begin{lemma} \label{lem: ubset}
 Let $\sigma$ be a signature and let $\classA, \classB$ be two sets of compatible $k$-pebbled structures over $\sigma$ which disagree on a formula of $\mathcal{L}^k$. Moreover, suppose  each $\strucX \in \classA \cup \classB$ has at most $n$ elements. Then there is a formula $\phi$ which $\classA$ and $\classB$ disagree on with 
 \[ 
  \quant(\phi) = n^{O(n^{k-1})}.
 \] 
\end{lemma} }
\confORfull{}{
\begin{proof}
We may assume w.l.o.g.\! that each relation in $\sigma$ has arity at most $k$ and that any structure $\strucX \in \strucA \cup \strucB$ has domain $[|X|]$. Then for a structure with domain size $t$ there are at most $2^{t^k}$ 
possibilities for the interpretation of each relation in $\sigma$ and $
(t+1)^k$ possibilities for the positions of the pebbles. Therefore, by a simple calculation, there are $2^{O(n^{k})}$ elements of $\classA \times \classB$. Note here that the constant hidden in the $O$ notation depends on $\sigma$.

Now we give a Spoiler winning strategy for the $k$-QVT game starting from position $[\classA, \classB]$. First Spoiler performs a split where they create one node labelled by $[\classA, \{\strucB\}]$ for each element of $\strucB \in \classB$. They then perform a further split of each of these nodes such that the game tree has exactly $|\classA \times \classB|$ leaves, one  labelled by every $[\{\strucA\}, \{\strucB\}]$ with $(\strucA, \strucB) \in \classA \times \classB$. 

At each of these leaves, we apply Corollary~\ref{cor: ub}. This leads to Spoiler winning in:
\[
|\classA \times \classB| \cdot \frac{n^{n^k-1}-1}{n-1} 
\]
rounds so from the bound on $|\classA \times \classB|$ the result follows.
\end{proof}}

\subsection{Formula Size} \label{ss: formula}
Our game also provides bounds on the size of the formulas separating two  sets of $k$-pebbled structures $\classA$ and $\classB$\confORfull{, see Appendix~\ref{a: game} for details}{}. It is well known that there is a tight connection between formula size and $\QN$; for instance, this is (at least) implicit in the work of Immerman in the early 1980s \cite{immerman1980, immerman1981number, immerman1982upper}.  

\confORfull{}{Suppose $\classA$, $\classB$ disagree on a sentence of $\klogic$. Consider the formula $\phi_{\lstrucAa}$ introduced in the proof of Theorem~\ref{thm:eqiv}, $\lstrucAa \in \classA \cup \classB$. Each arity $m>0$ relation $R$ contributes at most $k^m$ conjuncts and 
there are at most $k^2$ equalities coming from elements labelled by $\assA$.  Since for each atomic formula either it or its negation is a conjunct in $\phi_{\lstrucAa}$, we get that the size of $\phi_{\lstrucAa}$ is $O(k^k)$ and that there are $2^{O(k^k)}$ distinct possibilities for $\phi_{\lstrucAa}$. Therefore, $\phi_{\classX}$ has size $2^{O(k^k)}$ and this observation combined with the proof of Theorem~\ref{thm:eqiv} yields the following.}

\begin{lemma} \label{lem: sizesep}
Let $\sigma$ be a signature and let $\classA, \classB$ be sets of  $k$-labelled $\sigma$-structures, which disagree on a sentence with $r$ quantifiers. Then there exists some $\phi \in \klogic$ which $\classA$ and $\classB$ disagree on with $|\phi|= 2^{O(k^k)} \cdot r$.
\end{lemma}

The moral of the story is that if we wish to study formula size in $\klogic$ it suffices to study $\QN$. \confORfull{}{This is good news because the $k$-QVT game is simpler than the formula size game of \cite{adler2003n}.}

\section{Variations of the Game} \label{ss: exgame}

In this section, we introduce new games which will be used to obtain our results in Sections~\ref{s: lb} and \ref{s: exlb}. Firstly, we define a simplification of the $k$-QVT game which still allows us to prove lower bounds on the number of quantifiers; we call this the $k$-LB game. Secondly, we provide a game corresponding to $\QN$ in $\exklogic$. Finally, we give an analogue to the $k$-LB game for $\exklogic$.

\subsection{Lower Bound Game} \label{ss: lbgame}
\confORfull{}{We begin by motivating the $k$-LB game. In this game Duplicator collects points and the number of points they collect gives a lower bound on the number of rounds they can survive in 
the $k$-QVT game, hence the name.}

To see the idea behind the\confORfull{ $k$-LB}{} game imagine the following situation in the $k$-QVT game. Spoiler first performs a split and creates two new nodes, $t_1$ and $t_2$. Duplicator sees they can survive roughly the same number of rounds from position $\chi(t_1)$ as from position $\chi(t_2)$. So\confORfull{}{, to save themselves some time,} they offer Spoiler a deal. They will let Spoiler choose whether they want to continue from position $\chi(t_1)$ or $\chi(t_2)$ \confORfull{}{and then the two players will play out the game from this position}. In exchange, each time Spoiler makes a pebble move the round counter will increase by two. This is a fair deal for Spoiler. If they can win in $r_1$ rounds from position $\chi(t_1)$ and in $r_2$ rounds from position $\chi(t_2)$, they can win this modified game in $2\min(r_1 +r_2) \le r_1 +r_2$ rounds. 

Formally, the game is
played on two evolving sets of $k$-pebbled structures $\classA$ 
and $\classB$. Before the first round and at the end of every round 
we delete all pebbled structures which are not partially isomorphic to any structure on the other side, as in the $k$-QVT game. Another 
similarity is that at the beginning of each round Spoiler may choose 
to either split or pebble. Pebbling works in exactly 
the same way as before. The difference occurs when Spoiler splits $\classX \in \{\classA, \classB\}$ into $(\classX_i)_{i \in [\ell]}$. Then Duplicator 
chooses $L \subseteq [\ell]$ and Spoiler some $i \in L$. Call $|L|$ the \emph{degree} of the split. Play then continues by renaming $ \strucX_i$ to $\strucX$ 
and\confORfull{}{, as in the $k$-QVT game,} Spoiler can again choose whether to split or pebble. For technical reasons, we stipulate that if they choose to split for a second time they cannot split $\classX$.

The other difference is that there is a scoring system. We define $\points(r)$ to be the product of the degrees of every split that occurred in round $r$ or earlier, using the convention that an empty product evaluates to one. Then after round $r$, we add $\points(r)$  to Duplicator's total score\confORfull{.}{, i.e., if Duplicator survives for $r$ rounds then their score is $\sum_{s=1}^r \points(s)$.} 

\begin{lemma} \label{lem: lbgame}
Let $\classA, \classB$ be sets of $k$-pebbled structures. If there is a $\phi \in \klogic$ with $\quant(\phi) \le r$ which $\classA$ and $\classB$ disagree on, then Spoiler can limit Duplicator to $r$ points in the $k$-LB game from position $[\classA, \classB]$.
\end{lemma}

\begin{proof}
By Theorem~\ref{thm:eqiv} it is enough to show that if Spoiler can win the $k$-QVT game from position $[\classA, \classB]$ in $r$ rounds then they can limit Duplicator to at most $r$ points in the $k$-LB game. We proceed by induction on $r$. In the base case, Spoiler wins the $k$-QVT game in zero rounds, so---since the winning conditions are identical in this case---Spoiler also wins the $k$-LB game immediately and Duplicator gets zero points.

Now suppose that for any two compatible sets of pebbled structures, $\classX, \classY$, whenever Spoiler wins the $k$-QVT game in $s< r$ rounds from position $[\classX, \classY]$, they can limit Duplicator to at most $s$ points in the $k$-LB game. Further suppose that Spoiler wins in $r$ rounds from position $[\classA, \classB]$ in the $k$-QVT game. We claim that they can limit Duplicator to at most $r$ points in the $k$-LB game on $[\classA, \classB]$. 

First, suppose that Spoiler's $r$-round winning strategy in the $k$-QVT game on $[\classA, \classB]$ begins by pebbling. Then Spoiler can perform an identical move in the $k$-LB game; Duplicator receives one point for this round. Therefore, the position at the beginning of round two is identical in both games. Moreover, Spoiler can win in $r-1$ rounds from this position in the $k$-QVT game. So Spoiler can limit Duplicator to $r-1$ points after the first round in the $k$-LB game and therefore to $r$ points overall.

Otherwise, Spoiler begins the $k$-QVT game by splitting $\classB$ into $(\classB_i)_{i \in [\ell]}$. If at any of the new nodes Spoiler's next move is to perform 
a further split of some $B_j$ into $(B_{j_i})_{i \in [t]}$, then Spoiler 
could have originally performed a split of $\classB$ into 
\[B_1, \dots, B_{j-1}, B_{j_1}, \dots, B_{j_t},  B_{j+1}, \dots, B_{\ell}\] 
with the same overall 
effect. Therefore, we may assume w.l.o.g.\! that this phenomenon does not occur. 
Further, as in the proof of Theorem~\ref{thm:eqiv}, we may assume that Spoiler 
can win from position $[\classA, \classB_i]$ in $k_i<r$ rounds for every $i$. 
Then in the $k$-LB game Spoiler can begin by performing the same split; 
Duplicator responds by choosing some $L \subseteq [\ell]$. Let $k_j:= \min_{i 
\in L}\{k_i\}.$ Then Spoiler chooses to continue from position $[\classA, 
\classB_j]$.

By the induction hypothesis in the $k$-LB game starting from position 
$[\classA, \classB_j]$ Spoiler can limit Duplicator to at most $k_j$ 
points. Therefore, Spoiler can limit Duplicator to at most $|L| \cdot 
k_j $ points in the $k$-LB game starting from position $[\classA, 
\classB].$ On the other hand, Spoiler needs $\sum_{i=1}^{\ell} k_i $ 
rounds to win the $k$-QVT game starting from position $[\classA, 
\classB]$. By our choice of $j$, 
\[r = \sum_{i=1}^{\ell} k_i  \ge \sum_{i \in L} k_i \ge |L| 
\cdot k_j, \]
which completes the proof.
\end{proof}

\subsection{The Existential QVT Game}

The $\exists^{+}k$-QVT game is the same as the $k$-QVT game except for the following changes. \confORfull{Firstly, Spoiler is only allowed to pebble the $\LHS$. Secondly, before the game starts and at the end of every round, we only delete structures from the $\RHS$. We consequently say $t \in V(\mathcal{T}_r)$ is closed if it is not a leaf or if $\chi^r(t)$ is of the form $[\classA, \emptyset]$. Finally, the role of partial isomorphism is replaced by partial homomorphism. This means before the game starts and at the end of 
a round we delete $\lstrucBb \in \RHS$ if for every $\lstrucAa \in \LHS$, the canonical partial map from $\lstrucAa$ to $\lstrucBb$ is not a partial 
homomorphism. The following is proved similarly to Theorem~\ref{thm:eqiv}.}
{\begin{enumerate}
\item Spoiler is only allowed to pebble the $\LHS$\confORfull{.}{, i.e., if in round $r$, $\chi^r(c^r) = [\classA, \classB]$, then Spoiler must play on $\classA$. Note that Spoiler is still allowed to perform a split on either side.}
\item Before the game starts and at the end of every round, we only delete structures from the $\RHS$. We say $t \in V(\mathcal{T}_r)$ is closed if it is not a leaf or if $\chi^r(t)$ is of the form $[\classA, \emptyset]$.
\item We replace the role of partial isomorphism with partial homomorphism.
\end{enumerate}

\confORfull{P}{To be more explicit, p}utting (2) and (3) together, before the game starts and at the end of 
a round we delete $\lstrucBb \in \RHS$ if for every $\lstrucAa \in \LHS$, the canonical partial map between the two pebbled structures is not a partial 
homomorphism. \confORfull{}{Note that, while splitting is allowed on the $\LHS$, in all the examples we will look at the $\LHS$ consists 
initially of only a single labelled structure. Since Duplicator is never allowed to pebble on the left, at every position in the game there is exactly one labelled structure on the $\LHS$ and so Spoiler never splits the $\LHS$.} The following \confORfull{is proved similarly to Theorem~\ref{thm:eqiv}.}{is proved in Appendix~\ref{a: exequiv} in a similar way to Theorem~\ref{thm:eqiv}.}}

\begin{theorem} \label{thm: exeqiv}
Let $\classA, \classB$ be sets of $k$-pebbled structures. Then Spoiler can win the $\pexists k$-QVT game from position $[\classA, \classB]$ in $r$ rounds if and only if there exists a $\exklogic$-formula $\phi$ that $\classA$ and $\classB$ disagree on with $\quant(\phi) = r$.
\end{theorem}

By the same argument as Lemma~\ref{lem: ub}, we can translate quantifier depth lower bounds in $\exklogic$ to quantifier number lower bounds in $\exklogic$. Moreover, in this case the correct upper bound is known: Berkholz showed that there are two $n$-element structures distinguishable in $\exklogic$, but that agree on every sentence with  $o(n^{2k-2})$  quantifier depth \cite[Theorem 1]{DBLP:conf/cp/Berkholz14}. This matches the $n^{2k-2}$ upper bound up to a constant factor. We obtain the following analogue\confORfull{ to Lemma~\ref{lem: ub}}{s to Corollary~\ref{cor: ub} and Lemma~\ref{lem: ub}}.

\begin{lemma} \label{lem: exub}
If two $n$-element structures can be distinguished by 
 $\exklogic$, then there is a sentence $\phi \in \exklogic$ which the structures disagree on with 
 \[ 
  \quant(\phi) \le \frac{n^{n^{2k-2}-1}}{n-1}.
 \] 
\end{lemma}

\confORfull{}{\begin{lemma}
Suppose $\classA, \classB$ are two sets of compatible $k$-pebbled structures which disagree on a formula of $\exklogic$. Moreover, suppose each $\strucX \in \classA \cup \classB$ has at most $n$ elements. Then there is a formula $\phi \in \exklogic$ which $\classA$ and $\classB$ disagree on with 
 \[ 
  \quant(\phi) = n^{O(n^{2k-2})}
 \] 
\end{lemma}}

\confORfull{\subsection{The Existential Lower Bound Game}}{\subsection{The Existential Lower Bound Game}}

In Section~\ref{s: exlb} we prove Theorem~\ref{thm: exlb} using an existential-positive version of the $k$-LB game: the $
\exists^{+} k$-LB game. It is defined in the obvious way. Play proceeds in the same way as the $\exists^{+} k$-QVT game except in the case of splits---which are dealt with as in the same way as the $k$-LB game---and the existence of a points system---which is identical to that of the $k$-LB game. \confORfull{}{The proof of the following is essentially the same as Lemma~\ref{lem: lbgame}.}

\begin{lemma} 
Let $\classA, \classB$ be sets of $k$-pebbled structures. If there is a $\phi \in \exklogic$ with $\quant(\phi) \le r$ which $\classA$ and $\classB$ disagree on, then Spoiler can limit Duplicator to $r$ points in the $\pexists k$-LB game starting from position $[\classA, \classB]$.
\end{lemma}

\confORfull{}{ The nice thing about the lower bound games is that we no longer have to keep track of a whole 
tree and instead have only two evolving sets of $k$-pebbled 
structures. Moreover, in Section~\ref{s: exlb} we will be in the  existential-positive case and 
start from a position with only one structure on each side; we will therefore only ever have a 
single $\LHS$ structure.}

\section{An Exponential Lower Bound} \label{s: lb}

In this section, we prove Theorem~\ref{thm: genlb}. To this end,\confORfull{}{ in Section~\ref{ss: XOR_gen},} we introduce a method for building pairs of structures from sets of constraints over elements of an abelian group. The key is that the two structures are isomorphic if and only if the set of constraints is satisfiable. \confORfull{}{This generalises a construction which produces pairs of relational structures from XOR formulas \cite{berkholz2016}, which in turn was inspired by previous graph constructions based on XOR formulas \cite{cai1992optimal, furer2001weisfeiler, immerman1981number}. We will only need the cases $G = (\mathbb{Z}_2^k, +)$, but we are able to show some nice properties of this construction in the general case essentially for free.}\confORfull{}{In \confORfull{Appendix}{Section}~\ref{ss: XOR_fail} we gesture at why it does not suffice to consider structures built from XOR formulas.} \confORfull{}{This highlights a surprising Spoiler resource which in turn helps motivate the exact construction we give in Section~\ref{ss: LB} to prove our lower bound.}

\subsection{XOR Generalisation} \label{ss: XOR_gen}

\confORfull{}{We now introduce a construction for turning a set of constraints into a pair of relational structures. Playing games on these structures is closely related to the original constraints; this allows us to analyse games played on such structures more easily.} 

\confORfull{Fix}{So fix} an integer $k$, $(G,+)$ a finite abelian group and $V$ a set of variables. Let $F$ be a set of constraints in variables $V$ of the form $\sum_{x \in S} x \in H$ for some $S \subseteq V$ with $|S| \le k$ and some $H \subseteq G$.  We write $(x_{1}, \dots, x_{|S|}, H)$ for such a constraint, if $S= \{x_1, \dots, x_{|S|}\}$. We will also abbreviate this as $(S, H)$. For a constraint $C$, we write $|C|$ for the number of variables occurring in $C$, $\dom(C)$ for the set of these variables and $H_C$ for the subset of $G$ appearing in $C$. We say an assignment $f : V \to G$ satisfies $C$ if $\sum_{i \in [|S|]} f(x_i) \in H$ and that $F$ is satisfied by $f$ if every $C\in F$ is satisfied by $f$. 

We define a relational structure $\strucX=\strucX(F)$ as follows. The domain of the structure contains $|G|$ elements for every variable, i.e., $X = \{x^g \, \mid \, x\in V, \, g \in G \}$. We also have a colour for each variable, that is for each $x \in V$ we define a unary relation $R_x^{\strucX}:= \{x^g \, \mid \, g \in G\}$. If $ y\in R_x^{\strucX}$, we write $\pi(y):= x$ and say $y$ projects to $x$; we extend this notation to sets in the natural way. We also define $|x^g| := g$. Then for each constraint $C =(x_{1}, \dots, x_{m}, H)$ we have a $m$-ary relation symbol $R_C$ with:
\[
R_C^{\strucX} := \left\{(x_{1}^{a_1}, \dots, x_{m}^{a_{m}}) \; \middle|  \; \sum_i a_i \in H \right\} 
\]
Now partition $F$ into two sets $N$ and $D$. If $C 
\in D$ we say it is a \emph{distinguishing constraint}. Define $
\hat{F} := N \cup \hat{D}$, where 
$\hat{D}:= \{\hat{C}:= (S, G\setminus H) \, \mid \, C=(S,H) \in D\}$. Then $\strucA(F):= \strucX(\hat{F})$ and $\strucB(F) := \strucX(F)$. One technicality: for every $C \in D$ we rename the relation $R_C$ as $R_{\hat{C}}$ in $\strucA(F)$, so that both structures share a common signature.

\confORfull{}{We have made a few choices in our definition which are simply the most convenient in our present setting. For example, we introduce a fresh relation for every constraint to allow us to have multiple constraints on the same set of variables.} We restrict ourselves to  instances where every $H \subseteq G$ appearing in a constraint of $\hat{F}$ is a subgroup of index two, i.e., where every $H$ appearing in a constraint of $F$ is either a subgroup of index two or the complement of such a subgroup. \confORfull{}{The reason to consider such a restriction is that the structures emerging from such sets of constraints are easier to analyse.} We will refer to constraints of this form as \emph{clauses} and to an $F$ consisting exclusively of clauses as a \emph{formula}. Before plunging in we give an example.

\begin{example}
Let $G = (\mathbb{Z}_2, +)$. The only subgroup of $G$ of index two is $\{0\}$. Let $F$ be a formula over $G$. Then every clause is of the form $\sum_i x_i \in \{0\}$ or $\sum_i x_i \in \{1\}$. Therefore, $F$ is an XOR formula and conversely, every XOR formula is a formula over $G$. In fact, in this case, our construction is identical to that given in \cite{berkholz2016}. 
\end{example}

Our use of index two subgroups is motivated by the following easy characterisation\confORfull{.}{; for completeness, we give a proof in the appendix.}

\begin{lemma} \label{lem: index2}
Let $G$ be a finite group, let $H \leqslant G$. Then $H$ has index two iff for every $g, h \not \in H$, $g + h \in H$.
\end{lemma}

How does playing games on such structures relate to the original set of constraints? \confORfull{S}{To see the idea s}uppose Spoiler plays $x^g$ on some board in the $k$-pebble game on $[\strucA(F), \strucB(F)]$. Then Duplicator must reply with an element of the same colour, i.e. some $x^{\hat{g}}$. This corresponds to Spoiler choosing variable $x$ and Duplicator setting it to be equal to $g + \hat{g}$. Similarly, each position in the $k$-pebble game corresponds to a partial assignment of variables in $F$ and Spoiler wins after $r$ rounds iff the corresponding partial assignment does not satisfy $F$. We can therefore recast the $k$-pebble game as being played on $F$. This is not\confORfull{}{ quite} the case for the $k$-QVT game, but we can still exploit the relationship to the original formula in our analysis. 

We now make the above remarks precise by giving a correspondence between (partial) assignments satisfying $F$ and (partial) isomorphisms between $\strucA(F)$ and $\strucB(F)$.  So let $Y= \{y_1, \dots, y_t\} \subseteq A(F)$ and  $\sigma : Y \to B(F)$. Denote $\pi(y_i)$ by $x_{i}$ and collect these in a set $X$. We say that $\sigma$ is \emph{c defining} if two conditions are met. Firstly, the map must respect the unary relation $R_x$ for every $x \in X$. Secondly, for every $y, z$ with $\pi(y) = \pi(z)$, $|\sigma(y)| +|y| = |\sigma(z)| + |z|$. From a partial assignment defining $\sigma$ we define a partial assignment $\hat{\sigma} : X \to G$ by $\hat{\sigma}(x) = |\sigma(y)| + |y|$, where $y$ is any element of $Y$ such that $\pi(y)=x$. 

Conversely, let $X \subseteq V$, $f : X \to G$ and $\{y_1, \dots, y_t\}:=Y \subseteq A(F)$ with $\pi(Y) =X$. Then we define a map $\hat{f}_Y: Y \to B(F)$, which respects the relations $R_x$ for every $x \in X$, by $|\hat{f}_Y(y)| = f(x) - |y|$, where $x = \pi(y)$. \confORfull{This map is clearly partial assignment defining.}{Then, since if $\pi(y) = \pi(z)$, $|\hat{f}_Y(y)| -|\hat{f}_Y(z)| = |z|-|y|$, this map is assignment defining. }

\begin{lemma} \label{lem: passignment}
If a partial assignment defining map $\sigma$ is a partial isomorphism then $\hat{\sigma}$ does not violate any clause of $F$. Conversely, if $X \subseteq V$, $f: X \to G $ is a partial assignment which does not violate any clause of $F$ and $Y \subseteq A(F)$ with $\pi(Y) = X$, then $\hat{f}_Y$ is a partial isomorphism. \confORfull{}{In particular, $F$ is satisfiable iff $\strucA(F)$ and $\strucB(F)$ are isomorphic.}
\end{lemma}

\begin{proof}

First suppose that $\sigma: Y \to B(F)$ is a partial isomorphism, for some $Y \subseteq A(F)$. Let $C 
=(x_1, \dots, x_m, H)$ be a clause of $F$ with $x_i \in \pi(Y)$ 
for every $i$. We need to show that $\hat{\sigma}$ 
does not violate $C$. For each $i$ pick some $y_i \in Y$ that projects to 
$x_i$. Then $\hat{\sigma}(x_i) = |\sigma(y_i)| +|y_i|$ by 
definition. Therefore, $\sum_i \hat{\sigma}(x_i) = \sum_i |
\sigma(y_i)| +|y_i|$. Since $\sigma$ is a partial isomorphism 
it in particular preserves the relation $R_C$. Therefore, if 
$C$ is not a distinguishing clause we have that $
\sum_i |\sigma(y_i)| \in H$ iff $\sum_i |y_i| \in H$.  
Since $H$ is a subgroup of index two we may conclude 
that $\sum_i \hat{\sigma}(x_i) \in H$. If $C$ is a 
distinguishing clause we have that $\sum_i |\sigma(y_i)| \in 
H$ iff $\sum_i |y_i| \not \in H$. Since in this case $H$ is the complement of a subgroup we again obtain $
\sum_i \hat{\sigma}(x_i) \in H$. Therefore $\sigma$ does not violate any clause of $F$.

For the other direction suppose that $f: X \to G$ does not violate any clause of $
F$ and let $Y \subseteq A(F)$ with $\pi(Y) = X$. Suppose for a contradiction that $\hat{f}_Y$ is 
\emph{not} a partial isomorphism. Then there must be some 
relation $R$, say of arity $m$, and sequence $\underline{y}= 
(y_{i})_{i \in [m]}$ of elements in $Y$ such that 
$R^{\strucA}(\underline{y})$ and $R^{\strucB}
(\hat{f}_Y(\underline{y}))$ have different truth values. Moreover, by the definition of $\sigma $, we know that $R = R_C$ for some $C\in F$, so $\strucB \models R(\hat{f}_Y(\underline{y}))$ iff $\sum_{i=1}^m|\hat{f}_Y(y_{i})| \in H_C$. Since $f$ does not violate any clause of $F$ we also know that $\sum_{i=1}^m f(x_{i}) \in H_C$.  Then by definition 
\[\sum_{i=1}^m|\hat{f}_Y(y_{i})| = \sum_{i=1}^m f(x_{i}) - |y_{i}| = \sum_{i=1}^m f(x_{i}) - \sum_{i=1}^m |y_i|.
\]
If $C$ is \emph{not} a distinguishing constraint then $\sum_{i=1}^m|\hat{f}_Y(y_{i})| \in H_C$  iff $\sum_{i=1}^m |y_i|\in H_C$ iff $\strucA \models R_C^{\strucA}(\underline{y})$. This implies that $\strucB \models R(\hat{f}_Y(\underline{y}))$ iff $\strucA \models R_C^{\strucA}(\underline{y})$. So $C$ must be a distinguishing constraint. But then $\sum_{i=1}^m|\hat{f}_Y(y_{i})| \in H_C$ iff $\sum_{i=1}^m |y_i| \not\in H_C$ iff $\strucA \models R_C^{\strucA}(\underline{y})$. Here the first iff holds becomes $H_C$ is the complement of an index two subgroup. This is a contradiction, so $\hat{f}_Y$ is a partial isomorphism.
\end{proof}

The above lemma is very useful and we will now outline some consequences. Firstly, to specify an isomorphism between $\strucA(F)$ and $\strucB(F)$ it is enough to give a satisfying assignment for $F$. Secondly, the following corollary will be frequently used.

\begin{corollary} \label{cor: sumreduct}
Let $\assA, \assB$ be two compatible $k$-assignments such that $\pi(\assA(i)) = \pi(\assB(i))$ for all $i \in \dom(\assA)$. Moreover, suppose that the map $\assA(i) \to \assB(i)$ is partial assignment defining. Then $\mathcal{A}(F)^{\assA}$ and $\mathcal{B}(F)^{\assB}$ are partially isomorphic iff for every $C = (x_1, \dots, x_m, H) \in F$ and $(i_1,\dots  ,i_m ) \in \dom(\assA)^m$ with $\pi(\alpha(i_j)) = x_j$, 
\[\sum_{j =1}^m |\assA(i_j)| + |\assB(i_j)| \in H
\]
\end{corollary}

\confORfull{}{
\begin{proof}
Let $\lstrucAa$ and $\lstrucBb$ be as above and let $\sigma$ be the associated canonical partial map. Suppose that $\sigma$ is not a partial isomorphism. Then by Lemma~\ref{lem: passignment}, $\hat{\sigma}$ violates a clause of $F$, say $C = (x_1, \dots, x_{m}, H)$. Let $(i_1,\dots  ,i_m ) \in \dom(\assA)^m$ with $\pi(\alpha(i_j)) = x_j$. Then, since $\hat{\sigma}$ violates $C$,
\[ 
\sum_{j=1}^m |\assA(i_j)| + |\assB(i_j)| = \sum_{j=1}^m |\assA(i_j)| + |\sigma(\assA(i_j))|= \sum_{i=1}^m \hat{\sigma}(x_i) \not \in H.
\] 

Conversely suppose $\sigma$ is a partial isomorphism. Suppose $C = (x_1, \dots, x_{m}, H)$ is a constrain from $F$ and $(i_1,\dots  ,i_m ) \in \dom(\assA)^m$ with $\pi(\alpha(i_j)) = x_j$. Then $\hat{\sigma}$ satisfies $C$ by Lemma~\ref{lem: passignment}. Therefore, $\sum_{j=1}^m |\assA(i_j)| + |\assB(i_j)| = \sum_{i=1}^m \hat{\sigma}(x_i) \in H$.
\end{proof}}

Thirdly, we can now reconceive $k$-pebble games\confORfull{}{ played on such structures} as being played directly on formulas. This, again, generalises the XOR case. \confORfull{To be precise, Spoiler can win the $k$-pebble game iff they can win the $k$-formula game, defined as follows. The game is played on $V$ and Spoiler has a set of $k$ pebbles labelled with elements of $[k]$. If pebble $p$ lies on $x$ we write $\pi(p)=x$. In each round Spoiler places a pebble $p$ on some variable $x \in V$ and Duplicator chooses an element $g \in G$. We write $f(\pi(p)) = g$. Then Spoiler wins at the end of round $r$ if the map $f=f_r : \{\pi(p) \, \mid \, p \in [k]\} \to G$ violates a clause of $F$.

The above corresponds in the $k$-pebble game to Spoiler choosing $x_i^h$ and Duplicator responding with the $x_i^u$ such that $h + u =g$. Note that Duplicator makes such a response w.l.o.g. as any other response is immediately losing. Then by Corollary~\ref{cor: sumreduct} and induction it is easy to see that if Spoiler has a winning strategy for the formula game on $F$ this corresponds to a winning strategy for the $k$-pebble game on $[\strucA(F), \strucB(F)]$. This is the direction we need; the correspondence in the other direction is similar but do not use it.}{We need some terminology. If on Spoiler's turn they choose, on some board, an element $x^{g}$ then we say that Spoiler plays $g$ on $x$. Because of the unary relations, any valid Duplicator response to this move must involve them playing $x_i^{\hat{g}}$, we simply say that Duplicator replies with $\hat{g}$ as we may assume w.l.o.g.\! that Duplicator only makes valid responses.

Suppose Spoiler plays $g$ on $x$ in the $k$-pebble game. Then we may assume w.l.o.g.\! that Duplicator always replies in a partial assignment defining manner. Moreover, it doesn't matter which element Spoiler plays on $x$. This follows from the following three facts.
\begin{enumerate}
\item  Spoiler wins iff $\hat{\sigma}$ violates some clause.
\item The value of this valuation on $x$ is determined by the sum of  Spoiler's move and Duplicator's response.
\item The map $ h \to h + g$ is a bijection on $G$ for every $g\in G$.
\end{enumerate}

Therefore, Spoiler can win the $k$-pebble game iff they can win the $k$-formula game, defined as follows. The game is played on $V$ and Spoiler has a set of $k$ pebbles labelled with elements of $[k]$. If pebble $p$ lies on $x$ we write $\pi(p)=x$. In each round Spoiler places a pebble $p$ on some variable $x \in V$ and Duplicator chooses an element $g \in G$ to write on $p$. We write $f(\pi(p)) = g$. Then Spoiler wins at the end of round $r$ if the map $f=f_r : \{\pi(p) \, \mid \, p \in [k]\} \to G$ violates a clause of $F$.

The above corresponds in the $k$-pebble game to Spoiler choosing $x_i^h$ and Duplicator responding with the $x_i^u$ such that $h + u =g$. Note that Duplicator makes such a response w.l.o.g. as any other response is immediately losing. Then by Corollary~\ref{cor: sumreduct} and induction it is easy to see that if Spoiler has a winning strategy for the formula game on $F$ this corresponds to a winning strategy for the $k$-pebble game on $[\strucA(F), \strucB(F)]$. This is the direction we need; the correspondence in the other direction is similar but we do not use it.}

\confORfull{}{
\subsection{The Problem with XOR formulas} \label{ss: XOR_fail}

We now give an extended example highlighting why we need to generalise beyond XOR formulas. So fix $G=(\mathbb{Z}_2,+)$, so that a formula over $G$ is an XOR formula. Here, and throughout, we will often forgo the set notation and instead use notation involving equalities, with the obvious interpretation. Let $V = \{x_i \, \mid \, i \in [8] \}$ and let $F$ be the formula with the following clauses:
\begin{enumerate}
\item $x_1=1$, this is the only distinguishing clause.
\item $x_1 + x_2 + x_3 = 0$
\item $x_2 + x_4 + x_5 = 0$
\item $x_3 + x_6 + x_7 = 0$
\item $x_4 + x_5 + x_8 = 0$
\item $x_6 + x_7 + x_8 = 0$
\item $x_8 = 0$.
\end{enumerate}
We can visualise our formula as a graph where each triangle represents a non-distinguishing clause on the vertices involved in the triangle and a black (resp. blue) loop corresponds to a variable being forced to zero (resp. one), see Figure~\ref{fig: XOR_vis}. It is then easy to see that $F$ is unsatisfiable: the first clause forces $x_1$ to be one and we can then `push' this from left to right eventually forcing $x_8$ to be zero, a contradiction.\footnote{A similar intuition to this can be formalised and in fact in \cite{berkholz2016} they focus on XOR formulas induced by DAGS. Generalising this to our setting is probably possible but unnecessary for our purposes.} For example, the second clause forces either $x_2$ or $x_3$ to be one, since we know $x_1=1$. It follows by Lemma~\ref{lem: passignment} that $\strucA(F)$ and $\strucB(F)$ are not partially isomorphic. 

\begin{figure} 
\centering
\includegraphics[width=\textwidth]{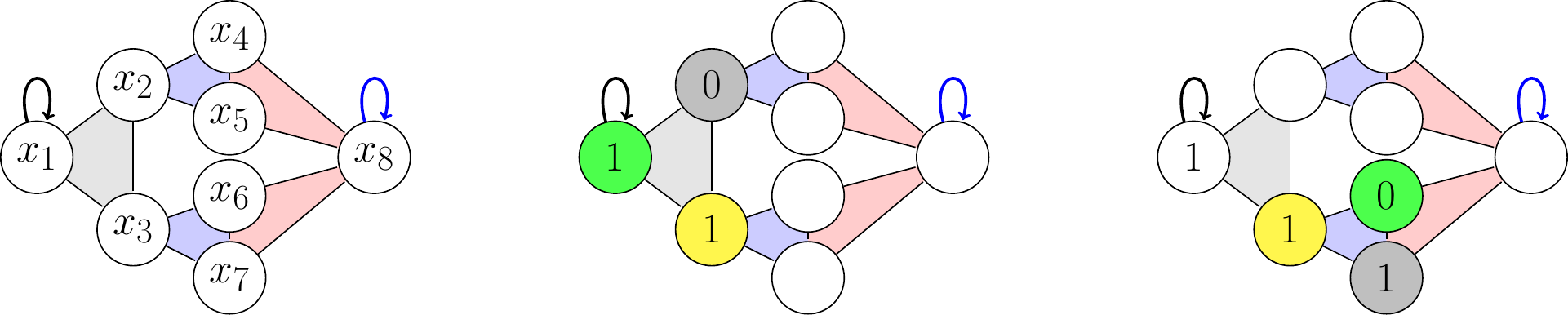}
\caption{From left to right: a visualisation of $F$; the first three rounds of the $3$-formula game on $F$, the filled vertices represent the fact the corresponding elements are pebbled, the labels are the Duplicator moves; a possible continuation of the game. In the next move Spoiler will move the yellow pebble to $x_8$ and win.}
\label{fig: XOR_vis}
\end{figure}

One can similarly see that Spoiler can win the $3$-QVT game on these structures. To see this recall that it is enough to show that Spoiler can win the $3$-pebble game which is equivalent to the $3$-formula game. Now Spoiler's winning strategy in the $3$-formula game is to `push' the one through from left to right. In detail they begin by placing pebble one on $x_1$, Duplicator must reply with one by constraint (1). They then place pebble two on $x_2$ and pebble three on $x_3$. Duplicator must reply with one on either pebble two or three by constraint (2). In the former case, Spoiler leaves a pebble on $x_2$ and moves the other two pebbles to $x_4$ and $x_5$. Then by constraint (3) the sum of the Duplicator responses to these two moves must be 1. Finally Spoiler moves the pebble lying on $x_2$ to $x_8$. By constraint (5) Duplicator must reply with $1$ but by constraint (7) this reply is not valid. Therefore Spoiler wins. The case where after three rounds there is a one on pebble three is similar, except now Spoiler utilises constraints (4) and (6) instead of (3) and (5).

A natural question is whether Spoiler has to perform a split in order to win the $3$-QVT game on these structures. This was our original idea for obtaining a lower bound: we wanted to take $n$ copies of these structures and then chain them together in the right way. Our idea was that Spoiler will have to perform $n$ splits of degree two in order to win. But in fact Spoiler can win without performing any splits at all!

First: why did we think that Spoiler needed to split in order to win? Well because this is what happens if we restrict Spoiler to play only on the $\LHS$. Moreover, as we have seen, in the pebble game over structures emerging from XOR formulas we may assume w.l.o.g. that Spoiler always plays on the $\LHS$, this is why we can recast the game as being played directly on the formula. However, interestingly, in the $k$-LB game this assumption does involve a loss of generality. Let us give some more details.

Why does Spoiler need to split if they are restricted to play on the $\LHS$? We will provide a sketch of this argument. Firstly, it is easy to see that we may assume w.l.o.g. that Duplicator always plays elements of the form $x^0$. Next it is easy to prove, that if we never reach a position where for each $i\in [3]$, on the unique $\LHS$ board $\lstrucAa$, $\assA(i) = x_i^0$ (up to renaming of pebbles) then Spoiler cannot win, call this position the \emph{critical position}. Now suppose we are in the critical position and Spoiler wants to make progress;  there will be two partially isomorphic $\RHS$ boards $\lstrucAb$, $\lstrucAc$ with $\assB(1) = \assC(1) = x_1^1$ and $\assB(2) = x_2^1 \neq \assC(2)$ and $\assC(3) = x_3^1 \neq \assB(3)$. If Spoiler wants to win they need to win on both these boards. But there is no way to make progress on one board without destroying all progress on the other, see Figure~\ref{fig: xor_multi}.

\begin{figure} 
\centering
\includegraphics[width=0.7\textwidth]{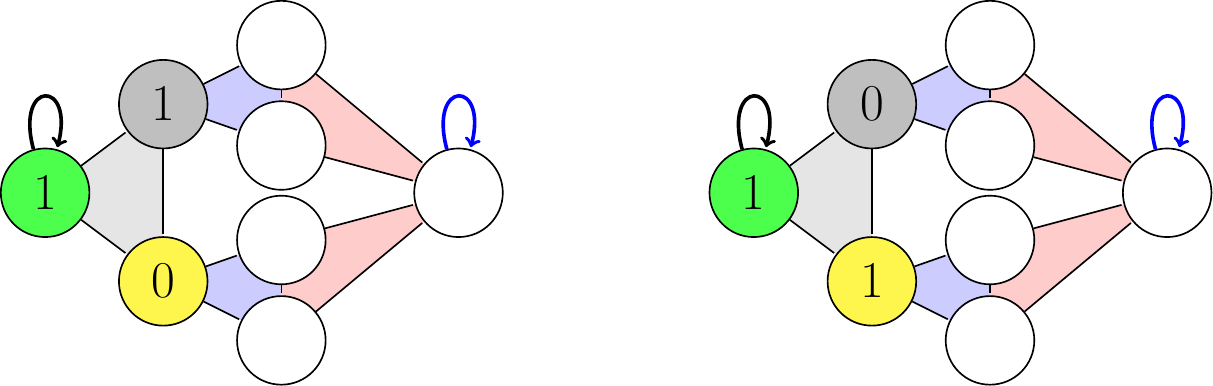}
\caption{A visualisation of the two $\RHS$ boards in a critical position.}
\label{fig: xor_multi}
\end{figure}

To see this consider just the board $\lstrucBb$. To make progress Spoiler must `push the one' further and the only way to do this is to moves pebbles one and three to cover $x_4^0$ and $x_5^0$. Then because of constraint (3) and as $\assB(2) = x_2^1$, Duplicator must reply with either $x_4^1$ and $x_5^0$ or $x_4^0$ and $x_5^1$. Spoiler rejoices: progress has been made! But not so fast: for now consider the situation on $\lstrucBc$. There Duplicator can simply reply to Spoiler's moves with $x_4^0$ and $x_5^0$: not only has no progress been made but progress has actually been reversed. Now Duplicator discards every board apart from this one and it is easy to see we are once again in a position where Spoiler needs to reach a critical position in order to win.

It is straightforward, if a little tedious, to turn the above intuitions into a formal argument. Now let us see how Spoiler can win the $3$-QVT game without splitting. Again let us consider the critical position, but now let us give Spoiler the ability to also play on the $\RHS$. Then Spoiler has the following devious tactic: that of freezing a board. What exactly does this mean? Consider what happens if Spoiler moves pebble $i$ on the $\RHS$ and on $\lstrucBc$ moves $i \to \assC(i)$, that is they do not really move pebble $i$ at all but rather pick it up and place it back where it was before. Then it is easy to see that on $\lstrucAa$ the only valid response is to play $i \to \assA(i)$, due to constraint (2). We call this \emph{freezing} a board. Note that it does not matter which pebble Spoiler chooses.

Of course, such a manoeuvre is entirely pointless in the case of pebbles games, but the presence of multiple boards changes things entirely. Now starting from the critical position Spoiler on the $\RHS$ in each round. On $\lstrucBb$ they  push the one through the structure, as in the winning pebble game strategy. While they do this they freeze board $\lstrucBc$. At the end of this procedure, Spoiler wins on all boards stemming from $\lstrucBb$ and the position is $[\{\lstrucAa\}, \{\lstrucBc\}]$. It is then easy to see that Spoiler can win from this position without splitting again by `pushing the one though the structure.'

We should observe that this phenomenon of freezing was not specific to the $F$ used above. Rather the issue stems from the fact that if in an XOR clause we fix all but one variable then the value of the final variable is fixed. This raises the question: can any exponential lower bounds be obtained using structures emerging from XOR formulas? Instead of answering this question we avoid the hurdle of freezing altogether by working with a different group. }

\subsection{Lower Bound} \label{ss: LB}

\confORfull{}{The musings of the last subsection motivate us to use $G= (\mathbb{Z}_{2}^k, +), k>1$. The point is that now if we fix the value of all variables in a clause bar one, then the value of the remaining variable isn't  determined. Therefore, Spoiler will not necessarily be able to perform freezes. There are two additional benefits. Firstly, in this setting we can force splits of degree $k-1$. Secondly, the lower bound argument is relatively simple.}

\subsubsection{Forcing Spoiler to Split}
We now show that there exists two structures $\strucA, \strucB$ such that Spoiler can win the $k$-QVT game from position $[\{\strucA\}, \{\strucB\}]$ but must perform a split in order to do so. In the next subsection, we chain these structures together to obtain Theorem~\ref{thm: genlb}. Here we assume $k$ is odd and set $G= (\mathbb{Z}_2^{k},+)$; a similar construction works for even $k$, see Appendix~\ref{a: even}.

We need some notation. Let $\mathsf{EVEN}$ ($\mathsf{ODD}$), denote those elements of $\mathbb{Z}_{2}^k$ whose coordinates sum to an even (odd) number. We denote the $i$th coordinate of $g \in \mathbb{Z}_2^k$ by $g[i]$. Moreover, we write $\underline{0}$ for the identity element, $\underline{0}_i$ for the element of $\mathbb{Z}_2^k$ with $\underline{0}_i[\ell]= 1$ iff $\ell =i$ and $\underline{1}:= (1, 1, \dots, 1)$. Let $V:=\{s_1, \dots s_{k}, e_1, \dots, e_k\}$. We define $F$ to be the formula consisting of the following clauses.
\begin{itemize}
\item $s_1 \in \mathsf{ODD}$, this is the only distinguishing clause.
\item $s_i \in \mathsf{EVEN}$ for $i \in [k] \setminus \{1\}$. 
\item  $(e_{\ell} +\sum_{i \in [k]\setminus \{\ell\}} s_i)[\ell] =0$, for  $\ell \in [k]$.
\item $ s_{i}[i]=0$, for  $i \in [k]$
\item $e_{\ell}[i] = 0$, for $\ell, i \in [k]$
\end{itemize}

Observe that the last set of constraints forces every $e_{\ell}$ to be $\underline{0}$. \confORfull{}{We need to phrase this in terms of multiple constraints due to our requirement that for every clause $C$, $H_C$ is a subgroup of index two or the complement of such a subgroup.}\confORfull{}{\footnote{We could have replaced all the $e_{\ell}$ with a single variable. The reason we for not doing this is that it will allow us to `chain together' copies of these structures more easily later on.}} We will write $\strucA/\strucB$ instead of $\strucA(F)/\strucB(F)$.

\begin{lemma} \label{lem: spoiler-win-gen}
$\strucA$ and $\strucB$ disagree on a sentence of $\klogic$.
\end{lemma}

\begin{proof}
It suffices to consider the $k$-formula game on $F$. In the first $k$ rounds, Spoiler moves pebble $i$ to $s_i$ for every $i\in [k]$. Suppose Duplicator makes responses $g_i$ to the move $s_i$. If there is some $\ell \in [k]$ such that $\sum_{i \in [k]\setminus \{\ell\}} g_i[\ell] =1$ then in the next round Spoiler moves pebble $\ell$ to $e_{\ell}$. Duplicator must reply with some element of $g$ whose $\ell$th coordinate is one, to satisfy the constraint $(e_{\ell} +\sum_{i \in [k]\setminus \{\ell\}} s_i) [\ell] = 0$. But this is a contradiction to the constraint $e_{\ell}[\ell] =0$. 

Now suppose that, $\sum_{i \in [k]\setminus \{\ell\}} g_i[\ell] =0$, for every $\ell$. Note that $\sum_{i\in[k]} g_i \in \mathsf{ODD}$, as $g_1 \in \mathsf{ODD}$ and every other $g_i \in \mathsf{EVEN}$. But also $\sum_{i\in[k]} g_i = (g_1[1], g_2[2], \dots, g_k[k]) = \underline{0}$, a contradiction. \confORfull{}{Here the second equality follows from the constraint $s_{\ell}[\ell] = 0$}. 
\end{proof}

We now show that Spoiler needs to perform a split in the $k$-LB game in order to win. One might think that one could re-imagine the $k$-LB game as played directly on $F$. Unfortunately, this doesn't work, at least not in the naive way of defining things. The blockage comes from the multiple boards on each side. In particular, there can be an asymmetry between Spoiler playing on the $\LHS$/$\RHS$, which seems hard to capture when playing directly on $F$\confORfull{.}{, see Section~\ref{ss: XOR_fail} for a concrete example of this.} Instead, we proceed by reasoning directly about the structures. \confORfull{}{Still, we frequently deploy Corollary~\ref{cor: sumreduct} in our analysis.}

Consider Spoiler's winning strategy from Lemma~\ref{lem: spoiler-win-gen}. After Spoiler covers all the $s_i$ there must be some $\ell$ such that $\sum_{i\in [k] \setminus \{\ell\}} s_{i}[\ell]=1$, but Duplicator can choose for which value of $\ell$ this occurs. \confORfull{}{Then, depending on $\ell$, Spoiler has to move a different pebble in the next round.} Therefore, in the $k$-LB game Duplicator can ensure that all of these possibilities occur simultaneously, which stops Spoiler from winning without performing a split. We now formalise this intuition.

First, some preliminaries. We will describe a Duplicator strategy which ensures that, at the end of every round, all boards are partially isomorphic. Therefore, pebble $i$ always lies on the same colour element on every board. We denote this colour by $\pi(i)$. Then we define $\type(\lstrucXa): = \{\pi(i) \, \mid \, i \in \dom(\assA)\}$ \confORfull{}{and refer to this set as the \emph{type} of $\lstrucXa$}. Given $\lstrucXa, \lstrucYb$, $P \subseteq [k]$ we refer to $\sum_{i\in P} |\assA(i)| + |\assB(i)|$ as the sum of $P$ relative to $\lstrucXa, \lstrucYb$. The following definitions are crucial.

\confORfull{\begin{figure} 
\centering
\includegraphics[width=\linewidth]{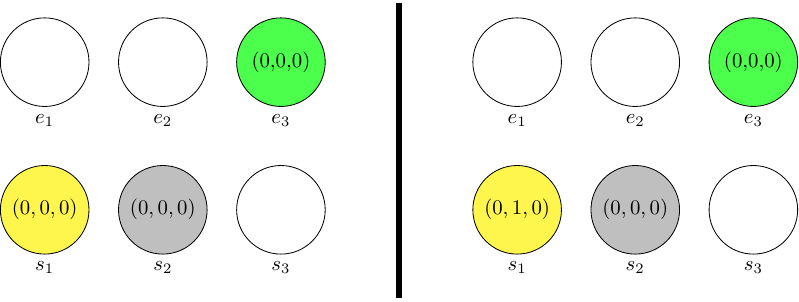}
\caption{A pair of good boards. Here we consider the case $k=3$. The colours represent the three different pebbles.} 
\Description[An example of a pair of good boards.]{An example of a pair of good boards.}
\label{fig: good}
\end{figure}}
{\begin{figure} 
\centering
\includegraphics[width=\textwidth]{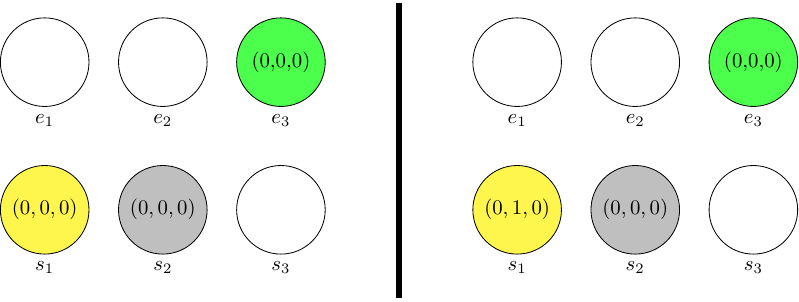}
\caption{A pair of good boards. Here we consider the case $k=3$. The colours represent the three different pebbles.  To see that the boards are good one just needs to check they are good for three.}
\label{fig: good}
\end{figure}}

\begin{definition} \label{def: good}
Let $\lstrucXa, \lstrucYb$ be partially isomorphic, $\{\strucX, \strucY\} = \{\strucA, \strucB\}$. Let $S := [k] \setminus {\ell}$. We say that $\lstrucXa$ is \emph{good} for $\ell$ relative to $\lstrucYb$ if for every set of pebbles $P=\{p_i \, \mid \, i \in S\}$ with $\pi(p_i) = s_i$, the sum of $P$ relative to $\lstrucXa, \lstrucYb$ has $\ell$th coordinate zero. If $\lstrucXa, \lstrucYb$ are good for every $\ell \in [k]$ we say that they are \emph{good} (relative to one another).
\end{definition}

By similar reasoning to that in the proof of Lemma~\ref{lem: spoiler-win-gen} it can be seen that if an element of every $R^{\strucX}_{s_i}$ is pebbled then there cannot be a board $\lstrucYb$ which is good relative to $\lstrucXa$. But crucially whenever such a situation occurs Duplicator can play such that there is a \emph{dual set} of boards.

\begin{definition} \label{def: dual}
Fix some $\ell \in [k]$. Let $T:= \{\lstruc{X}{\assA_{i}} \, \mid \, i \in [k] \setminus \{\ell\} \}$, be a set of $k$-pebbled structures all partially 
isomorphic to $\lstrucYb$, $\{\strucX, \strucY\} = \{\strucA, 
\strucB\}$. Moreover, suppose that $\type(\lstrucYb)= \{s_i \, \mid \, i\in [k]\}$. We say that $T$ is a 
\emph{dual set} (relative to $\lstrucYb$) if for every $j \neq i$, $\lstruc{X}{\assA_{i}}$ is good for $j$ relative to $\lstrucYb$.
\end{definition}

The following lemma shows that Spoiler has to perform a split in order to win.

\begin{lemma} \label{lem: dual}
In the $k$-LB game starting from position $[\{\strucA\}, \{\strucB\}]$, Duplicator may play so that until Spoiler performs a split of degree $k-1$ the following  properties hold at the end of every round.
\begin{enumerate}
\item Every board is partially isomorphic.
\item If the type of these boards is not $\{s_i \mid \, i\in [k]\}$, there is\confORfull{}{ exactly} one board on each side and these are good relative to each other. 
\item Otherwise, one side contains one board and the other side consists of a dual set relative to that board.
\end{enumerate} 
\end{lemma}

\begin{proof}
We prove this by induction on the number of rounds played, noting that the base case is trivial. So suppose (1)-(3) hold at the beginning of round $r$ and that Spoiler does \emph{not} perform a split of degree $k-1$ in this round. Then at least one side only contains one board, we assume it is the $\RHS$: the case where it is the $\LHS$ is symmetric. Let $\lstrucBb$ be the unique $\RHS$ board. We will analyse the situation where Spoiler moves  on the $\LHS$, the case where Spoiler moves on the $\RHS$ is easier. Let $p$ be the pebble moved in round $r$\confORfull{}{ and let $T := \{ \pi(i) \, \mid \, i \in [k] \setminus \{p\} \}$.}; the following crucial claim is almost immediate from the induction hypotheses; see Appendix~\ref{a: dual} for details. 

\begin{claim} \label{claim}
Suppose that \confORfull{$\{\pi(i)  \mid  i \in [k] \setminus \{p\}\} = \{s_i  \mid  i \in [k] \setminus \{\ell\}\}$}{$T =\{s_i  \mid  i \in [k] \setminus \{\ell\}\}$}, for some $\ell \in [k]$. Then when Spoiler pebbles there is some $\LHS$ board $\lstrucAa$ which is good for $\ell$ relative to $\lstrucBb$.
\end{claim}

\confORfull{\subfile{must_split_conf}}{%
\underline{Case 1:} $T = \{s_i \, \mid \, i\in [k] \setminus \{\ell\}\}$ \\
Take the board $\lstrucAa$ given to us by Claim~\ref{claim}. Duplicator does not need any other $\LHS$ board, so they delete them. While this is not formally part of the rules of the game, allowing Duplicator to perform deletions makes it harder for them to win and so does not affect the veracity of our eventual lower bound. Suppose Spoiler plays $p\to x$ on $\lstrucXa$. Note that by relabelling pebbles we may assume w.l.o.g.\! that $p = \ell$; we do this henceforth to reduce notational clutter.

\underline{Case 1a:} $\type(\lstruc{A}{\assA(p \to x)}) = \{s_i \, \mid \,i \in [k]\}$ \\
Recall, that as $\lstrucAa, \lstrucBb$ are good for $p$, the sum over $[k]\setminus \{p\}$ relative to these two boards has $p$th coordinate zero. First assume $p \neq 1$. Then we claim that the sum over $[k]\setminus \{p\}$ is in $ \mathsf{ODD}$. This follows as, by applying Corollary~\ref{cor: sumreduct}, the sum over $\{1\}$ is in $\mathsf{ODD}$ and the sum of every other pebble is in $\mathsf{EVEN}$. Define $y$ to be the element with $\pi(y) =s_p$ and $|y|=|x| + \sum_{i\in [k]\setminus \{p\}} |\assA(i)| + |\assB(i)|$. It follows that $|y|\in \mathsf{ODD}$ iff $|x|\in \mathsf{EVEN}$. Then for $j \in [k] \setminus \{p\}$ set $y_j$ to be the element on $s_{p}$ such that $|y_j|[i] = |y|[i]$ iff $i \neq j$. Then $|y_j| \in \mathsf{EVEN}$ iff $|x| \in \mathsf{EVEN}$, for every $j$ and $|y_j|[p] = |x|[p]$.

 We claim that $\{\lstruc{B}{\assB(p \to y_j)} \, \mid \, j \in [k] \setminus \{p\} \}$ is a dual set relative to $\lstruc{A}{\assA(p \to x)}$.  Firstly, it is clear from the previous discussion that this response satisfies the unary constraints on $s_{p}$. Secondly, recall that every board in the set is good for $p$ relative to $\lstruc{A}{\assA(p \to x)}$. Finally, fix some $i \not \in \{j, p\}$ and consider the $i$th coordinate of the sum of $S:= [k] \setminus \{i\}$ relative to $\lstruc{A}{\assA(p \to x)},\lstruc{B}{\assB(p \to y_j)}$, which, noting that the sum over $\{i\}$ of these two boards has $i$th coordinate zero, equals:
\confORfull{\begin{align*}
&\sum_{t=1}^k |\assA(t)|[i] + |\assB(t)|[i] = \\ &|y|[i] +|x|[i] + \sum_{\substack{t \in [k] \\ t \neq p}} |\assA(t)|[i] + |\assB(t)|[i]  
= |y|[i] + |y|[i]=  0
\end{align*}}
{\begin{align*}
\sum_{t=1}^k |\assA(t)|[i] + |\assB(t)|[i] =|y|[i] + \left(|x|[i] + \sum_{\substack{t \in [k] \\ t \neq p}} |\assA(t)|[i] + |\assB(t)|[i]\right) = |y|[i] + |y|[i]=  0
\end{align*}}

where $+$ denotes addition modulo 2. It follows that $\lstruc{B}{\assB(p \to y_j)}$ is good for $i$ relative to $\lstruc{A}{\assA(p \to x)}$ for all $i \neq j$. Therefore$ \{\lstruc{B}{\assB(p \to y_j)} \, \mid \, j \in [k] \setminus \{p\} \}$ is indeed a dual set

The case where $p =1$ is very similar. The difference is that 
here the sum over $[k] \setminus \{p\}$ is in $\mathsf{EVEN}$, so 
by replacing every $\mathsf{EVEN}$ with $\mathsf{ODD}$ and vice versa in the above argument we also get a dual set in this case. 
This works because $s_1 \in \mathsf{ODD}$ is 
a distinguishing constraint. Therefore the induction hypotheses 
are maintained.

\underline{Case 1b:} $\type(\lstruc{A}{\assA(p \to x)}) \neq \{s_i \, \mid \, i \in [k]\}$ \\
Again take the $\lstrucAa$ which is good for $p$ given to us by Claim~\ref{claim}. If $\pi(x) \neq e_{p}$ then there is no non-unary constraint $C$ with $\dom(C) \subseteq \type(\lstruc{A}{\assA(p \to x)})$. It is therefore easy to see that Duplicator has a response such that (2) holds at the end of every round. Otherwise $\pi(x) = e_{p}$ and Duplicator replies with $x$. The two boards in the resulting position are clearly good since $\lstrucAa$ is good for $p$ relative to $\lstrucBb$.

\underline{Case 2:} $T \neq \{s_i \, \mid \, i \in [k] \setminus \{\ell\}\}$ for any $\ell$ \\
In this case, the types of the boards prior to Spoiler's move cannot have 
been $\{s_i \, \mid \, i\in [k]\}$. So by the induction hypotheses, there is 
a unique $\lstrucAa \in \LHS$. Suppose $\type(\lstruc{A}{\assA(p \to x)}) 
= \{s_i \, \mid \, i \in [k] \setminus \{\ell\} \} \cup \{e_{\ell}\}$ for some $\ell$. Then 
$\pi(x) = s_t$ for some $t \in [k] \setminus \{\ell\}$. We need to show that Duplicator has a reply 
such that the resulting boards are good for $\ell$. Let $\delta$ be the $
\ell$th coordinate of the sum of $[k] \setminus \{p\}$. Then clearly any 
Duplicator reply $y$ with $|y|[\ell]= \delta + x[\ell]$ will be valid. By setting the 
other coordinates in any way such that $x+y$ satisfies the unary 
constraints on $s_t$ we get a valid Duplicator response. Otherwise, there 
is no non-unary constraint $C$ with $\dom(C) \subseteq \type(\lstruc{A}
{\assA(p \to x)})$. It is therefore easy to see that Spoiler can play such 
that there are good boards at the end of the round. The result follows.  %
}
\end{proof}

\confORfull
{\begin{figure} 
\centering
\includegraphics[width=\linewidth]{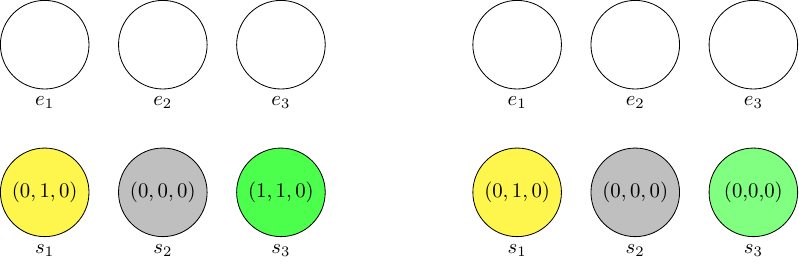}
\caption{If from the position in Figure~\ref{fig: good} Spoiler moves the green pebble to $s_3^{(0,0,0)}$ on the $\RHS$, then Duplicator can reply with $s_3^{1,1,0}$ and $s_3^{0,0,0}$; the $\RHS$ (above) is a dual set relative to the unique $\LHS$ board.}
\Description[An example of a dual set.]{An example of a dual set.}
\label{fig: dual}
\end{figure}}
{
\begin{figure} 
\centering
\includegraphics[width=\textwidth]{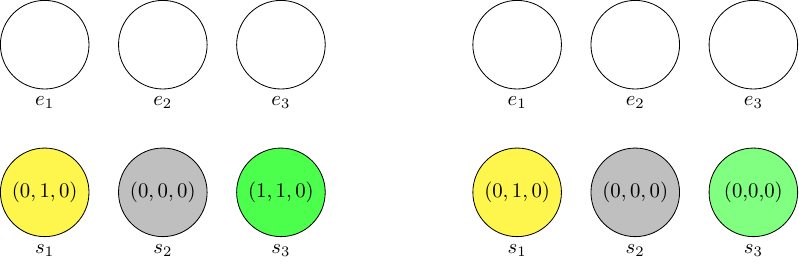}
\caption{Suppose that from the position in Figure~\ref{fig: good} Spoiler moves the green pebble to $s_3^{(0,0,0)}$ on the $\RHS$. Then Duplicator can make two copies of the $\RHS$ board and play $s_3^{1,1,0}$ on one and $s_3^{0,0,0}$ on the other. This is then a dual set relative to the $\LHS$ board and is depicted above.}
\label{fig: dual}
\end{figure}}

\subsubsection{Chaining structures together}

To achieve our lower bound we need to chain copies of the above
structures together. So let $F$ be the formula above. Define $V_i := \{x(i) \, \mid \,
x\in V\}$ and let $F_i$ be obtained from $F$ by replacing every $x\in V$ by $x(i)$. $F_i$ is then a set of constraints over $V_i$. 
We begin with $\bigcup_{i=1}^n F_i$ and edit it
to get a new formula $\mathcal{F}$ as follows. For every $i \neq 1$, we remove the clause  
$s_1(i) \in \mathsf{ODD}$. Similarly, for every $\ell$ and every $i \neq n$ we remove the  
clause $e_{\ell}(i)[\ell]=0$. The idea is that we only need 
these clauses for the start and end points respectively. Finally, for every $i \in [n-1],\, \ell \in [k]$, we add the clause $e_{\ell}(i) + s_1(i+1) \in \mathsf{EVEN}$. This provides the `link' between $V_i$ and $V_{i+1}$. 

Let $\strucC := \strucA(\mathcal{F})$ and $\strucD:= \strucB(\mathcal{F})$.  The removal of the $s_{\ell}$-clauses prevents Spoiler from winning the $k$-LB game too quickly on $[\strucC, \strucD]$. Previously these elements had to be mapped to $\underline{0}$, but now they may also be $\underline{0}_{\ell}$. This results in Spoiler having to pass through every $V_i$ to win. When they finally reach $V_n$ the original constraints are added back in, which gives Spoiler the ability to win. The lower bound comes from the fact that for every $V_i$ they pass through they must perform a split of degree $k-1$. 

\confORfull{It is easy to see that there is a sentence of $\klogic$ that $\strucC$ and $\strucD$ disagree on. The idea is that we apply the strategy from Lemma~\ref{lem: dual} on each $V_i$ in turn, see Appendix~\ref{a: chain_ub}.}{}

\begin{lemma} \label{lem: chain_ub}
There exists a sentence $\phi$ such that $\strucC \models \phi$ and $\strucD \models \neg \phi$ such that that the quantifier rank of $\phi$ is linear in $n$.
\end{lemma}

\begin{proof}
It suffices to analyse the $k$-formula game on $\mathcal{F}$. First, Spoiler carries out their winning strategy from Lemma~\ref{lem: dual} on $V_1$. To be precise if Spoiler plays $x$ in round $r$ of their winning strategy on $F$ they play $x(1)$ in round $r$ in the game on $\mathcal{F}$. 

Recall that in the final round of the game on $F$ Spoiler moves pebble $\ell$ to  $e_{\ell}$, for some $\ell \in [k]$, while leaving pebbles on $\{s_i \, \mid \, i \neq \ell\}$. Then Spoiler wins because the sum of the elements played on the pebbles other than $\ell$ has $\ell$th coordinate one, which means any Duplicator response must have $\ell$th coordinate one, contradicting the fact our constraints imply that $e_{\ell} = \underline{0}$. In our game on $\mathcal{F}$, our constraints instead imply that $e_{\ell}(1) \in \{\underline{0}, \underline{0}_{\ell}\}$ and so Duplicator has a unique valid response, $\underline{0}_{\ell}$.

Next Spoiler moves any pebble other than $\ell$ to  $s_1(2)$. Then because of the constraint $e_{\ell}(1) + s_1(2) \in \mathsf{EVEN}$, it follows that Duplicator must respond with some member of $\mathsf{ODD}$. The idea is that this replaces the constraint $s_1 \in \mathsf{ODD}$ that was deleted in the transition from $F$ to $\mathcal{F}$. Therefore, Spoiler can now carry out their winning strategy from Lemma~\ref{lem: dual} on $V_2$. Repeating this process $n$ times we see that Spoiler wins in a linear number of rounds.
\end{proof}

Our aim is now to lift Lemma~\ref{lem: dual} to $\mathcal{F}$. We need some notation. For $x(i) \in V_i$ we define $\pi_i(x(i)) = x$ and $\pi_i(x) = 
\emptyset$ for all $x\not \in V_i$. For an assignment $\assA$ on $
\strucA$ we define $\pi_i(\assA)$ to be the assignment $ p \to \pi_i(\assA(p))$, where we stipulate that $\pi_i(\assA(p)) =\emptyset$ means that $p$ is unassigned. Then for a 
board $\lstrucCa$ we define $\pi_i(\lstrucCa)$ to be $\lstruc{A}{\pi_i(\assA)}$. Call this the \emph{$i$th projection} of $\lstrucCa$. We define the corresponding notions for $\strucD$ boards analogously, but with $\strucB$ playing the role of $\strucA$.  Given a position $P$ in the $k$-LB game starting from position $ [\strucC, \strucD]$ we write $ \pi_i(P)$ to denote the position obtained by replacing every board with its $i$th projection.

With this notation in hand, we can lift the notion of a dual set to our new context. We say that a position $P$ in the $k$-LB game starting from position $[\strucC, \strucD]$ contains a \emph{dual set of degree $i$} if its $i$th projection contains a dual set. We analogously define good boards (for $\ell$) of degree $i$. As we now want to show that Spoiler must perform multiple splits of order $k-1$, we introduce the notion of the \emph{order} of a split. The higher the order the closer Spoiler is to winning when the split takes place.

\begin{definition} \label{def: order_split}
Let $\{\strucX, \strucY\} = \{\strucA, \strucB\}$. Suppose we are in a position $[\classX, \{\lstrucYb\}]$ or $[\{\lstrucYb\}, \classX\}]$ in the 
$k$-LB game starting from position $[\strucC, \strucD]$ where 
$\classX  =\{\lstruc{X}{\assA_i} \, \mid \, i\in [k] \setminus \{\ell\}\}$ is a dual set of order $i$  relative to $
\lstrucYb$. Then if Spoiler performs a split of degree $k-1$ we say that this split has order $i$. 
\end{definition}

Our lower bound is achieved via the following proposition, which is a lifting of Lemma~\ref{lem: dual} to the game on $[\strucC, \strucD]$.

\begin{proposition} \label{prop: lb_technical}
Duplicator may play such that for all $2 \le j \le n$ if after $r$ rounds Spoiler has performed a split of degree $j-1$ but not one of degree $j$  then the following properties hold.
\begin{enumerate}
\item Every board is partially isomorphic.
\item If the type of these boards is not $\{s_i(j) \mid \, i\in [k]\}$, there is exactly one board on each side and these are good of order $j$ relative to each other. Furthermore, whenever $\pi(p) \not \in \{s_i(j) \mid \, i\in [k]\} \cup V_{j-1}$ then the boards agree on the image of $p$.
\item Otherwise, one side contains one board and the other side consists of a dual set of order $j$ relative to that board.
\end{enumerate} 
Moreover, if Spoiler has not yet performed a split of degree one then (1)-(3) also hold with $j=1$.
\end{proposition}

A crucial point is to ensure that Spoiler never wins by `going backwards'. More precisely we shown that once a split of degree $j$ has taken place it never helps Spoiler to play on $\bigcup_{i=1}^{j-1} V_i$. The following notion will be useful. We say the $i$th restriction of $\mathcal{F}$, which we denote by $\mathcal{F}$, consists exactly of those clauses involving only variables in $\bigcup_{j=1}^i V_i$. We then define the $i$th restriction of $\strucC$, to be $\strucC_i := \strucA(\mathcal{F}_i)$. Similarly, we define $\strucD_i:= \strucB(\mathcal{F}_i).$ We have the following technical lemma.

\begin{lemma} \label{lem: isom}
Let $\lstrucCa, \lstrucDb$ be partially isomorphic pebbled structures such that for some $j \in [n]$, $\ell \in[k]$ their type is $\{s_i(j) \, \mid \, i\in [k] \setminus \{\ell\} \}$. Moreover, suppose that $\lstrucCa$ is \emph{not} good relative to $\lstrucDb$. Then there is an isomorphism $\sigma$ from $\strucC_j$ to $\strucD_j$ such that for every $p\in \dom(\assA)$ $\sigma(\assA(p)) = \assB(p)$ and $\hat{\sigma}(e_{\ell}(j-1)) \in \mathsf{ODD}$ for all $\ell \in [k]$. We say such an isomorphism \emph{respects} $\lstrucCa,\, \lstrucDb$ up to order $j$.
\end{lemma}

\begin{proof}
We prove this by induction on $j$. For the base case $j=1$ it is enough---by Lemma~\ref{lem: passignment}---to give a satisfying assignment $\gamma^1 : V_1 \to G$ for $\mathcal{F}_1$, such that for $i \in \dom(\assA)$, if $\pi(i) = x$ then $\gamma^1(x) = |\assA(i)| + |\assB(i)|$. 

Since by assumption the type of both boards is $\{s_i(1) \, \mid \, i\in [k] \setminus \{\ell\} \}$ the values of $\gamma^1$ at these variables are fixed. We then define $\gamma^1(s_{\ell}(1)) := \underline{1} + \sum_{i \in [k] \setminus \{\ell\}} \gamma^1(s_i(1))$. We set $\gamma^1(e_i(1)) = \underline{0}_i$ for every $i\in [k]$ and claim this is a satisfying assignment. 

To see this first observe that since $\lstrucCa$ is not good for $\ell$ relative to $
\lstrucDb$ it follows that $\sum_{i \in [k] \setminus 
\{\ell\}} \gamma^1(s_i(1))$ has $\ell$th coordinate one, so $\gamma^1(s_{\ell}(1))[\ell] =0$. Since $
\lstrucCa$ and $\lstrucDb$ are partially isomorphic we also 
know, via Corollary~\ref{cor: sumreduct}, that $\gamma^1(s_i(1)) \in \mathsf{EVEN}$ for all $i \neq 1, \ell$ 
and $\gamma^1(s_1(1)) \in \mathsf{ODD}$, if $\ell \neq 1$. Therefore, in the case $\ell \neq 
1$ we know that $\sum_{i \in [k] \setminus \{\ell\}} 
\gamma^1(s_i(1)) \in \mathsf{ODD}$ and so by definition $\gamma^1(s_{\ell}(1)) \in \mathsf{EVEN}$. Here we use the fact that $k$ is odd. Similarly, if $\ell=1$ we have that $\gamma^1(s_{1}(1)) \in \mathsf{ODD}$, since $\sum_{i \in [k] \setminus \{1\}} \gamma^1(s_i(1)) \in \mathsf{EVEN}$ and as $k$ is odd. It follows that $\gamma^1$ satisfies every unary constraint in $\mathcal{F}_1$. Moreover, for every $t \in [k]$ we obtain
\begin{align*}
&\gamma^1(e_t^1(1))[t] + \sum_{\substack{i \in [k] \\ i \neq t}} \gamma^1(s_i^1(1))[t] = \\
 &1 + \sum_{i \in [k]} \gamma^1(s_i(1))[t] =  \gamma^1(s_{\ell}(1))[t] +  1 +  \sum_{\substack{i \in [k] \\ i \neq \ell}} \gamma^1(s_i(1))[t] = 0.
\end{align*}

Therefore $\gamma^1$ satisfies each clause $(e_{t}(1) +\sum_{i \in [k]\setminus \{t\}} s_i)[t] =0$ and therefore $\mathcal{F}_1$. 

For the induction step we suppose we have some $\gamma^{j-1}: V_{j-1} \to G$ which satisfies $\mathcal{F}_{j-1}$ and such that 
$\gamma(e_{i}(j-1)) \in \mathsf{ODD}$ for all $i$. Note, that this second invariant is satisfied by the base case. We want to define $\gamma^j$, the argument is almost the same as the base case. We set $\gamma^j(x) = \gamma^{j-1}(x)$ for all $x \in V_{j-1}$. Now by assumption the type of our boards is $\{s_i(j) \, \mid \, i\in [k] \setminus \{\ell\} \}$ and so the value of $\gamma^j$ is fixed on these variables. We then define $\gamma^1(s_{\ell}(j)):= \underline{1}+ \sum_{i \in [k] \setminus \{\ell\}} \gamma(s_i(j))$ and set $\gamma^j(e_i(j)) = \underline{0}_i$ for every $i\in [k]$. We claim this is a satisfying assignment. The check is the same as in the base case except for two things. Firstly, we do not have a constraint $s_1(j) \in \mathsf{ODD}$. However if $\ell \neq 1$ we can infer that $\gamma^1(s_1(j)) \in \mathsf{ODD}$ as:
\begin{enumerate}
\item $\gamma^{j}(e_t(j-1))=\gamma^{j-1}(e_t(j-1)) \in \mathsf{ODD}$,
\item we have a constraint $e_t(j-1) + s_1(j) \in \mathsf{EVEN}$ and 
\item  $\lstrucXa$ and $\lstrucYb$ are partially isomorphic.
\end{enumerate} 
Moreover, if $\ell=1$ since we still have all constraints of the form $s_i(j) \in \mathsf{EVEN}$, for $i\neq 1$, by the same argument as the base case $\gamma^{j}(s_1(j)) \in \mathsf{ODD}$. Note, that this implies that $\gamma^j$ satisfies the constraint $e_t(j-1) + s_1(j) \in \mathsf{EVEN}$ for every $t$. Secondly, we now have that the induction hypothesis guarantees that all constraints involving only elements from $V_{j-1}$ are satisfied by $\gamma^j$. The result follows
\end{proof}

Now we are in a position to prove Proposition~\ref{prop: lb_technical}.

\begin{proof}[Proof of Proposition~\ref{prop: lb_technical}]
We proceed by induction. For the base case we need to give Duplicator's strategy before a split of order one occurs. This will be to copy Spoiler's move outside of $V_1$ and on $V_1$ to use the strategy from Lemma~\ref{lem: dual}.

More formally we assume inductively that, if Spoiler has not performed a split of order one, at the beginning of round $r$, (1)-(3) from the statement of the proposition hold for $j=1$. Note these properties hold trivially at the beginning of the first round. So suppose (1)-(3) hold at the beginning of round $r$ and that Spoiler moves pebble $p$ in this round. As in the proof of Lemma~\ref{lem: dual}, it is easy to see that if we remove pebble $p$ from every board, then there exists a pair of boards lying on opposite sides that are good relative to each other. Duplicator deletes every other board. If Spoiler plays $x\not \in V_1$, then Duplicator replies with $x$ on the only remaining board on the other side. Otherwise, before replying, Duplicator looks at the first projection of the resulting position. They then look at what response(s) they would make in the projected position according to the strategy set out in Lemma~\ref{lem: dual}. If they would make the response $p \to x$ on $\pi_1(\lstrucYb)$ then they play $p \to x(1)$ on $\lstrucYb$. 

It is easy to see that if Duplicator plays in this way at the end of each round no boards are deleted. For by Lemma~\ref{lem: dual}, no constraints involving only elements of $V_1$ are violated and since Duplicator copies Spoiler moves outside of $V_1$ all other constraints are respected. Note that in particular, Spoiler always matches moves played on each $e_{\ell}(1)$ and so (2) is maintained and the constraints $e_{\ell}(1) + s_1(2) \in \mathsf{EVEN}$ are respected.

The induction step is similar but there is one extra thing to take care of: we need to ensure Spoiler cannot win by going backwards. For this we use Lemma~\ref{lem: isom}.

In detail suppose that in the last round Spoiler performed a split of order $j-1$. Then before the split, there was a dual set of order $j$ which was separated in the split. Duplicator is happy to continue from any of the resulting positions and so the split has degree $k-1$. Therefore, the current position is $[\{\lstrucCa\}, \{\lstrucDb\}]$, where $\type(\lstrucCa) = \{s_i(j-1) \, \mid \, i \in [k]\}$ and where there is a unique element $\ell \in [k]$ such that $\pi_{j-1}(\lstrucCa)$ is \emph{not} good for $\ell$ relative to $\pi_{j-1}(\lstrucDb)$. Suppose Spoiler moves some pebble $p$ that does not cover $s_{\ell}(j-1)$. Then the position with pebble $p$ removed consists of two boards which are good of order $j-1$. Therefore, Duplicator can revert to their strategy from before the split and Spoiler will have to perform another split of order $j-1$ in order to win. 

Otherwise, Spoiler moves the pebble $q$ covering $s_{\ell}(j-1)$. Let $\lstrucCc$ (resp. $\lstrucDd$) be the board obtained from $\lstrucCa$ (resp. $\lstrucDb$) by removing pebble $q$. Duplicator's strategy from this point is then as follows. First, they find an isomorphism $\sigma$ respecting $\lstrucCc, \lstrucDd$ up to order $j-1$, the existence of which is guaranteed by Lemma~\ref{lem: isom}. If Spoiler plays on $V_i$ for $i>j$, Duplicator will match the move. On $V_j$ they will use the strategy from Lemma~\ref{lem: dual}. On $V_i$, for $i<j$ they use $\sigma$. To be precise if Spoiler plays $g$ on $x$ for $x\in V_i$, then Duplicator plays $\sigma(g)$ on $x$. Then we show inductively that, until Spoiler performs a split of order $j$, (1)-(3) from the statement of the proposition hold.

Note that these properties hold immediately after the split of order $j-1$. So suppose Spoiler plays $x \in V_i$, on some board. The case where $i>j$, is the same as the base case. The case where $j=i$ is also very similar. The only difference occurs when Spoiler plays on $s_1(j)$ and some $e_{\ell}(j-1)$ is already pebbled. But since by Lemma~\ref{lem: isom} we know that if a pebble $p$ projects to $e_{\ell}(j-1)$, then the sum over $p$, relative to the only two boards, is in $\mathsf{ODD}$; this ensures everything else goes through as before. The case $i<j$ is easy noting that since $\sigma$ is an isomorphism between $\strucC_{j-1}$ and $\strucD_{j-1}$ no constraint involving only elements from $V_{j-1}$ is violated.
\end{proof}

Our result now follows

\begin{corollary}\confORfull{[Proposition~\ref{prop: lb} restated]}{\label{prop: lb}}
Duplicator can play so that they are guaranteed to get more than $(k-1)^n$ points in the $k$-LB starting from position $[\{\strucC\}, \{\strucD\}]$.
\end{corollary}

\begin{proof}
By Proposition~\ref{prop: lb_technical} in order for Spoiler to win from position $[\{\strucC\}, \{\strucD\}]$ they must perform a split of degree $n$. Further, the proposition shows that if Spoiler performs a split of degree $n$ they have already performed a split of degree $j$ for all $j<n$. After the split of degree $n$ Spoiler requires one more round to win and Duplicator receives $(k-1)^n$ points for this round, so the result follows.
\end{proof}
Note that the only time we used the oddness of $k$ was in the proof of Lemma~\ref{lem: isom}. It is in fact not difficult (although a little combinatorially involved) to prove Proposition~\ref{prop: lb} (for all $k \ge 3$) without invoking Lemma~\ref{lem: isom}. We just need to explicitly give a Duplicator strategy for the case where Spoiler tries to `go backwards.' Alternatively, we can use a different formula for the even case and then the lower bound follows by almost the same argument as the above, see Appendix~\ref{a: even}. 

From this Theorem~\ref{thm: genlb} is immediate via Lemma~\ref{lem: lbgame}\confORfull{.}{, noting that $|A|= |B|= 2k \cdot n$ and that we can `pad out' either structure by adding a constant number of isolated elements without affecting the lower bound.}

\subsection{\texorpdfstring{$\mathcal{L}^{k+1}$}{L(k+1)} is exponentially more succinct than \texorpdfstring{$\klogic$}{L(k+1)}}

As a consequence of the above, it is relatively easy to prove Theorem~\ref{thm: suc}. The argument goes as follows. First, we show that there is a sentence $\phi$ using $k+1$ variables and of size $O(C)$ separating $\strucC$ and $\strucD$. Then by Theorem~\ref{thm: genlb} we know that any equivalent sentence to $\phi$ using $k$ variables must be exponentially longer than $\phi$. It therefore suffices to show that we can rewrite $\phi$ to an equivalent sentence using one fewer variable; this can be done by a simple recursive procedure. 

Fix an odd $k\ge 3$: the even case is similar using the structures defined in Appendix~\ref{a: even}. Let $\strucA, \strucB, \strucC, \strucD$ be \confORfull{the structures defined in Section~\ref{s: lb} from a formula over $G = (\mathbb{Z}_2^{k},+)$.}{as above.}  

We first give a sentence of $\mathcal{L}^{k+1}$ which separates $\strucA$ from $\strucB$. We will often want to say that a variable in our formula plays the role of some $v \in V := \{s_1, \dots, s_k, e_1, \dots, e_k\}$. To this end for $ v\in V$, define $U_v$ to be the set of unary relations associated with $v$, i.e. $U_v$ contains every relation symbol $R_C$ where $C = (v, H)$ for some $H$ as well as the colour associated with $v$. We define $\good_v(x) := \bigwedge_{R \in U_v} R(x)$. Let 
$R_i$ be the relation corresponding to the constraint $(e_{\ell} +\sum_{i \in [k]\setminus \{\ell\}} s_i)[\ell] =0$. Then consider the following sentence.
\begin{align*}
\psi \equiv& \exists x_1 \dots \exists x_k \forall x_{k+1} 
\Big[ \bigwedge_{i=1}^k \good_{s_i}(x_i) \wedge \\ 
 &\bigwedge_{i=1}^k  \left( \good_{e_i}(x_{k+1}) \rightarrow 
R_i(x_1, x_2, \dots, x_{i-1}, x_{i+1}, x_{i+2}, \dots, x_{k+1}) \right) \Big]
\end{align*}
The sentence corresponds to the following Spoiler strategy in the $(k+1)$-pebble game. In the first $k$ rounds they place pebble $i$ on $s_i^{\underline{0}}$ for every $i\in [k]$. This corresponds to the first $k$ existential quantifiers. Then Duplicator has to make `good' responses to these moves. For example, in the first round they must select an element of the form $s_1^g$ such that $g \in \mathsf{ODD}$ and $g[1] =0$. Suppose the Duplicator moves were to play $s_i^{g_i}$ in round $i$. We have already seen in Section~\ref{s: lb} that there must be some $i \in [k]$ such that $\sum_{j \in [k] \setminus \{i\}} g_j[i] =1$. Since the next quantifier is $\forall x_{k+1}$ in round $k+1$ Spoiler moves pebble $k+1$ on the $\RHS$. They place this pebble on $e_i^{\underline{0}}$;  the point is that Spoiler chooses a witness to the fact that $\strucB$ does \emph{not} model $\psi$. Then the only response Duplicator can make to respect the unary relation on $e_i$ is to also play $e_j^{\underline{0}}$. But then the $\LHS$ board models $R_i$ but the $\RHS$ board does not. Therefore, Spoiler wins after $k+1$ rounds and so $\strucA$ and $\strucB$ disagree on $\psi$.

Now let us rewrite this as a sentence using only $k$ variables. This is easy just using the distribution of universal quantification over disjunction and by renaming variables, concretely we obtain the formula $f(\psi) \equiv$
\begin{align*}
& \exists x_1 \dots \exists x_k 
\Big[ \bigwedge_{i=1}^k  \good_{s_i}(x_i) \wedge \\  
 &\bigwedge_{i=1}^k \forall x_i \left(\good_{e_i}(x_{i}) \rightarrow R_i(x_1, x_2 \dots, x_{i-1}, x_{i+1}, x_{i+2}, \dots, x_k, x_{i}) \right) \Big]
\end{align*}
The same idea works for $\strucC$ and $\strucD$; we just have to do things recursively. Here we use the $\good_v$ notation in the same way 
and use $R_i^j$ to correspond to the constraint $(e_{\ell}(j) +\sum_{i 
\in [k]\setminus \{\ell\}} s_i(j)))[\ell] =0$. We now also have 
constraints of the form $e_{\ell}(j) + s_1(j+1) \in \mathsf{EVEN}$, 
for which we use the relation symbol $E_{\ell}^j$. Define the formula 
$E^j(x,y) \equiv \bigvee_{\ell \in [k]} E_{\ell}^j (x,y).$  Now for $1\le j\le n-1$, set $
\phi^j(x_1)$ to be the formula:
\begin{align*}
&\exists x_2 \dots \exists x_k \forall x_{k+1} \Big[ \bigwedge_{i=2}^k  \good_{s_i(j)}(x_i) \wedge \\
 &\left( \bigvee_{i=1}^k  \good_{e_i(j)}(x_{k+1}) \wedge
R_i^j(x_1, x_2, \dots, x_{i-1}, x_{i+1}, x_{i+2}, \dots, x_{k+1}) \right) 
\rightarrow  \\
&\exists x_1 \Big(( \good_{s_1(j+1)}(x_1) \wedge E^j(x_{k+1}, x_1)  \wedge 
\phi^{j+1}(x_1)\Big)\Big]
\end{align*}

Let $\phi^n(x_1)$ be $\psi$ without the first existential quantification and with every mention of some $v\in V$ replaced by $v(n)$. Define $\phi \equiv \exists x_1 (\good_{s_1(1)}(x_1) \wedge \phi^1(x_1))$. 

Similarly, to the above, it is easy to verify that $\strucC$ and $\strucD$ disagree on this sentence: to show this we outline the corresponding $(k+1)$-pebble game strategy. In the first $k$ rounds, Spoiler plays analogously to the above strategy on $\strucA, \strucB$: they play $s_i(1)^{\underline{0}}$ in round $i$ on the $\LHS$. Suppose that Duplicator responds with $s_i(1)^{g_i}$ in round $i$. Again, there must be some $i \in [k]$ such that $\sum_{j \in [k] \setminus \{i\}} g_j[i] =1$. In round $k+1$ Spoiler then plays on the $\RHS$ and chooses $\underline{0}_i$ on $e_i(1)$. Therefore the $\RHS$ board models 
\[
\good_{e_i(1)}(x_{k+1}) \wedge
R_i^1(x_1, x_2, \dots, x_{i-1}, x_{i+1}, x_{i+2}, \dots, x_{k+1})
\]

Duplicator must then make a response so that the $\LHS$ also models this formula: the only way to do this is to play $\underline{0}$ on $e_i(1)$. In the next round Spoiler then plays on the $\LHS$ and moves pebble one (reflecting the quantifier $\exists x_1$) to $s_1(2)^{\underline{0}}$; Duplicator must reply with some element $y$ of the same colour such that $|y| \in \mathsf{ODD}$. By repeating essentially the same process on each $V(i)$, $i \in [n-1]$ in turn Spoiler will eventually reach a position where on the $\LHS$, $s_1(n)^{\underline{0}}$ is pebbled and on the $\RHS$ the corresponding pebble lies on $s_1(n)^g$ for some $g\in \mathsf{ODD}$. Then, similarly to the above, by following the formula $\phi^n$ from this position Spoiler wins.

We now need to rewrite this sentence using only $k$ variables. To do this 
we will use the fact $(\bigvee_{i=1}^k \psi_i) \rightarrow \psi$ is 
equivalent to $\bigwedge_{i=1}^k (\psi_i \rightarrow \psi)$ as well as the 
distribution of universal quantification over conjunction. Using this it 
follows that for $1\le j \le n-1$, $\phi^j(x_1)$ is equivalent to $
\theta^j(x_1)$ defined as 
\begin{align*}
&\exists x_2 \dots \exists x_k  \Bigg[ \bigwedge_{i=2}^k \good_{s_i(j)}(x_i) \wedge \bigwedge_{i=1}^k \forall x_{k+1} \bigg( \\
 &  \Big( \good_{e_i}(x_{k+1}) \wedge
 R_i(x_1, x_2, \dots, x_{i-1}, x_{i+1}, x_{i+2}, \dots, x_{k+1}) \Big) \rightarrow  \\
&\exists x_1 \left(  \good_{s_1(j+1)}(x_1)  \wedge E^j(x_{k+1}, x_1)  \wedge \theta^{j+1}(x_1) \right)  \bigg) \Bigg].
\end{align*}
Where we set $\theta^n$ to be equal to $f(\psi)$ without the first existential quantification and with every mention of some $v\in V$ replaced by $v(n)$. Then by renaming variables, it is easy to see that this can be rewritten as a formula only using variables $x_1, \dots, x_k$. To see this note that the first big conjunction in the definition of $\theta^j(x_1)$ does not use $x_{k+1}$ and, for $i \neq 1$ the $i$th conjunct from the second big conjunction is equivalent to:
\begin{align*}
&\forall x_{i} \bigg(  \Big( \good_{e_i}(x_{i}) \wedge
 R_i(x_1, x_2, \dots, x_{i-1}, x_{i+1}, x_{i+2}, \dots, x_{i}) \Big) \rightarrow  \\
&\exists x_1 \left(  \good_{s_1(j+1)}(x_1)  \wedge E^j(x_{i}, x_1)  \wedge \theta^{j+1}(x_1) \right)  \bigg).
\end{align*}
The case $i=1$ is a little messier since the variable $x_1$ already appears in this conjunct. But this only occurs in the last line, so we can first replace all occurrences of $x_1$ with $x_2$ and then replace $x_{k+1}$ with $x_1$. We then need to make appropriate adjustments deeper in the recursion but this just amounts to swapping every $x_1$ to a $x_2$ and vice versa. Therefore, we can obtain a sentence $\theta \in \klogic$ equivalent to $\phi$. 

Since $\strucC$ and $\strucD$ disagree on $\phi$ and $|\phi| = O(n)$ it follows by Proposition~\ref{prop: lb} that any  equivalent sentence in $\klogic$ must have size exponential in $|\phi|$ and so Theorem~\ref{thm: suc} follows.

\section{An Existential Lower Bound} \label{s: exlb}

We now present a construction which allows us to lift quantifier depth lower bounds to quantifier number lower bound for the logic $\exklogic$. The idea
 is to leverage the substantial literature on quantifier depth as a `black box' to yield results concerning quantifier number. 

 The construction begins with two relational structures $\strucA$ and $\strucB$ over a common signature $\sigma$ that 
can be distinguished in $\exklogic$ but which require a sentence of quantifier depth at least 
$r$ to do so. We now form two new structures $S(\strucA)$ and $S(\strucB)$. These structures 
take as their core $
\strucA$ and $\strucB$ respectively which we then expand with new elements and relations. 

The idea is that in order to win the existential lower bound game from position $[\{\strucSA\}, \{\strucSB\}]$, Spoiler has to `essentially' carry 
out their winning strategy from the existential pebble game on $\strucA, \strucB$. However, carrying out one step of this winning strategy is relatively expensive for Spoiler 
on our new structures, in that they are forced to perform a split of the game into two parts 
which are both equally hard. This is ensured by adding two things to our structure. First, we 
add dummy elements (the set $D$ below). If Spoiler simply tries to carry out their $\pexists k$-pebble 
strategy and plays some $a\in A$, Duplicator will reply with a dummy element, thwarting Spoiler's
progress. Of course, we need to give Duplicator some way to win so we also introduce a `split 
gadget'. The split gadget operates similarly to structures built from a simple XOR formula. Effectively there are two possible paths Spoiler can take to make progress and on each $\RHS$ board Duplicator can stop progress on one of these paths. 
 Duplicator must then perform a split to separate these two types of $\RHS$ board\confORfull{.}{, otherwise they cannot make progress.} We also give Spoiler an extra pebble to enable them to navigate the 
split gadgets while maintaining the pebbles on $A$. The upshot is that for every move in Spoiler's winning strategy for the $k$-pebble game on $[\strucA, \strucB]$, they must perform a split in the $\pexists(k+1)$-LB game on $[\{\strucSA\}, \{\strucSB\}]$. 

\subsection{Definitions of the Structures}

We impose two technical conditions on $\strucA$ and $\strucB$. Firstly, we assume that for every $R \in \sigma$, and $\strucX \in \{ \strucA, \strucB\}$, $R^{\strucX}$ has \emph{disjoint positions}, i.e., for every tuple $(x_1, \dots, x_{m}) \in R^{\strucX}$, $x_i \neq x_j$ for $i\neq j$. Secondly, we assume that there are $k$ isolated elements of each structure, that is elements $a(1), \dots, a(k)$ in $A$ and $b(1), \dots, b(k)$ in $B$, such that each $a(i)$ does not appear as a coordinate of any tuple of any relation in $\strucA$ and each $b(i)$ does not appear as a coordinate of any tuple of any relation in $\strucB$. \confORfull{}{If such elements do not exist we simply add them in. These elements will be crucial in allowing Spoiler to win, as we will see.}

Let $\strucX \in \{\strucA, \strucB \}$. Then within $S(X)$ we have a copy of $\strucX$. To be more precise the signature of $\strucSX$ is an expansion of $\sigma$---we denote it by $S(\sigma)$---and for every $R \in \sigma$, $R^{S(\strucA)} = R^{\strucA}$ and $R^{S(\strucB)} \subset R^{\strucB}$. We also have dummy elements in both structures, $D:= \{d_i \, \mid \, i \in [k+1]\}$; these occur in the additional tuples we add into $R^{S(\strucB)}$:
\[
R^{S(\strucB)} = R^{\strucB}  \cup \{(v_1, \dots, v_{m}) \in (B \cup D)^m \confORfull{}{\,} \mid \confORfull{}{\,}  v_i \in D,\textit{ for some $i \in [m]$} \}.
\]
where $m$ is the arity of $R$. If our signature only consisted of $\sigma$ then Duplicator could reply to every Spoiler move with elements of $D$ and thus Spoiler would be unable to win. But our full structure allows Spoiler to avoid this problem; it will allow Spoiler to make elements of $A$ \emph{active} in the following sense.

\begin{definition} \label{def: active}
We say a pebble $i$ is \emph{active} if 
\begin{enumerate}
\item there is some $\RHS$ board (i.e. Spoiler has not already won),
\item on the unique $\LHS$ board $\lstrucSAa$, $\assA(i) \in A$ and
\item on every $\RHS$ board $\lstrucSBb$, $\assB(i) \in B$.
\end{enumerate}
 We say an element $a \in A$ is active if it is covered by an active pebble.
\end{definition}

\subsubsection{The Start Gadget} \label{sss: start}

Spoiler, in the main, makes pebbles active via the \emph{split gadget}. However,\confORfull{}{ for reasons we will come to,} Spoiler has to fulfil certain initial conditions in order to use the split gadget. To enable Spoiler to fulfil these conditions the structures include a \emph{start gadget}. \confORfull{}{This is only necessary to ensure Spoiler can win: if we start from a position where the initial conditions for using the split gadget are fulfilled then the start gadget will never be used.}

It is the start gadget which uses the $k$ independent elements in each structure. These will be the first elements which become active. So let 
\confORfull{
\begin{align*}
G := \{g_0^0, g_0^1, g_1^0, g_1^1 \} &\cup \{l_i^a \, \mid \, i \in [k], a\in \{0,1\}\} \\ &\cup \{r_i^a \, \mid \, i \in [k], a\in \{0,1\}\} \subset S(\strucX).
\end{align*}
}
{\[G := \{g_0^0, g_0^1, g_1^0, g_1^1 \} \cup \{l_i^a \, \mid \, i \in [k], a\in \{0,1\}\} \cup \{r_i^a \, \mid \, i \in [k], a\in \{0,1\}\} \subset S(\strucX).
\]}
Then in both $\strucSA$ and $\strucSB$, for each element $x_i$, with $x \in \{g, l, r\}$ we give a unique colour (i.e. unary relation) to the set $\{x_i^0, x_i^1\}$. We also introduce a binary relation $E$ and a $
(k+1)$-ary relation $R_s$\confORfull{.}{, see Figure~\ref{fig: start}.}
\confORfull{
\begin{figure} 
\centering
\includegraphics[width=\linewidth]{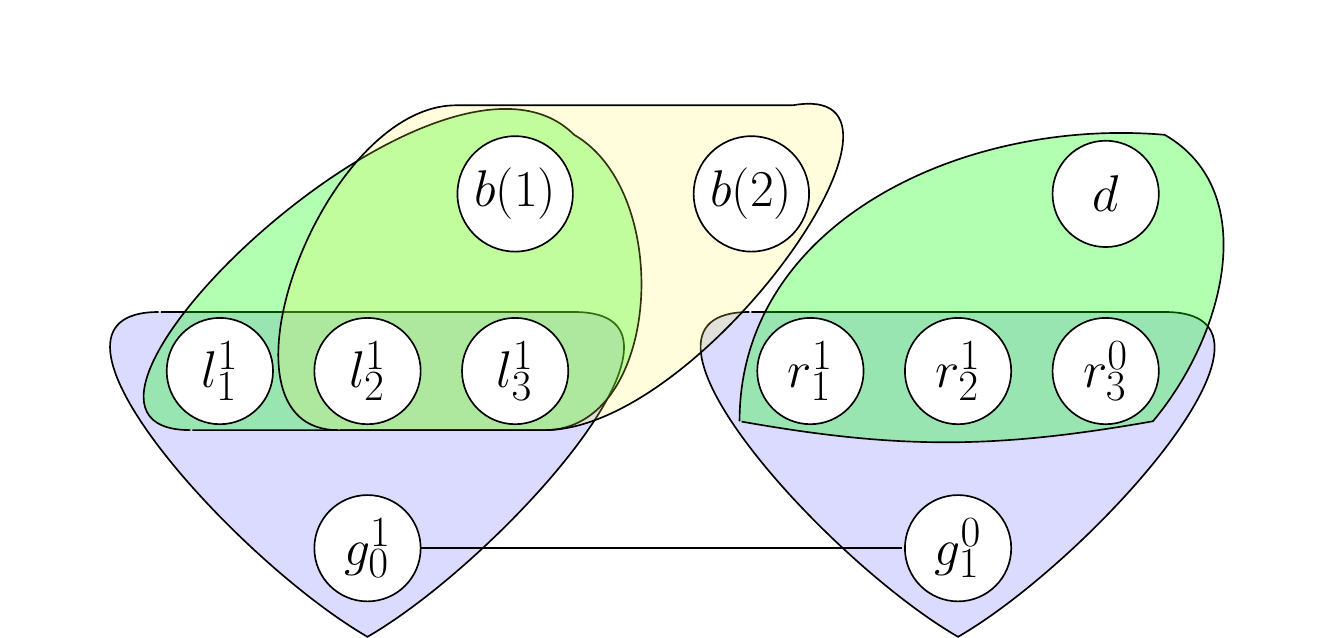}
\caption{A part of the start gadget in $\strucSB$, for $k=3$. The edge represents a tuple in $E^{\strucSB}$ and each hyperedge corresponds to a tuple in $R_s^{\strucSB}$. The element $d$ is a member of $D$. Note that many tuples and elements are excluded from this picture.}
\Description[A visualisation of the start gadget.]{A visualisation of the start gadget.}
\label{fig: start}
\end{figure}
}{
\begin{figure} 
\centering
\includegraphics[width=0.8\textwidth]{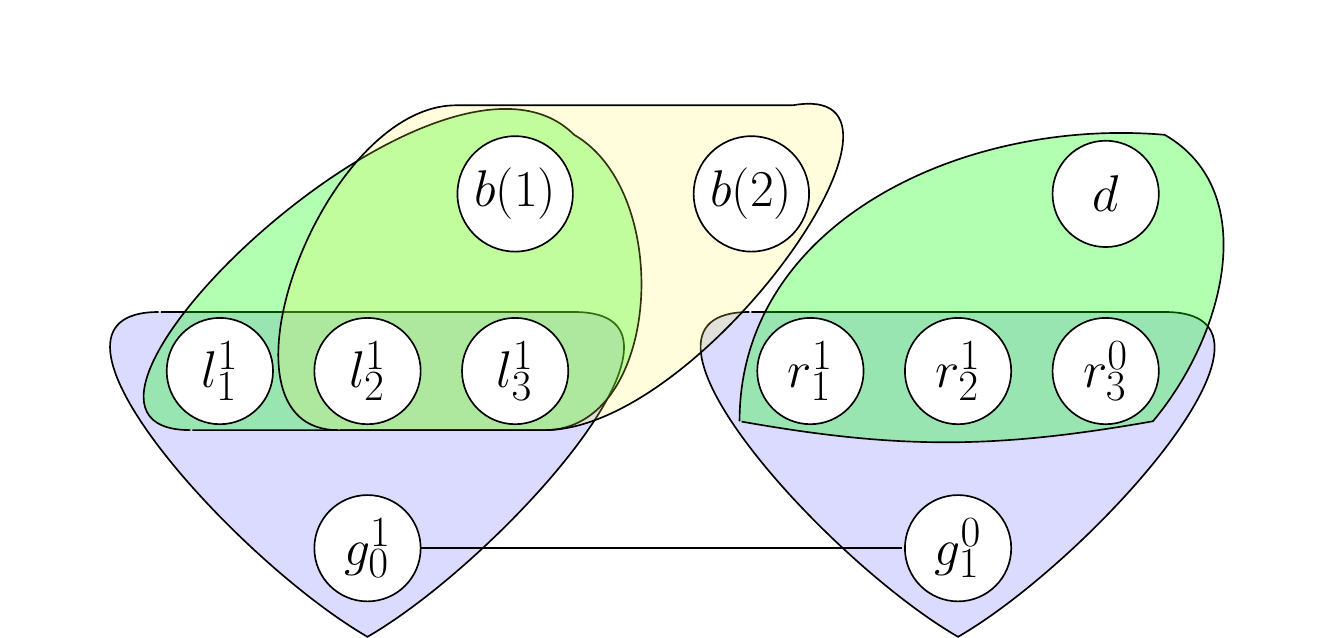}
\caption{A part of the start gadget in $\strucSB$, for $k=3$. The edge represents a tuple in $E^{\strucSB}$ and each hyperedge corresponds to a tuple in $R_s^{\strucSB}$. The element $d$ is a member of $D$. Note that many tuples and elements are excluded from this picture.}
\label{fig: start}
\end{figure}
}

We define 
$E^{\strucSA} := \{(g_0^0, g_{1}^0)\}$ and 
$E^{\strucSB} := \{(g_0^0, g_{1}^1), (g_0^1, g_{1}^0)\}$. This effectively enforces the XOR constraints $g_0 + g_1 =1$. 

$R_s^{\strucSA}$ consists of $(g_0^0, l_1^0, \dots, l_{k}^0)$,
$(g_1^0, r_1^0, \dots, r_k^0)$ and for $1 \le  i < k$, $y\in \{l,r\}$ the tuples $(y_i^0, \dots, y_{k}^0, a(1), \dots, a(i))$.

The tuples in $R_s^{\strucSB}$ are of two types. The first type mirrors those in $R_s^{\strucSA}$. To be precise $R_s^{\strucSA}$ contains every tuple obtained via the following procedure. For each $t \in R_s^{\strucSA}$ look at each coordinate in turn. If it is of the form $g_i^0$, change it to $g_i^1$. If it is of the form $r_i^0$ (resp. $l_i^0$), change it to $r_i^1$ (resp. $l_i^1$). Finally, if it is of the form $a(i)$, change it to $b(i)$. 

The second type of tuple in $R_s^{\strucSB}$ are what might be thought of as dummy tuples. The idea is that we build these structures with an `intended strategy' in mind for Spoiler. These are tuples which are only included to prevent Spoiler from trying any funny business and deviating from this strategy. In particular we add every $(k+1)$-tuple which,
\begin{enumerate}
\item has $g_i^0$ as a coordinate, $i \in \{0,1\}$, or
\item has at least one coordinate of the form $l_i^0$ (resp.  $r_i^0$), $i\in [k]$ and does not have $g_0^1$ (resp. $g_1^1$) as a coordinate, or
\item has two coordinates in $D$.
\end{enumerate}
This completes the description of the start gadget; to build some intuition we will now see how it enables Spoiler to make $a(1), \dots, a(k)$ active, in the setting of the $\pexists (k+1)$-pebble game. \confORfull{}{This will form the first part of a Spoiler winning strategy in this game. Note that to show Spoiler can win the $\pexists(k+1)$-QVT game from $[\{\strucSA\}, \{\strucSB\}]$ iff there is some sentence of $\pexists \mathcal{L}^{k+1}$ which $\strucSA$ and $\strucSB$ disagree on iff Spoiler can win the $\pexists (k+1)$-pebble game from $[\strucSA, \strucSB]$.}

In the first two rounds Spoiler pebbles $g_0^0$ and $g_1^0$. Then Duplicator must either reply with $g_0^0$ and $g_1^1$ or $g_0^1$ and $g_1^0$ because of the relation $E$. We suppose the former case occurs, the latter case is almost identical. Then Spoiler moves every pebble except that lying on $g_1^0$ to cover $\{r_i^0 \, \mid \, i \in [k]\}$. Then the $\LHS$ models $R_s$ so, since $g_1^1$ is covered on the $\RHS$, Duplicator must reply such that if pebble $p$ lies on $r_i^0$ on the $\LHS$ it lies on $r_i^1$ on the $\RHS$. Next Spoiler moves the pebble lying on $g_1^0$ to $a(1)$ and Duplicator must reply with $b(1)$. Then Spoiler moves the pebble lying on $r_1^0$ to $a(2)$ and Duplicator must reply with $b(2)$. Continuing similarly we reach a position where, for each $i$, $a(i)$ is pebbled and the corresponding pebble lies on $b(i)$ on the $\RHS$. 

\subsubsection{The Split Gadget}

We next describe the split gadget. Firstly, for each $x\in X$ we introduce two elements $x_0$ and $x_1$ and define $X_i := \{x_i \, \mid \, x\in X\}$. The universe also contains $S:=  \{s_0^0, s_0^1, s_1^0, s_1^1 \}$. This completes the description of the universe of $S(\strucX)$; we also add colours. \confORfull{}{These ensure that whenever Spoiler makes a move Duplicator always has to reply with an element of the right `type'.} On $\strucSA$ we give a colour to $\{s_i^0, s_i^1\}$ and on $\strucSB$ this colour covers $\{s_i^0, s_i^1\} \cup D$. This ensure that if Spoiler plays some $s_i^a$, Duplicator must reply with $s_i^0, s_i^1$ or an element of $D$. Similarly we introduce colours to ensure that if Spoiler plays on $A$, Duplicator must reply with an element of $B \cup D$ and if Spoiler plays on $A_i$, $i \in \{0,1\}$, Duplicator must reply with an element of $B_i \cup D$. 

The notation for the elements of $S$ is\confORfull{ deliberately}{} reminiscent of that used in Section~\ref{s: lb}. \confORfull{}{This is deliberate.} We effectively encode an XOR clause $s_0 + s_1 = 1$. The rough idea is that if $s_i = 1$ on every $\RHS$ board then Spoiler can make a new element of $A$ active. We just have to `wire everything up' in the right way. \confORfull{We introduce four $(k+1)$-ary relations, $R_0, R_1, R_2, R_3$ to do this wiring}{For notational convenience we introduce four $(k+1)$-ary relations, $R_0, R_1, R_2$ and $R_3$ to do this wiring; it is possible to do the same thing with a single relation. }

Before giving the interpretation of these relations we fix some notation. For sets $U, V$ and integers $t_1, t_2$ we write $U^{t_1} \otimes V^{t_2}$ to denote the set of tuples which have exactly $t_1$ entries from $U$ and $t_2$ entries from $V$, in any order. \confORfull{We define $\underline{x}:= \{(x(1), \dots, x(k))\}, x\in \{a,b\}$.}{} One further technicality: we assume that all relations have distinct coordinates. Thus, below when we define relations\confORfull{}{ in terms of cross products} we implicitly exclude tuples with repeated coordinates. 

We now give the interpretation of these relations, along with an intuitive explanation of what they encode. Firstly, $R_0$ encodes the initial conditions for Spoiler to access the split gadget. These are that either: Spoiler has just finished a pass through the split gadget or that every $a(i)$ is active. This is encoded as follows.

\confORfull{\begin{align*}
R_0^{S(\strucA)}=& \, \left[\underline{a} \cup \left(A^{k-2} \times \bigcup_{a\in A} \{(a_0,a), (a_1,a)\} \right) \right] \times \{s_0^0\}   \\
R_0^{S(\strucB)}=& \, \left[ \underline{b} \cup \left(B^{k-2} \times \bigcup_{b\in B} \{(b_0,b), (b_1,b)\} \right) \right] \times \{s_0^0, s_0^1\} \\
\cup& \, (D \otimes S(B)^{k-1}) \times S(B)
 \end{align*} }
 { \begin{align*}
R_0^{S(\strucA)}=& \, \left[\{(a(1), \dots, a(k))\} \cup A^{k-2} \times \bigcup_{a\in A} \{(a_0,a), (a_1,a)\} \right] \times \{s_0^0\}   \\
R_0^{S(\strucB)}=& \, \left[ \{(b(1), \dots, b(k))\} \cup B^{k-2} \times \bigcup_{b\in B} \{(b_0,b), (b_1,b)\}\right] \times \{s_0^0, s_0^1\} \\
\cup& \, (D \otimes S(B)^{k-1}) \times S(B)
 \end{align*} }

Like in $R_s^{\strucSB}$, the tuples in $R_0^{\strucSB}$ are divided into two types. The first type---on the first line---are those which will be used if Spoiler follows the intended strategy. The second type---on the second line---stops Spoiler from trying any funny business. This essentially means that if Spoiler does not fulfil the initial conditions they cannot enter the split gadget. \confORfull{}{Note that the tuples represented by the second line are exactly those where at least one of the first $k$-coordinates lies in $D$.} Each of the remaining three relations corresponds to one step of Spoiler's passage through the split gadget; we now give each of them in turn.

\begin{align*}
R_1^{S(\strucA)}=& \,  A^{k-1} \times \{(s_0^0, s_1^0)\}  \\
R_1^{S(\strucB)}=& \,B^{k-1} \times \{(s_0^0, s_1^1), (s_0^1, s_1^0) \}
\\ \cup& \, D^2 \otimes S(B)^{k-1} \\
\cup& \, (D \otimes S(B)^{k-2}) \times S(B)^2
 \end{align*} 
We may think of this relation as encoding the XOR clause $s_1 +s_2 = 1$. To see this suppose that in round $r-1$ Spoiler successfully `enters' the split gadget using $R_0$. Let $\lstrucSBb$ be a $\RHS$ board, then $\assB(p) = s_0^i$ for some pebble $p$ and $i\in \{0,1\}$. In round $r$, Spoiler can move any pebble other than $p$ to $s_1^0$ and the $\LHS$ will model $R_1$. Then Duplicator will be forced to reply with $s_1^{i \oplus 1}$ on $\lstrucSBb$. Here, and throughout, we use $\oplus$ to denote addition modulo two. If Spoiler makes the $\LHS$ model $R_1$ without doing the requisite groundwork then
the second and third lines in the definition of $R_1^{\strucSB}$ ensure that they don't make any progress.

\begin{align*} 
R_2^{S(\strucA)}= \,& A^{k-1} \times \bigcup_{i=0}^1 \{s_i^0\} \times A_i \\
R_2^{S(\strucB)}=& \,B^{k-1} \times \bigcup_{i=0}^1 \{s_i^1\} \times B_i \\
\cup& \, \left[\left(D \cup \{s_0^0, s_1^0\}\right) \otimes S(B)^{k-1}\right] \times S(B)   
\end{align*}
Imagine a situation where the $\LHS$ models $R_1$ and there is a $\RHS$ board with an element of $B^{k-1} \times \{s_i^1, s_{i \oplus 1}^0\} $ covered, $i\in \{0,1\}$. Then Spoiler can move a pebble on the $\LHS$ from $s_{i \oplus 1}^0$ to some element of $A_i$ so that the resulting board models $R_2$. Duplicator is then forced to reply with an element of $B_i$. The second line encodes the fact that if we are not in the `intended' scenario Duplicator is not forced to play an element of $B_i$.
\begin{align*}
R_3^{S(\strucA)}=&\, A^{k-1} \times \bigcup_{a\in A} \{(a_0,a), (a_1,a)\} \\
R_3^{S(\strucB)}=&\, B^{k-1} \times \bigcup_{b\in B} \{(b_0,b), (b_1,b)\} \\
\cup& \,  (D \otimes S(B)^{k-1}) \times S(B) 
\end{align*}
Finally, suppose that on the $\LHS$ an element of $A^{k-1} \times \{(s_i^0, a_i)\}$ is covered and on some $\RHS$ board an element of $B^{k-1} \times \{(s_i^1, b_i)\}$ is covered, $a \in A, b\in B, i \in \{0,1\}$. Then Spoiler can move the pebble lying on $s_i^0$ on the $\LHS$ to $a$, making it a model of $R_3$.  Duplicator is forced to reply with $b$. Again the second line says if we are not in this specific scenario Duplicator can thwart Spoiler's progress. \confORfull{}{We will see more detail of how Spoiler makes elements of $A$ active in Section~\ref{ss: spoiler_win}}

\confORfull{}{The construction ensures that any winning Spoiler strategy for the $\pexists (k+1)$-\confORfull{QVT}{LB} game on 
$[S(\strucA), S(\strucB)]$ corresponds to a winning Spoiler strategy in the $\pexists k$-pebble game on $[\strucA, \strucB]$ and this correspondence is such that we can 
lift lower bounds. To be more 
precise if Duplicator can survive for $r$-rounds in the $k$-pebble game on $
[\strucA, \strucB]$, then they can obtain $2^r$-points in the $\pexists (k+1)$-LB game on $[\strucSA, \strucSB]$. Moreover, our construction blows up the 
number of elements in the structure by only a constant factor. Therefore we can 
lift the $\Omega(n^{2k-2})$ quantifier depth lower bound from \cite{DBLP:conf/cp/Berkholz14}  
to a $2^{\Omega(n^{2k-2})}$ quantifier number lower bound. }

 \subsection{Simulating Pebbles Games} \label{ss: simul}

We now explain the key idea behind Duplicator's strategy: to use Spoiler's moves to simulate a $\pexists k$-pebble game on $[\strucA, \strucB]$\confORfull{}{ and then to use this game to guide their responses}. \confORfull{}{We will also use this simulated game to show that Spoiler can win and therefore that there is some $\phi \in \pexists \logic^{k+1}$ which $\strucSA$ and $\strucSB$ disagree on.} To formalise this idea we will need some definitions.

If at some point in the game $\lstrucSAa \models R_3(\bar{p})$ for some 
$\bar{p}$, where $\lstrucSAa$ is the unique $\LHS$ board,  we say we are in 
a \emph{critical position}. Here, and throughout, we use $\bar{p}$ 
to represent a tuple of pebbles (= variables). If we are in a critical 
position then there is a unique $a\in A$, such that either $(a, a_0)$ or 
$(a, a_1)$ is covered, call this $a$ the \emph{potentially active} 
element. If instead either $\assA = \emptyset$ or a tuple of the form $(y_i^0, \dots y_k^0, a(1), \dots, a(i))$, $y\in \{l,r\}$, is covered on the unique $\LHS$ board we say we are in an \emph{initial position}.

Our simulated game, defined below, is well-defined iff Duplicator plays such that the following two facts hold throughout the $\pexists (k+1)$-LB game.

\begin{fact}\label{fact1} If pebble $p$ is active at the end of a round, then:
\begin{enumerate}
\item[\textbf{(a)}]  $p$ was active at the start of the round and has not moved, or
\item[\textbf{(b)}] we are in a critical position, $p$ was moved during this round and $p$ covers the potentially active element, or
\item[\textbf{(c)}] we are in an initial position and $p$ was moved during this round.
\end{enumerate} 
\end{fact}

\begin{fact} \label{fact2}
If at the end of some round pebble $p$ is active then, for any two RHS board $\lstrucSBb$ and $\lstrucSBc$, $\assB(p) = \assC(p)$. 
\end{fact}

In Section~\ref{ss: duplicator_stat} we will outline a strategy such that the above facts hold. \confORfull{The following is immediate}{An important consequence of Fact~\ref{fact1} is the following.}

\begin{lemma} \label{lem: fact3}
If Duplicator plays a strategy such that Fact~\ref{fact1} holds, then at the end of each round there can be at most $k$ active elements.
\end{lemma}
\confORfull{}{
\begin{proof}
Suppose at the end of round $r-1$ there are at most $k$ active elements but that at the end of round $r$, there are 
more than $k$ active elements. Therefore, a new element of $A$ becomes active in round $r$ and so by Fact~\ref{fact1} we know that we must be in either an initial position or a critical position. But in both cases at most $k$ elements of $A$ are pebbled, a contradiction.
\end{proof}}

We next define the simulated game. Denote by $P_t$ the position of the $\pexists(k+1)$-LB game starting from $[\{\strucSA\}, \{\strucSB\}]$ after round 
$t$ and the active pebbles in $P_t$ by $\act(P_t)$. We will use this to define a position 
in the $k$-pebble game on $[\strucA,\strucB]$, $S(P_t)$. We also define an injection, 
$f_r : \act(P_r) \to [k]$, which gives a correspondence between the active pebbles in the game on $[\{\strucSA, \strucSB\}]$ and pebbles in the simulated game on $[\strucA, \strucB]$. There are several cases.
\begin{enumerate}
\item $S(P_0) =[\strucA, \strucB]$.
\item If $\act(P_{r+1}) = \act(P_r) \setminus \{p\}$, then $S(P_{r+1})$ is attained from $S(P_r)
$ by deleting the pebble $f_r(p)$, $f_{r+1} = f_r|_{\act(P_{r+1})}$. 
\item Finally, suppose that $\act(P_{r}) \subseteq \act(P_{r+1}) $ and Spoiler moved $p\to x$ in round $r+1$. Then:
\begin{enumerate}
\item if $x \not \in A$ or $x$ isn't active at the end of round $r+1$  $S(P_{r+1}) = S(P_r)$ and $f_{r+1} =f_r$
\item otherwise $f_{r+1}(q) 
=f_r(q)$ for $q \neq p$ and $f_{r+1}(p) = \min \{m \in [k] \, \mid \, m \notin f_r(\act(P_r)\setminus \{p\}) \}$. 
If on some $\lstrucSBb \in \RHS$,  $\assB(p) = b$, then 
$S(P_{r+1})$ is obtained from $S(P_r)$ by moving $f_{r+1}(p)$ to  $x$ 
on $\strucA$ and $b$ on $\strucB$.
\end{enumerate}  
\end{enumerate}

Note that splits do not affect the simulated game, since by Fact~2 the active pebbles are placed on the same elements of every $\RHS$ 
board. Moreover, by Fact~\ref{fact1}, the above 
cases are exhaustive and, by Fact~\ref{fact2} and Lemma~\ref{lem: fact3}, case (3)(b) is
well-defined.

\confORfull{
\begin{figure*} 
\centering
\begin{subfigure}[t]{0.33\textwidth}
\begin{mybox}
\includegraphics[width=\linewidth]{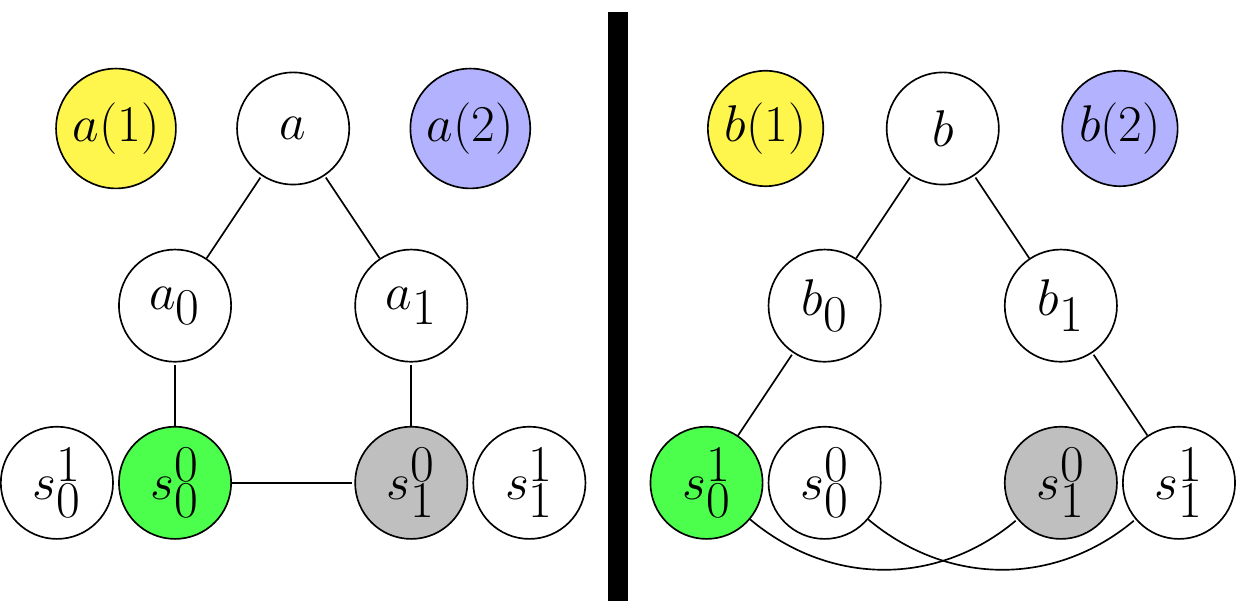}
\end{mybox}
\end{subfigure}
\begin{subfigure}[t]{0.33\textwidth}
\begin{mybox}
\includegraphics[width=\linewidth]{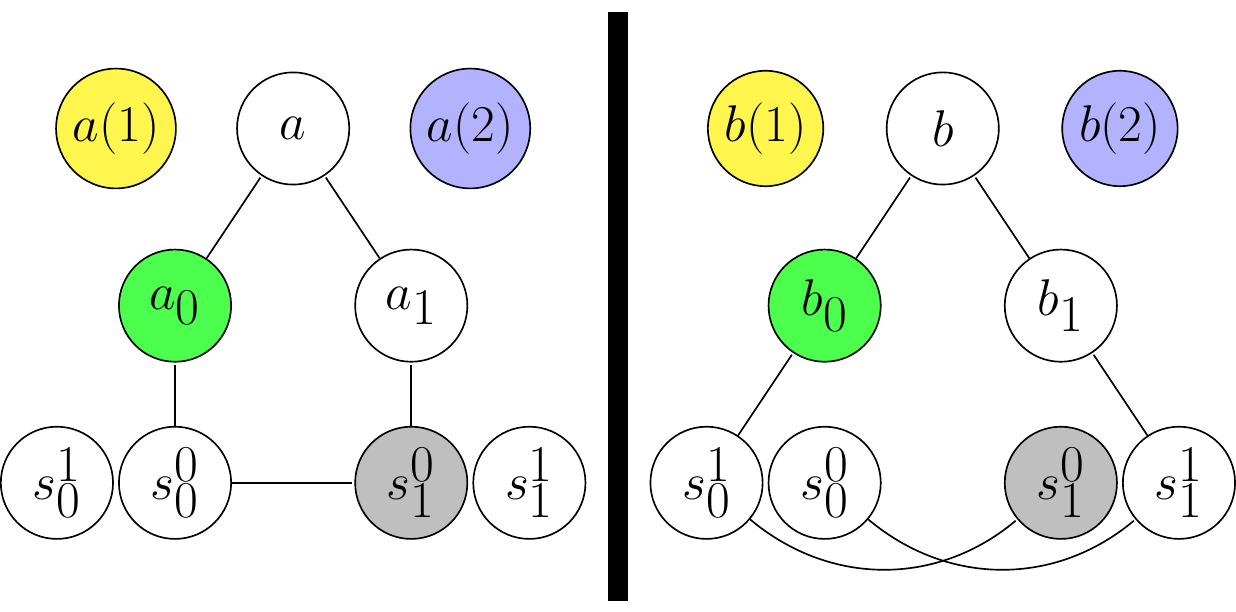}
\end{mybox}
\end{subfigure}
\begin{subfigure}[t]{0.33\textwidth}
\begin{mybox}
\includegraphics[width=\linewidth]{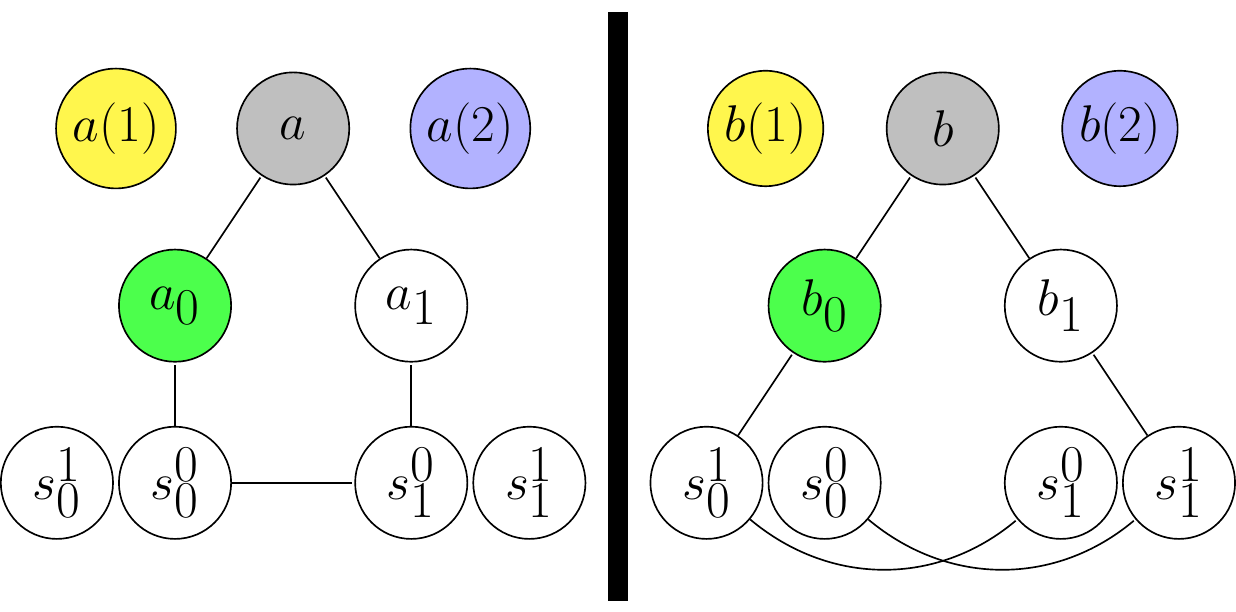}
\end{mybox}
\end{subfigure}

\caption{Part of the process by which Spoiler goes from a position where $a(1)$ and $a(2)$ are active to one where $a$ is also active, for $k=3$. The different colours represent the different pebbles. There is an edge $(u,v)$ on the left if $(a(1), a(2), u, v)$ lies in $\bigcup_{i=1}^3 R_i^{\strucSA}$. Similarly, there is an edge $(u,v)$ on the right if $(b(1), b(2), u, v)$ lies in $\bigcup_{i=1}^3 R_i^{\strucSB}$.}
\label{fig: spoilwin}
\Description[A figure showing how Spoiler can make the independent elements active.]{A figure showing how Spoiler can make the independent elements active.}
\end{figure*}}
{\begin{figure}
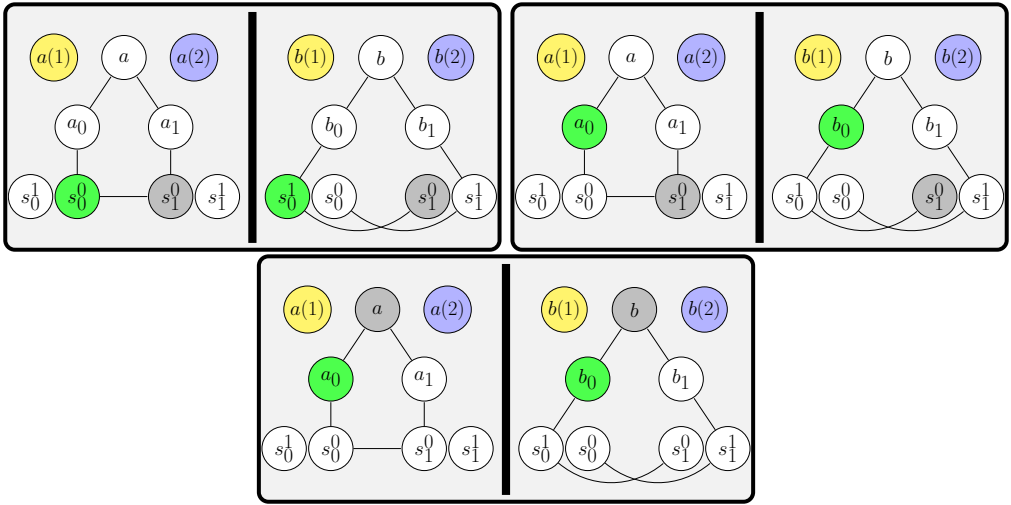
 
\centering
\begin{subfigure}[t]{0.49\textwidth}
\begin{mybox}
\includegraphics[width=\linewidth]{figures/fig-Spoilerwin}
\end{mybox}
\end{subfigure}
\begin{subfigure}[t]{0.49\textwidth}
\begin{mybox}
\includegraphics[width=\linewidth]{figures/fig-Spoilerwin1}
\end{mybox}
\end{subfigure}
\begin{subfigure}[t]{0.49\textwidth}
\begin{mybox}
\includegraphics[width=\linewidth]{figures/fig-Spoilerwin2}
\end{mybox}
\end{subfigure}
\caption{Part of the process by which Spoiler goes from a position where $a(1)$ and $a(2)$ are active to one where $a$ is also active, for $k=3$. The different colours represent the different pebbles. There is an edge $(u,v)$ on the left if $(a(1), a(2), u, v)$ lies in $\bigcup_{i=1}^3 R_i^{\strucSA}$. Similarly there is an edge $(u,v)$ on the right if $(b(1), b(2), u, v)$ lies in $\bigcup_{i=1}^3 R_i^{\strucSB}$. \confORfull{}{This corresponds to the case where $A_0 = \{a(1), a(2)\}$}}
\label{fig: spoilwin}
\end{figure}} 

\subsection{A Spoiler Winning Strategy} \label{ss: spoiler_win}
 
We can now show that there is a sentence of $\pexists\mathcal{L}^{k+1}$ separating $\strucSA$ and $\strucSB$. To do this we show that Spoiler can win the 
$\pexists (k+1)$-pebble game on $[S(\strucA), S(\strucB)]$. It still 
makes sense to talk about active pebbles in this context, as we may consider the 
$\pexists (k+1)$-pebble game as a special case of the $\pexists (k+1)$-LB game where Duplicator chooses to 
make a unique response on each turn. 

For the initial part of the Spoiler strategy they---unsurprisingly---use the start gadget to make $a(1), \dots, a(k)$ active, as described at the end of Section~\ref{sss: start}. Having done this Spoiler makes use of the split gadget for the first time, see Figure~\ref{fig: spoilwin}. If in an optimum strategy in the $\pexists k$-pebble game on $[\strucA, \strucB]$ Spoiler plays $a$ in the first round, then over the next few moves they will make $a$ active. 
First, they move the only pebble not lying on some $a(i)$ to $s_0^0$. Since the $\LHS$ models $R_0$ and as every $b(i)$ is pebbled on the $\RHS$ Duplicator must reply with $s_0^0$ or $s_0^1$. Spoiler then moves the pebble covering $a(k)$ to $s_1^0$.  Then the $\LHS$ models $R_1$ so, since an element of $B^{k-1}$ is covered on the $\LHS$, Duplicator replies such that either $s_0^0$ and $s_1^1$ or $s_0^1$ and $s_1^0$ are pebbled. Suppose we are in the latter case, the former case is almost identical. Then Spoiler moves the pebble lying on $s_1^0$ onto $a_0$. Since the $\LHS$ board is a model of $R_2$ and an element of $B^{k-1} \times \{s_0^1\}$ is covered on the $\RHS$, Duplicator must reply with some $b_0 \in B_0$. Finally, Spoiler moves the pebble lying on $s_0^1$ to $b$ and \confORfull{}{then since the $\LHS$ is a model of $R_3$ and on the $\RHS$ an element of $B^{k-1} \times \{b_0\}$ is pebbled} Duplicator must reply with $b$. \confORfull{W}{At the end of this process, w}e get a board where $a$ is active. 

Now suppose after $r$ rounds we are in a critical position where $k$ elements of $A$ are active. This is the case after the first use of the split gadget which acts as our base case. Then Spoiler looks at $S(P_r)$ and finds the next move in an optimal Spoiler winning strategy, say moving pebble $p$ from $a'$ to $a$. Over the new few moves Spoiler will---without moving any pebbles in $\act(P_r) \setminus \{p\}$---reach a position where $a$ is active. \confORfull{This is similar to the above; we give the details in Appendix~\ref{a: exlb}.}{Then let $A_0$ be the set of all active elements of $A$ other than that covered by $q$, where $q$ is the unique pebble such that $f_r(p) = q$. Over the next few moves Spoiler will reach a position such that $A_0 \cup \{a\}$ are the active elements of $A$, see Figure~\ref{fig: spoilwin}. 

To do this they first move $q$ to cover $s_0^0$. Then the $\LHS$ models $R_0$ and so Duplicator must reply with $s_0^0$ or $s_0^1$. Next Spoiler moves the unique pebble covering an element of $A_0 \cup A_1$ to $s_1^0$. Then the $\LHS$ models $R_1$ and so, again, since an element of $B^{k-1}$ is pebbled on the $\LHS$ we get that Duplicator reply such that either $s_0^0$ and $s_1^1$ or $s_0^1$ and $s_1^0$ are pebbled. So from here Spoiler follows the same pattern as in their first pass through the split gadget and we eventually get a board where $A_0 \cup \{a\}$ are active.} 

Suppose Spoiler wins the 
$\pexists k$-pebble game on $\strucA, \strucB$ in $r$-rounds. Then by 
passing through the split gadget at most $r$ times, we reach a position which 
is winning for Spoiler in the $\pexists (k+1)$-pebble game on $\strucSA, \strucSB$. This gives an $O(r)$ upper bound for the quantifier depth 
of a $\exklogic$ sentence separating $\strucSA$ and $\strucSB$. It can be shown (by a similar 
but easier argument to that given below) that this upper bound is 
tight. Therefore, in terms of quantifier depth, a sentence separating $S(\strucA)$ and $S(\strucB)$ has the 
same complexity, up to a constant, as one separating $\strucA$ and $\strucB$. But there 
can be a big difference in terms of number of quantifiers, as we will now see.

\confORfull{}{\subsection{Lower Bound Argument} \label{ss: lb}

The aim of this section is to prove the following lower bound.

\begin{theorem} \label{thm: main_tec}
Suppose Spoiler needs $r$-rounds to win the $\pexists k$-pebble game on $[\strucA, \strucB]$. Then Duplicate can get at least $2^{r}$ points in the $\pexists (k+1)$-LB game on $[\strucSA, \strucSB]$. 
\end{theorem}
}

\confORfull{\subsection{Duplicator Strategy} \label{ss: duplicator_stat}}
{\subsubsection{Duplicator Strategy} \label{ss: duplicator_stat}}

Before stating the strategy let us make a few simplifications. We may assume Spoiler never plays any $d\in D$ since these elements do not appear in any tuple of any $R^{\strucSB}$, $R \in S(\sigma)$. Similarly we may assume that Spoiler doesn't play any element of the form $x_i^1$, $x\in \{s, g, l, r\}$, for on the $\LHS$ each such element appears only in a single unary relation along with $x_i^0$ . It is easy to see that if Spoiler plays $p \to x_i^1$ in a winning strategy, then if they replaced this move with $p \to x_i^0$ this would also be winning.  Also if Spoiler plays on top of a pebble, then Duplicator always plays their only valid move, i.e., to play on top of the corresponding pebble on every $\RHS$ board. As every relation has disjoint positions and since there is always a unique $\LHS$ board it follows that such moves do not help Spoiler, so we assume that such moves are not made.\footnote{The fact that there is a unique $\LHS$ board is crucial for this\confORfull{, see \cite{fagin2021}.}{. In the context of pebble games, we are used to assuming that it never helps Spoiler to play `on top' of an existing move. However, perhaps surprisingly, Duplicator's ability to make copies of boards means this assumption ceases to be valid in our games, see \cite{fagin2021} for an example of this.}}

\confORfull{For the start gadget, Duplicator always plays the same element on every $\RHS$ board and at each point they determine their response via a procedure we call the start strategy. We describe this in Appendix~\ref{a: exlb} where we also show the following easy lemma.}{We begin by looking at what happens when Spoiler plays on the start gadget and show that Duplicator always has a valid response to such moves, independently of their strategy on the rest of the structure. In such situations, Duplicator determines their reply on a $\RHS$ board $\lstrucSBb$ via the following procedure, which we call the \emph{start strategy}. If Spoiler plays some $g_i^0$ then Duplicator replies with $g_i^0$ unless either (1) this response is not valid or (2) $i=0$ and an element of the form $r_j^0$ is pebbled on $\lstrucSBb$ or (3) $i=1$ and an element of the form $l_j^0$ is pebbled on $\lstrucSBb$. If (1), (2) or (3) hold Duplicator instead plays $g_i^1$. If Spoiler instead plays $y_i^0$, $y \in \{l, r\}$, Duplicator plays $y_i^1$ whenever it is valid and $y_i^0$ otherwise. The following follows by a simple induction, we give the details in Appendix~\ref{a: start}. }

\begin{lemma} \label{lem: start}
Suppose that in round $r$ of the $\pexists (k+1)$-LB game starting from position $[\{\strucSA\}, \{\strucSB\}]$, Spoiler plays $p\to x \in G$. Suppose that in every previous round in which Spoiler played an element of $G$ Duplicator followed the start strategy on every $\RHS$ board. Then the move recommended by the start strategy in round $r$ is valid on every $\RHS$ board.
\end{lemma}

We now describe Duplicator's complete strategy, which we call the \emph{lower bound strategy}. This ensures there are at most two $\RHS$ boards at the end of every round. It will be helpful to have names for our $\RHS$ boards: $T_0$ and $T_1$. Allowing Duplicator to maintain two identical boards does not help them to win so we do this at points for notational convenience. 

Suppose that Spoiler moves pebble $p \to x$ in round $r$. If $x\in G$, Duplicator follows the start strategy on every board. Otherwise, if on a $\RHS$ board there is some $d\in D$ which is a valid response Duplicator plays such a $d$. Call this the passive response. If the passive response is not valid on some $T_{\ell}$, then Duplicator's move depends on the colour of $x$.

 If $x \in A $, then if there is some $b\in B$ which is a valid response Duplicator chooses one such $b$ and plays it. If no such $b$ 
exists Duplicator plays arbitrarily---the board will be deleted anyway. If $x = s_i^0$, for some $i \in \{0,1\}$, then Duplicator plays $s_i^0$ on $T_{i \oplus 1}$ and $s_i^1$ on $T_i$; put differently Duplicator plays $s_i^{i \oplus \ell \oplus 1}$ on $T_{\ell}$.

Finally, suppose $x = a_i \in A_i$ for some $i\in \{0,1\}$. Then Duplicator looks 
at $S(P_{r-1})$, if $p$ was not previously active, and the position obtained from $S(P_{r-1})$ by removing $f_{r-1}(p)$ otherwise. From this position, they imagine that
Spoiler places pebble $q := 
\min\{m \, \mid \, m \not\in f_{r-1}(\act(P_{r-1})\setminus \{p\})\}$ on $a$. Then Duplicator finds an
optimum response in the $\pexists k$-pebble game to this move, say $b$. Then $b_i$ is the active response.\confORfull{}{\footnote{It may seem 
here that we are defining the simulated game in terms of our 
strategy and our strategy in terms of the simulated game therefore 
creating some kind of circularity. But in fact this circularity is 
not vicious. All we need to define the simulated game $S(P_r)$ is 
the position of the game at the start of round $r$ and all we need to determine 
what to do in round $r$ is the state of $S(P_{r})$. So from the 
position of the game at the start of round $1$ (i.e. the trivial position) $S(P_0)$ 
is determined, from $S(P_0)$ Duplicator knows how to respond in 
round $1$. Once this is done we can find $S(P_1)$ and so on.}} Note that for every $R \in \sigma$, no $a(i)$ appears as a coordinate of a tuple in $R^{\strucA}$. Therefore, any Duplicator response to this in the $\pexists k$-pebble game is optimal and so we stipulate that Duplicator always chooses $b(i)$ in the simulated game in this case. 

If Spoiler performs a split then there is a unique way this can be done, since there are always at most two $\RHS$ boards, and Duplicator gives Spoiler a free choice on which partition to continue the game on. In particular, the degree of every split is two.

Finally, Duplicator will sometimes delete boards. This is not formally part of the rules of the game, but allowing Duplicator to do this makes it harder for them to win so we may allow such moves without affecting the veracity of our lower bound. Deletions occur only if at the end of the round we are in an initial or a critical position. If we are in an initial position then Duplicator chooses a surviving board arbitrarily and deletes any other board. Afterwards, they make two identical copies of the remaining board and rename one as $T_0$ and one as $T_1$. Similarly, if we are in a critical position and there are two surviving boards, then Spoiler chooses the board with the lowest number of pebbles lying on $B$, breaking ties arbitrarily. They then delete any other board, make two identical copies of the remaining board and re-name one as $T_0$ and one as $T_1$. We call both of these manoeuvres \emph{resets}.
 
\subsubsection{Analysis of the Duplicator Strategy}

As our strategy depends on the simulated game being well-defined, our first step is to prove that Facts~\ref{fact1} and \ref{fact2} hold. Afterwards, we show that Duplicator's strategy respects every relation in $S(\sigma) \setminus \sigma$; this is the key point allowing us to prove Theorem~\ref{thm: main_tec}. In order to make the inductive arguments go through it is convenient to prove Facts 1 and 2 simultaneously; in fact, we prove a slight strengthening of Fact 2. 

\begin{lemma}[Facts 1 and 2] \label{lem: facts12}
If Duplicator plays the $\pexists (k+1)$-LB game starting from position $[\{\strucSA\}, \{\strucSB\}]$ using the lower bound strategy, then a pebble $p$ can only be active at the end of a round, if:
\begin{enumerate}
\item[a.]  $p$ was active at the beginning of the round and has not moved or
\item[b.] we are in a critical position, $p$ was moved during this round and $p$ covers the potentially active element, or
\item[c.] we are in an initial position and $p$ was moved during this round.
\end{enumerate} 
Moreover, for any pebble $p$ if $p$ lies on $b\in B$ on one $\RHS$ board, then it lies on $b$ on every $\RHS$ board.
\end{lemma}

\begin{proof}
We prove this by induction on the number of rounds played, $r$. The base case $r=0$ is trivial. Suppose that in round $r$, Spoiler moves $p\to x$ and some $a \in A$ 
becomes active. Let $\lstrucSAa \in \LHS$ after round 
$r$. Suppose that $\assA(q) = a$, for some pebble $q \neq p$. Since $q$ is not active at the beginning of round $r$ it lies on $D$ on some $\RHS$ board. Therefore, by the induction hypotheses, $q$ lies on $D$ on every $\RHS$ board. It follows that $q$ cannot become active in round $r$. 

So suppose $a=x$ and that we are in neither an initial nor a critical position at the end of round $r$. Then the $\LHS$ board is not a model of $R_3$ or $R_s$ and so, by the interpretation of every other relation in $\strucSB$, on any $\RHS$ board $p \to d$ is a valid response for any unused $d\in D$. Since Duplicator's strategy dictates that they play such a move whenever it is valid, it follows that $a$ does not become active. Similarly, if $a=x$ and we are in a critical position but $a$ is not the potentially active element then by the interpretation of $R_3$ it follows that $p \to d$ is a valid response on every $\RHS$ board. 

Now let $p$ be a pebble lying on some $b\in B$ on some $\RHS$ board at the end of round $r$. We claim that $p$ lies on $b$ on every $\RHS$ board. If $p$ was not moved during the last round this follows from the induction hypothesis. So suppose $p$ was moved during the last round. Then by the previous analysis since $p$ lies on $b$ on some $\RHS$ board it must be that we are in either a critical or an initial position. But in both these cases Duplicator performs a reset at the end of the round, leaving two identical boards. The claim follows.
\end{proof}

We next move on to the task of showing that Duplicator's strategy always respects every relation in $S(\sigma) \setminus \sigma$. To do this we will need auxiliary hypotheses which, as a bonus, imply that Spoiler cannot change the set of active pebbles without splitting. To help simplify the statement of our next lemma we introduce the notion of a \emph{safe} pebble. The idea is that if every pebble lying on $A$ is safe then the position in the simulated game has not got better for Spoiler since the last initial/critical position. If every pebble lying on $A_0$ and $A_1$ is also safe then Spoiler is, in some sense, not even close to changing this situation.

\begin{definition} \label{def: safe}
Let $P=[\{\lstrucSAa\}, \classB ]$ be a position in round $r$ of the $\pexists (k+1)$-LB game starting from position $[\{\strucSA\}, \{\strucSB\}]$ where Duplicator plays using the lower bound strategy. Suppose that the last time prior to round $r$ that an initial or a critical position was reached $\{v(0), \dots, v(t)\} \subseteq A$ were the active elements. Moreover, suppose that for every $i\in[t]$, the pebble which covered $v(i)$, covered $u(i) \in B$ on every $\RHS$ board. Then we say a pebble $p$ with $\assA(p) \in A$ is \emph{safe} at $P$ if one of the following holds.
\begin{enumerate}
\item There is some $i\in[t]$ such that $\assA(p) = v(i)$ and $\assB(p) =u(i)$ for every $\lstrucSBb \in \classB$.
\item There is some $i\in[k]$  such that $\assA(p) = a(i)$ and $\assB(p) =b(i)$ for every $\lstrucSBb \in \classB$.
\item $\assB(p) \in D$ for every $\lstrucSBb \in \classB$.
\end{enumerate}
Similarly, we say a pebble $p$ such that $\assA(p) \in A_0 \cup A_1$ is \emph{safe} at $P$ if one of the following holds.
\begin{enumerate}
\item There is some $i\in[t]$, $j\in \{0,1\}$ such that $\assA(p) = v(i)_j$ and $\assB(p) =u(i)_j$ for every $\lstrucSBb \in \classB$.
 \item $\assB(p) \in D$ for some $\lstrucSBb \in \classB$.
\end{enumerate}
Finally, it will be linguistically convenient to stipulate that every pebble $p$ not lying on $A \cup A_0 \cup A_1$ is safe.
\end{definition}

Note that in the definition we refer to the last round \emph{prior} to the current round that an initial or a critical position occurred. So even if we reach a critical position if every pebble is safe then no new elements of $A$ become active. We are now in a position to prove the following lemma.

\begin{lemma} \label{lem: dual-boards}
Suppose that Duplicator plays the lower bound strategy in the $\pexists (k+1)$-LB game starting from position $[\{\strucSA\}, \{\strucSB\}]$. Suppose further, that at the start of round $r$ there are two $\RHS$ boards and that in this round Spoiler does \emph{not} perform a split. Then there are also two $\RHS$ boards at the end of round $r$. Moreover, every pebble is safe at the end of the round.
\end{lemma}

\begin{proof}
We prove this by induction on $r$; the base case $r=1$ is trivial. So suppose that there are two $\RHS$ boards at the start of round $r$ and that in this round Spoiler does not perform a split. Since there are two $\RHS$ boards at the start of the round either Duplicator performed a reset in the last round or there were two boards at the start of round $r-1$ and Spoiler didn't perform a split. In both cases, every pebble is safe at the beginning of round $r$, since resets create two identical boards and by the induction hypotheses. So suppose Spoiler moves pebble $p \to x$ and let $\lstrucSAa$ be the unique $\LHS$ board after Spoiler's move. We proceed by case analysis. 

\underline{Case 1:} $\mathbf{x\not \in A \cup A_0 \cup A_1}$ \\
Here it suffices to show that Duplicator's response is valid on every $\RHS$ board. For $x\in G$ this follows directly from Lemma~\ref{lem: start}. Next, suppose $x = s_i^0$. If Duplicator replies with the passive response on $T_{\ell}$ this is valid by definition. So suppose this response is not valid, then we claim the Duplicator response $s_i^{\ell \oplus i \oplus 1}$ is valid. Let $ \lstrucSBb$ be $T_{\ell}$ after this move is played. Then, since the passive response is not valid, $\lstrucSAa \models R_0(\bar{p}) \vee 
 R_1(\bar{p})$, for some tuple of variables $\bar{p}$. Since every pebble is safe at the beginning of the round, it is easy to deduce that $s_i^0$ and $s_i^1$ are valid responses. If $\lstrucSAa \models R_1(\bar{p})$
then there is some pebble $q$ with $\assA(q) = s_{i \oplus 1}^0$. Since the passive response is not valid we know that $\assB(q) \not \in D$. So
by the definition of the Duplicator strategy $\assB(q) 
= s_{i \oplus 1}^{i + \ell}$. Therefore $\assB(\bar{p}) \in B^{k-1} 
\times \{(s_0^1, s_1^0), (s_0^0, s_1^1)\} \subset R_1^{\strucSB}$. 

\underline{Case 2:} $\mathbf{x = a_i \in A_i}$ \\
Again it is enough to consider the case where the passive response is not valid on some $T_{\ell} := \lstrucSBb$. This only occurs if $\lstrucSAa \models R_2(\bar{p})$. In this case, let $q$ be the pebble such that $\assA(q) = s_i^0$. If $\assB(q) \in D$ then the passive response would be valid, a contradiction. Therefore, by the definition of Duplicator's strategy, $q \to s_i^{i \oplus \ell \oplus 1}$. Therefore $i = \ell$ as else again the passive response would be valid. Then recall that Duplicator replies with an element of $B_{\ell}$ so it follows that $\assB(\bar{p}) \in B^{k-1} \times \{s_{\ell}^1\} \times B_{\ell} \subset R_2^{\strucSB}$. Moreover, as the analysis above implies that the passive response is valid on $T_{\ell \oplus 1}$, pebble $p$ is safe.

\underline{Case 3:} $\mathbf{x = a \in A}$ \\
Once again consider the case where the passive response is not valid on $T_{\ell}$. By Lemma~\ref{lem: facts12}, we may deduce that we are in either an initial or a critical position. Suppose we are in a critical position. Then since the passive response is not valid on a $\RHS$ board it must be that $p$ is the potentially active element. If $p$ is not active then we have two boards at the end of the round and $p$ is safe, since Duplicator performs a reset at the end of the round. So suppose $p$ is active. Let $q$ 
be the unique pebble lying on some $a_i \in A_0 \cup A_1$ on the $\LHS$. Then $q$ lies on $B_i$ on both $T_0$ and $T_1$ since the passive response is not valid on either $\RHS$ board. Since $q$ is safe $\assA(q) = v_i$ and $q$ lies on $u_i$ on every $\RHS$ board, where $u, v$ are such that the last time we were in a critical position $q$ lay on $u$ on the $\LHS$ and on $v$ on every $\RHS$ board. Therefore, in round $r$ Duplicator responds with $u$ on both $\RHS$ boards. In the resulting position, since Spoiler did not win the last time we were in a critical position, no relation in $\sigma$ is violated so $u$ is a valid response on $T_0$ and $T_1$. Moreover, $p$ is safe. 

If we are in an initial position then $\assA(\bar{p}) = (y_i^0, \dots y_k^0, a(1), \dots, a(i))$ for some $y \in \{l,r\}$ , $i\in [k]$. Then we know that $x= a(j)$ for some $1 \le j \le i$. If the passive response is not valid on some $T_{\ell}:= \lstrucSBb$ then this implies that no element of the form $y_t^0$ or $D$ is pebbled on $T_{\ell}$. Then Duplicator plays $b(i)$ which implies that $\assB(\bar{p}) =  (y_i^1, \dots y_k^1, b(1), \dots, b(i)) \in R_s^{\strucSB}$. Moreover, since Duplicator performs a reset at the end of the round it is easy to see that $p$ is safe. 
\end{proof}

By performing a similar analysis to the above for the cases where there is only one $\RHS$ at the beginning of the round and where Spoiler performs a split we can obtain the following lemma. In fact, since the move Duplicator plays on $T_{\ell}$ does not depend on whether there are two $\RHS$ boards much of the above analysis can be deployed for this task.

\begin{lemma} \label{lem: ignore}
Suppose Duplicator plays the $\pexists(k+1)$-LB game starting from position $[S(\strucA), S(\strucB)]$ according to the lower bound strategy. Let $\lstrucSAa$ be the unique $\LHS$ board after Spoiler's move in round $r$ and $\lstrucSBb$ be a $\RHS$ board after Duplicator's response. Then for all $R \in S(\sigma) \setminus \sigma$ and all $\bar{p}$, $\lstrucSAa \models R(\bar{p})$ implies $\lstrucSBb \models R(\bar{p})$. 
\end{lemma}

\begin{proof}
 We proceed by induction on the number of rounds, noting that the base case is trivial. Firstly, if there are two boards at the start of round $r$ and Spoiler does not perform a split this follows directly from the proof of Lemma~\ref{lem: dual-boards}. So suppose that there is one board at the start of round $r$ or that Spoiler performs a split in this round. Let $\lstrucSAa$ be the unique $\LHS$ board after round $r$. The proof in this case is similar to Lemma~\ref{lem: dual-boards} but we no longer have the guarantee that every pebble is safe. To replace this assumption we  maintain the following auxiliary hypotheses.   
\begin{enumerate}
\item If at the end of round $r$, $\assA(p) \in A$ and we are not in a critical position where $\assA(p)$ is the potentially active element, then $p$ is safe.
\item Suppose that at the end of round $r$, $\assA(p) =a$, $\assA(q) =a_i$, for some $a\in A$. If $\lstrucSBb \in \LHS$ with $\assB(p), \assB(q) \not\in D$ then $\assB((p,q)) \in \bigcup_{b\in B} (b, b_i)$.
\item If $\assA(p) = a(i)$ then $\assB(p) \in \{b(i)\} \cup D$.
\end{enumerate}

We first show that it is sufficient to analyse the case where there is a single board at the beginning of round $r$. To see this suppose Spoiler performs a split in round $r$. Then there were two boards at the end of round $r-1$. Therefore, either there were two boards at the start of round $r-1$ or Duplicator performed a reset in round $r-1$. By Lemma~\ref{lem: dual-boards} and since after a reset both boards are identical, every pebble is safe at the beginning of round $r$. Therefore, it is easy to see that after Spoiler performs a split (1)-(3) hold. It follows that it is sufficient to analyse the case where at the beginning of round $r$ the $\LHS$ consists of a single board and (1)-(3) hold. 

So suppose the unique $\RHS$ board is $T_{\ell}$ and that Spoiler plays $p \to x$ in this round to produce a $\lstrucSAa \in \RHS$. We proceed by case analysis.

\underline{Case 1:} $\mathbf{x\not \in A \cup A_0 \cup A_1}$ \\
Here it suffices to show that Duplicator's response is valid on the $\LHS$ board. The argument is almost the same as in Lemma~\ref{lem: dual-boards} but now using our auxiliary hypotheses. In detail we need to show that if the passive response is not valid then the Duplicator response $s_i^{\ell \oplus i \oplus 1}$ is valid. In this case, as before, $\lstrucSAa \models R_0(\bar{p}) \vee 
 R_1(\bar{p})$, for some tuple of variables $\bar{p}$. If $\lstrucSAa \models R_0(\bar{p})$, then by (2) and (3) both $s_i^0$ and $s_i^1$ are valid responses. If $\lstrucSAa \models R_1(\bar{p})$
then there is some pebble $q$ with $\assA(q) = s_{i \oplus 1}^0$. Since the passive response is not valid we know that $\assB(q) \not \in D$. So
by the definition of the Duplicator strategy $\assB(q) 
= s_{i \oplus 1}^{i + \ell}$. Therefore $\assB(\bar{p}) \in B^{k-1} 
\times \{(s_0^1, s_1^0), (s_0^0, s_1^1)\} \subset R_1^{\strucSB}$. 

\underline{Case 2:} $\mathbf{x = a_i \in A_i}$ \\
The argument that Duplicator's response is always valid is identical to Case~(2) in Lemma~\ref{lem: dual-boards}. If Duplicator makes the passive response the induction hypotheses are clearly maintained. Otherwise $\lstrucSAa \models R_2(\bar{p})$ and $k-1$ elements of $A$ are active. In this case, we need to show that (2) is maintained. So suppose that Spoiler moves pebble $q$ to $a_i$ and that some pebble $p$ lies on $a$. Then to decide on their response Duplicator imagines that in position $S(P_{r-1})$, Spoiler moves the unused pebble to $a$ and finds an optimum response in the $\pexists k$-pebble game. But pebble $f_{r-1}(p)$ already lies on $a$. Therefore, the only response which doesn't immediately lose for Duplicator in the pebble game is to play the element covered by $f_{r-1}(p)$ on the $\LHS$. But this element is $\assB(p)$ by definition, so in the $\pexists (k+1)$-LB game Duplicator plays $\assB(p)_i$ and so (2) is maintained.

\underline{Case 3:} $\mathbf{x = a \in A}$ \\
Once again it is enough to consider the case where the passive response is not valid on $T_{\ell}$. By Lemma~\ref{lem: facts12} we may deduce that we are in either an initial or a critical position. Suppose we are in a critical position. Then since the passive response is not valid on a $\RHS$ board it must be that $p$ is the potentially active element. But in this case, we do not need to show that $p$ is safe. Moreover, since we are in a critical position and the passive response is not valid, after Duplicator's move it follows that $\assB(\bar{p}) \in B^{k-1}  \times \bigcup_{b\in B} \{(b_0, b), (b_1, b)\} \subset R_3$. This also implies that (2) is maintained. To see that (3) is maintained, suppose that $\assA(p) = a(j)$. Then there is a pebble $q$ with $\assA(q) = a(j)_i$. Since the passive response is not valid $\assB(q) \not \in D$. Therefore $\assB(q) = b(j)_i$, by the definition of Duplicator's strategy. Therefore  Duplicator's only valid response is to play $b(j)$ and so (3) is maintained.

If we are in an initial position then $\assA(\bar{p}) = (y_i^0, \dots y_k^0, a(1), \dots, a(i))$ for some $y \in \{l,r\}$ , $i\in [k]$. Then we know that $x= a(j)$ for some $1 \le j \le i$. If the passive response is not valid on some $T_{\ell}:= \lstrucSBb$ then this implies that no element of the form $y_t^0$ or $D$ is pebbled on $T_{\ell}$. Then Duplicator plays $b(i)$ which implies that $\assB(\bar{p}) =  (y_i^0, \dots y_k^0, b(1), \dots b(i)) \in R_s^{\strucSB}$, by (3). Moreover, $p$ is safe, so (1) is maintained.
\end{proof}

This is important as it shows that Spoiler can only win the LB game if they win the associated simulated game. We have now done the heavy lifting and can begin putting everything together; first, we deduce the following corollary to Lemmas~\ref{lem: dual-boards} and \ref{lem: ignore}. 

\begin{corollary} \label{cor: must-split}
Suppose Duplicator plays the $\pexists(k+1)$-LB game starting from  position $[S(\strucA), S(\strucB)]$ using the lower bound strategy. Also, suppose that after $r$ rounds there are two $\RHS$ boards. Then in order to win Spoiler must perform a split. Moreover, if the last time a critical position occurred $A_0 \subseteq A$ were the active elements, then until a split takes place no element of $A \setminus (A_0 \cup \{a(i) \, \mid \,i \in [k]\})$ can become active.
\end{corollary}

\begin{proof}
 Suppose that at the beginning of round $r$ there are two $\RHS$ boards and that Spoiler moves $p \to x$. Also, suppose that the last time we were in a critical position $A_0 \subseteq A$ were the active elements. We note that no split can have occurred since the last initial or critical position ---after a split there is always one $\RHS$ board and we only get two $\RHS$ boards again if we reach an initial or a critical position. We claim that the active elements at the end of round $r$ are a subset of $A_0 \cup \{a(i) \, \mid \,i \in [k]\}$. 
 
To see this suppose it holds $t$ rounds after we were last in an initial or a critical position. So suppose that some $a\in A$ becomes active in round $t$. Then we know by Fact~\ref{fact1} that we are in an initial or a critical position. If we are in an initial position then $a \in \{a(i) \, \mid \,i \in [k]\}$. If we are in a critical position then $x=a$ is the potentially active element. But then by Lemma~\ref{lem: dual-boards}, since the passive response is not valid on either $\RHS$ board, it must be that $x\in A_0$. This completes the proof of our claim.

Moreover, by Lemma~\ref{lem: ignore} we know that only relations in $\sigma$ can be violated by $T_0$ or $T_1$. By the claim it is clear that no relations in $\sigma$ are violated in round $r$, therefore Duplicator survives round $r$. Since in round $r$ we are in an arbitrary 2-board position, it follows that for Duplicator to win they must at some point perform a split.
\end{proof}

We now need to show that Spoiler has to perform $r$ splits in order to win. To this end, we introduce the following notation. Suppose Spoiler wins the $\pexists (k+1)$-LB game in $\ell$ 
rounds starting from position $[\strucSA, \strucSB]$. Then for every $t\le \ell$ we  write $c(t)$ to denote 
the minimum number of moves in a Spoiler winning strategy in the $\pexists k$-pebble game on $S(P_t)
$. We know that $c(0) = r$, $c(\ell) =0$ and $c(t) =c(t-1)$, if the set of active pebbles 
doesn't change in round $t$. The second point follows by Lemma~\ref{lem: ignore}, along with the observation that if the $\LHS$ board models some $R\in \sigma$ and no $\RHS$ board models $R$ that the simulated game is in a winning position for Spoiler. The following proves that Spoiler must effectively carry out the strategy laid out in Section~\ref{ss: spoiler_win}. For expositional clarity we first deal with a special case.

\begin{lemma} \label{lem: incremenet_base}
Suppose that in round $t_1$ of the $\pexists (k+1)$-LB game starting from position $[\{\strucSA\}, \{\strucSB\}]$ a critical position is reached. Moreover, suppose that Duplicator plays the lower bound strategy and that in round $t_0$ an initial position is reached and no critical or initial position is reached in any round $t$ with $t_0 < t < t_1$. Then $c(t_1) \ge r-1$. 
\end{lemma}

\begin{proof}
First, observe that $c(t_0) = r$. This follows because in any initial position only the $a(i)$ elements are active and as these are isolated in $\strucA$. Now let $t$ be such that $t_0 < t < t_1$. We claim that $c(t) = r$. This follows because by Fact 1 the active elements are  all of the form $a(i)$ in round $t$. Moreover, if no new new element of $A$ becomes active in round $t_1$ then also $c(t_1) =r$. Otherwise, by Fact 1, Spoiler moves $p \to a$ where $a$ is the potentially active element and Duplicator responds with some $b$.

Note that the element $b$ is determined by the unique element $b_i \in B_0 \cup B_1$ 
which is pebbled on every $\RHS$ board, by Lemma~\ref{lem: dual-boards}. Furthermore, since in round $t_0$ we were in 
an initial position it must be that $b_i$ was pebbled in some round $
\hat{t}$ with $t_0 < \hat{t} < t_1$. Recall that this $\hat{b}$ was 
determined by Duplicator by looking at $S(P_{\hat{t}-1})$, imagining that 
Spoiler placed pebble $j:= \min\{m \, \mid \, m \not\in f_{\hat{t}-1}
(P_{\hat{t}-1}))\}$ on $a$ and then finding the optimal response. But because in the position $S(P_{\hat{t}-1})$ only elements of the form $a(j)$ are pebbled, then this Duplicator response is the same as if in the first round of the game on $[\strucA, \strucB]$, Spoiler had pebbled $a$. It follows that $c(t_1) \ge r-1$.
\end{proof}

\begin{lemma} \label{lem: increment}
Consider the $\pexists (k+1)$-LB game starting from \confORfull{\\}{}$[\{\strucSA\}, \{\strucSB\}]$, where Duplicator plays the lower bound strategy. Suppose a critical position is reached in this game after $t_0$ rounds and the next critical position is reached after $t_1$ rounds. Then $c(t_1) +1 \ge c(t_0)$. 
\end{lemma}
\begin{proof}
First, suppose there is some $t$ with $t_0 < t < t_1$ such that in round $t$ the game is in an initial position. Then by Lemma~\ref{lem: incremenet_base}, $c(t_1) \ge r-1 \ge c(t_0) - 1$. So we may suppose that no initial position occurs between round $t_0$ and $t_1$.

Now let $A_0$ be the active elements of $A$ after round $t_0$ and $A_1$ be the set of active elements of $A$ after round $t_1$. Suppose that Spoiler pebbles $a$ in round $t_1$. Let $t$ lie between $t_0$ and $t_1$ and let $p$ be a pebble which is active after round $t$. Then by Fact~\ref{fact1} and as no critical or initial position occurred between round $t_0$ and round $t$, we know that $p$ has not moved since round $t_0$. Therefore $c(t) \ge c(t_0)$. 

If no new element of $A$ becomes active in round $t_1$ we similarly have that $c(t_1) \ge c(t_0)$. Otherwise we know that if $a$ is the potentially active element in the position reached after $t_1$ rounds then $a$ becomes active. Suppose Duplicator responds with $b$ on some, and therefore on every, board. We have to update our simulated game, recall that this is done by moving some pebble, say $p$, onto $a$ on the $\LHS$ and $b$ on the $\RHS$. Then the position $S(P_{t_1})$ is identical to $S(P_{t_0})$ except possibly with some pebbles removed and with $p$ lying on $a$ on the $\strucA$ board and $b$ on the $\strucB$ board.

The element $b$ is determined by the unique element $b_i \in B_0 \cup B_1$ 
which is pebbled on all $\RHS$ boards. Moreover, since in round $t_0$ we are in 
a critical position, if $a \not \in A_0$, it must be that $b_i$ was pebbled in some round $
\hat{t}$ with $t_0 < \hat{t} < t_1$. Recall that this $b$ was 
determined by Duplicator by looking at $S(P_{\hat{t}-1})$, imagining that 
Spoiler placed pebble $j:= \min\{m \, \mid \, m \not\in f_{\hat{t}-1}
(P_{\hat{t}-1}))\}$ on $a$ and then finding the optimal response. Since $S(P_{\hat{t}-1})$ is identical to $S(P_{t_0})$, except possibly with some pebbles removed, it follows that 
$c(t_1) + 1  \ge c(t_0)$. 
\end{proof}

All our work has built up to proving the following corollary, which implies Theorem~\ref{thm: main_tec}, via Lemma~\ref{lem: lbgame}.

\begin{corollary}
If Duplicator plays the $\pexists (k+1)$-LB game from position $[\{\strucSA\}, \{\strucSB\}]$ using the lower bound strategy, then in order to win Spoiler must perform at least $r$ splits, where $r$ is the number of rounds Spoiler requires to win the $\pexists k-$pebble game on $[\strucA, \strucB]$. Therefore, Duplicator gets more than $2^{r}$ points.
\end{corollary}

\begin{proof}
 Suppose we are in a critical or an initial position after $t_0$ rounds that is not immediately winning for Spoiler. Then by Lemma~\ref{lem: ignore} in order to win Spoiler must make an element of $A\setminus (A_0 \cup \{a(i) \, \mid i \in [k]\})$ active. By Corollary~\ref{cor: must-split} a split must occur before this can happen. Moreover, by Fact~\ref{fact1},  to make such an element active we must be in a critical position after the split, say at round $t_1$. Then by Lemma~\ref{lem: increment} we know that $c(t_1)+1 \ge c(t_0)$. As we have seen that $c(t)=r$, for every initial position, and as $c(\ell)=0$, it follows by Lemma~\ref{lem: increment} that Spoiler must perform at least $r$ splits in order to win. Moreover, after every split there is a single board on each side and Spoiler has not yet won. Therefore, they need to perform at least one pebble move after every split. Since the degree of each split is two the result follows.
\end{proof}

\confORfull{We now apply the following result.}{We can now get a concrete lower bound by applying the following theorem proved by Berkholz.\footnote{We should note that in the paper the theorem is stated slightly differently, see the beginning of Section~3 for an explicit statement in terms of pebble games.}}

\begin{theorem}[{\cite[Theorem~1]{DBLP:conf/cp/Berkholz14}}] \label{thm: groheLB}
For every integer $k \ge 3$, there exists a constant $\varepsilon>0$ and two positive integers $n_0, m_0$ such that for every $n>n_0$ and $m>n_0$ there exists a pair of vertex coloured graphs $\strucA$ and $\strucB$ with $|A| =n$, $|B| = m$, that are distinguishable in $\exklogic$ but on which Spoiler needs at least $\varepsilon n^{k-1} m^{k-1}$ rounds to win the $\pexists k$-pebble game. 
\end{theorem}

Since these graphs do not contain any self-loops they satisfy the disjoint positions assumptions. Moreover, we may add in $k$ independent vertices to these graphs without affecting the $k$-pebble game. Therefore, we can apply Theorem~\ref{thm: main_tec} to obtain the following slight generalisation of Theorem~\ref{thm: exlb}.

\begin{theorem} \label{thm: genlbplus}
For every integer $k \ge 4$ there exists a constant $\varepsilon >0$ and two positive integers $n_0, m_0$ such that for every $n>n_0$ and $m>n_0$ there exists a pair of $k$-ary structures $\strucA, \strucB$ with $|A| = |n|, |B| = m$ 
that can be distinguished by a sentence in $\exklogic$, but which agree on every sentence $\phi \in \exklogic$ with $\quant(\phi) < 2^{\varepsilon n^{k-2}m^{k-2}}$. 
\end{theorem}

\section{Conclusion} \label{s: con}

The main technical contributions of this paper are the lower bounds in Theorems~\ref{thm: genlb} and \ref{thm: exlb} which, to the best of our knowledge, are the first lower bounds proved using the QVT game. The first result allows us to easily prove an exponential succinctness gap between $\klogic$ and $\mathcal{L}^{k+1}$. Moreover, the bounds from Theorem~\ref{thm: exlb} are close to the upper bounds from Lemma~\ref{lem: exub}. \confORfull{}{We would also like to emphasise the techniques  used to obtain these results, as well as some of the insights obtained along the way. In particular, in Section~\ref{s: lb} we observed a peculiar resource that Spoiler sometimes has at their disposal: that of performing `freezes'. In this section, we gave one way of overcoming this hurdle to proving lower bounds as well as demonstrating a type of compositional argument, that has similarities to those used in the context of of EF games. In Section~\ref{s: exlb} we showed how, in the context of $\exklogic$, to transfer lower bounds on quantifier depth to lower bounds on quantifier number.}

A natural future task is to close the gap between the lower bound of Theorem~\ref{thm: genlb} and the upper bound of Lemma~\ref{lem: ub}. One possible way to do this would be to adapt the techniques from Section~\ref{s: exlb} to this context: that is to transfer lower bounds from quantifier depth. It should also be possible to prove an analogue to Theorem~\ref{thm: suc} but for $\exklogic$; we leave this to future work. \confORfull{}{There are, however, several hurdles to this. Firstly, the specific construction of Section~\ref{s: exlb} heavily relies on dummy variables which only works because in the associated game we use partial homomorphism, not partial isomorphism as the criteria for boards being deleted. Secondly, the spectra of freezing looms large in the general context. 

It would also be interesting to see if our results or techniques could be applied to the problem of showing an exponential succinctness gap between $\exklogic$ and $\pexists \mathcal{L}^{k+1}$; this is stated as an open problem in \cite{berkholz2019compiling}. That is to show an analogue to Theorem~\ref{thm: suc} in the existential-positive case
}

We have developed a range of insights and techniques in our analysis. It would be interesting to see if these could be deployed to prove more lower bounds using the QVT games. Particularly interesting would be the settings of bounded arity structures, such as graphs, and ordered structures. Indeed, as outlined in \confORfull{Appendix~\ref{a: complexity}}{Section~\ref{s: prelims}}, proving lower bounds on the $\QN$ needed to express properties of ordered structures could imply answers to long standing questions in complexity theory.

Another interesting avenue is to extend the QVT games to counting logics. In the context of quantifier depth, these logics have been widely studied, in part due to their close connection to the Weisfeiler–Leman algorithm \cite{cai1992optimal}. In fact, many of the quantifier depth lower bounds for $\klogic$ also apply to $k$-variable counting logic with almost no extra work, for instance, those in \cite{berkholz2016} and \cite{grohe2023iteration}. This should, as far as we can see, also hold for Theorem~\ref{thm: genlb}.

\confORfull{}{\subparagraph{Acknowledgements} Funded by the Deutsche Forschungsgemeinschaft (DFG, German Research Foundation) - project numbers 385256563; 414325841. Thank you to Christoph Berkholz for many helpful discussions which helped shape this work. In particular, his many tips regarding possible proof techniques akin to those used in proof complexity were invaluable. Thank you also to the participants and organisers of the 2023 LICS workshop \emph{Combinatorial Games in Finite Model Theory} where an early version of this work was presented; the feedback obtained helped to drastically improve this manuscript. }

\bibliography{lit} 

\clearpage

\appendix

\section{Proof of Theorem~\ref{thm: exeqiv}} \label{a: exequiv}

\begin{theorem}[Theorem \ref{thm: exeqiv} restated]
Spoiler wins the r-round $\exists^{+} k$-QVT game from position $[\classA, \classB]$ if and only if there is am $\exklogic$-formula $\phi$, with $\quant(\phi) \le r$, which $\classA$ and $\classB$ disagree on.
\end{theorem}

\begin{proof}
We mimic the proof of Theorem~\ref{thm:eqiv}. For the backward direction, we can again induct on the length of the formula. Since we no longer have negation we now assume the formula uses only the connectives $\vee, \wedge, \exists$. If $\phi$ is an atomic formula which $\classA, \classB$ disagree on then by the definition of a partial homomorphism Spoiler wins after $0$ rounds. The $\wedge$ and $\exists$ cases are the same as in Theorem~\ref{thm:eqiv}. So suppose $\phi \equiv \theta_1 \vee \theta_2$ with $\quant(\theta_i)=k_i$. Define $\classA_1 := \{\lstrucAa \in \classA \, \mid \, \lstrucAa \models \theta_1\}$ and $\classA_2 : = \classA \setminus \classA_1$. Then Spoiler should perform a split of $\classA$ so that we get two children of the root, $t_1, t_2$ with $\chi(t_1) = [\classA_1, \classB]$ and $\chi(t_2)= [\classA_2, \classB]$. By the induction hypothesis Spoiler can win from position $\chi(t_i)$ in at most $k_i$ rounds. Therefore overall Spoiler wins in at most $k_1+k_2=\quant(\phi)$ rounds. 

For the forward direction, we induct on the number of rounds $r$ Spoiler needs to win the $\pexists k$-QVT game from position $[\classA, \classB]$. We have to modify the definition of $T_{\lstrucXa}$ to now only contain atomic formulas $\phi$ such that $\lstrucXa \models \phi$ (i.e.\! we do not include negations of atomic formulas). We then define $\phi_{\mathcal{X}}$ as before but using our new definition of $T_{\lstrucXa}$.
 
The base case is trivial. For the induction step, we assume that $r \geq 1$ and that for any two sets of structures $\classX, \classY$, if Spoiler wins the $\pexists k$-QVT game in $s<r$ moves from position $[\classX, \classY]$, then there is a formula $\phi \in \klogic$, such that $\classX$ and $\classY$ disagree on $\phi$, with $\quant(\phi) = s$. The main difference here is that because the $\LHS$ and the $\RHS$ play asymmetric roles in this game we have to consider separately the case where Spoiler performs a split of $\classA$.

If before the first round Spoiler splits $\classB$ into $(\classB_i)_{i \in [\ell]}$, then as in the proof of \ref{thm:eqiv} we may assume that Spoiler has a winning strategy in at most $r-1$ moves from position $[\classA, \classB_i]$ for every $i$. Therefore there are formulas $\phi_1, \dots, \phi_{\ell}$
  with $\classA \models \phi_i$ and $\classB_i\not \models
  \phi_i$, such that $\quant(\phi_i)$ is equal to
  the number of rounds in the Spoiler winning strategy from position $[\classA, \classB_i]$. Then $\classA$ and $\classB$ disagree on
  \[ \phi \equiv \phi_{\classA} \wedge \bigwedge_{i=1}^{\ell} \phi_i.
 \]
To see this note that in the proof of Theorem~\ref{thm:eqiv}, the only time we needed the disjunct $\neg \phi_{\classB}$ was for in $\classA$ that were deleted before the split and in our present context this does not occur. Every other case is dealt with in exactly the same way.

Now suppose before the first round Spoiler splits $\classA$ into $(\classA_i)_{i \in [\ell]}$. As before we may suppose that Spoiler wins from position $[\classA_i, \classB]$ in at most $r-1$ rounds for all $i \in [\ell]$, so by the induction hypothesis we have formulas $\phi_i$ which $\classA_i$ and $\classB$ disagree on. Then the formula 
\[ 
\phi = \phi_{\classA} \wedge \bigvee_{i=1}^{\ell} \phi_i 
\]
gives us what is required.

Next, suppose Spoiler's first move is to move
pebble $i$. We may assume that Duplicator makes all possible responses. Call the resulting position $[\hat{\classA}, \hat{\classB}]$. Then Spoiler has a $(r-1)$-move winning strategy from
this position and so by the induction hypothesis there is a formula
$\theta \in \exklogic$, with $\quant(\theta) =r-1$, which $\hat{\classA}$ and $\hat{\classB}$ disagree on. Let
\[
\phi = \phi_{\classA} \wedge  \exists x_i (\phi_{\hat{\classA}} \wedge \theta). 
\]
Clearly $\quant(\phi)=r$. We will show that $\classA, \classB$ disagree on this formula. So now let $\lstrucAa \in \classA$, then Spoiler moves pebble $i$ to some $a\in A$ on $\lstrucAa$ and we get some structure $\lstruc{\strucA}{\assA(i \to a)}$. Then, since no $\RHS$ structures are deleted in this game, $ \lstruc{\strucA}{\assA(i \to a)} \in \hat{\classA}$, so by the induction hypothesis $\lstruc{\strucA}{\assA(i \to a)} \models \theta$. Therefore, taking $a$ as the witness to the existential quantification we get that $\lstrucAa \models \phi$. Similarly, by essentially the same argument as in Theorem~\ref{thm:eqiv} we get that, for every $\lstrucBb \in \classB$, $\lstrucBb \not \models \phi$. 
\end{proof}

\section{Proof of Lemma~\ref{lem: index2}}

\begin{lemma}[Lemma \ref{lem: index2} restated]
Let $G$ be a finite group, let $H \leqslant G$. Then $H$ has index two iff for every $g, h \not \in H$, $g + h \in H$.
\end{lemma}

\begin{proof}
Fix some $g \not \in H$ and consider the bijection $\sigma : G \to G$ given by $\sigma(h) = g +h$. Suppose $H$ has index two. It is easy to see that $\sigma(h) \not \in H$ for all $h \in G$, as otherwise one can show that $g\in H$, a contradiction. Since $\sigma$ is a bijection of $G$ it follows that the image of every element outside of $H$ lies in $H$. Conversely suppose for every $h \not \in G$, $\sigma(h) \in H$. This implies that $\sigma$ restricted to $H$ is a bijection from $G \setminus H$ to $H$, which implies $H$ has index two. 
\end{proof}

\section{Proof of Claim~\ref{claim}.} \label{a: dual}
 
\begin{claim}[Claim~\ref{claim} restated]
Suppose that $\{\pi(i)  \mid  i \in [k] \setminus \{p\}\} = \{s_i  \mid  i \in [k] \setminus \{\ell\}\}$, for some $\ell \in [k]$. Then when Spoiler pebbles there is some $\LHS$ board $\lstrucAa$ which is good for $\ell$ relative to $\lstrucBb$.
\end{claim}

\begin{claimproof}
This is immediate from the induction hypotheses unless Spoiler performed a split this round. So suppose Spoiler performed a split at the beginning of round $r$. As there must be more than one board on the $\LHS$, we know by the induction hypotheses that this side consists of a dual set, $\{\lstruc{A}{\assA_{i}} \, \mid \, i \in [k] \setminus \{t\} \}$, for some $t\in[k]$. By assumption the degree of the split is less than $k-1$ so there are two boards $\lstruc{A}{\assA_i}$ and $\lstruc{A}{\assA_j}$ in the same partition. Duplicator forces play to continue from this partition. As the $\RHS$ consists of only a single board and since Spoiler isn't allowed to immediately split the $\LHS$ again, due to the rules of the $k$-LB game, Spoiler must pebble from this position. If $i \neq \ell$ then $\lstruc{A}{\assA_i}$ is good for $\ell$ relative to $\lstrucBb$. Similarly if $\ell \neq j$, then $\lstruc{A}{\assA_i}$ is good for $\ell$ relative to $\lstrucBb$. Therefore as $i \neq j$ there is a $\LHS$ board which is good for $\ell$ relative to $\lstrucBb$ and so the claim holds.
\end{claimproof} 

\section{The Even Case from Section~\ref{s: lb}} \label{a: even}

For even $k$ we consider a formula over $G= (\mathbb{Z}_2^{k-1},+)$. Here we have variables $V:=\{s_1, \dots s_{k}, e_1, \dots, e_{k-1}\}$ and the following clauses.
\begin{itemize}
\item $s_k \in \mathsf{ODD}$, this is the only distinguishing clause.
\item $s_i \in \mathsf{EVEN}$ for $i \in [k-1]$.
\item  $(e_{\ell} +\sum_{i \in [k]\setminus \{\ell\}} s_i)[\ell] =0$, for  $\ell \in [k-1]$.
\item $ s_{i}[i]=0$, for  $i \in [k-1]$.
\item $e_{\ell}[i] = 0$, for $\ell, i \in [k-1]$.
\end{itemize} 

Call this formula $F_1$. It is very similar to $F$; here $s_k$ plays the role of $s_1$ and because we are now working with $\mathbb{Z}_2^{k-1}$ we have removed a few constraints. Then by chaining copies of $F_1$ together, in much the same way as we chained copies of $F$ together, we may obtain a formula $\mathcal{F}_1$ which is analogous to $F$. The only difference in the chaining procedure is that we add constraints of the form $e_{\ell}(j) + s_k(j+1) \in \mathsf{ODD}$ rather than $e_{\ell}(j) + s_1(j+1) \in \mathsf{ODD}$. By essentially the same arguments as those given in Section~\ref{s: lb} one can prove that Duplicator can get $2^{\Omega(n)}$ points in the $k$-LB game on $\strucA(\mathcal{F}_1)$ and $\strucB(\mathcal{F}_1)$. 

\section{Proof of Lemma~\ref{lem: start}} \label{a: start}

To aid the reader we restate the start strategy here. Suppose that Spoiler plays on the start gadget. If Spoiler plays some $g_i^0$ then Duplicator replies with $g_i^0$ unless either (1) this response is not valid or (2) $i=0$ and an element of the form $r_j^0$ is pebbled on $\lstrucSBb$ or (3) $i=1$ and an element of the form $l_j^0$ is pebbled on $\lstrucSBb$. If (1), (2) or (3) hold Duplicator instead plays $g_i^1$. If Spoiler instead plays $y_i^0$, $y \in \{l, r\}$, Duplicator plays $y_i^1$ whenever it is valid and $y_i^0$ otherwise.

\begin{lemma}[Lemma \ref{lem: start} restated]
Suppose that in round $r$ of the $\pexists (k+1)$-LB game starting from position $[\{\strucSA\}, \{\strucSB\}]$, Spoiler plays $p\to x \in G$. Suppose that in every previous round in which Spoiler played an element of $G$ Duplicator followed the start strategy on every $\RHS$ board. Then the move recommended by the start strategy in round $r$ is valid on every $\RHS$ board.
\end{lemma}

\begin{proof}
To show this we induct on $r$ while maintaining the following auxiliary hypotheses on every $\RHS$ board.
 \begin{enumerate}
 \item If $g_0^1$ or $g_1^0$ is pebbled then no element of the form $l_j^0$ is pebbled.
 \item If $g_0^0$ or $g_1^1$ is pebbled no element of the form $r_j^0$ is pebbled.
 \item If an element of the form $l_i^0$ is pebbled then no element of the form $r_j^0$ is pebbled.
\item If an element of the form $r_i^0$ is pebbled then no element of the form $l_j^0$ is pebbled.
 \end{enumerate}
 
So suppose (1)-(4) hold at the end of round $r-1$ on $\lstrucSBb$, a $\RHS$ board, and that in round $r$ Spoiler plays $p \to x$. If $x \not \in G$ then clearly our hypotheses cannot break. Suppose $x = g_0^0$. If Duplicator responds with $g_0^0$ on $\lstrucSBb$ then by the specification of the start strategy this is a valid response. Moreover, the strategy dictates that if they make such a move no element of the form $r_j^0$ is pebbled, so the induction hypotheses are maintained. 

So suppose instead that Duplicator replies with $g_0^1$ on $\lstrucSBb$. Then either the response $g_0^0$ is not valid or an element of the form $r_i^0$ is pebbled on $\lstrucSBb$. Suppose the response $g_0^0$ is invalid. Then the $\LHS$ board models $E$, as $\{g_0^0\} \otimes S(B)^k \subset R_s^{\strucSB}$. Then since $g_0^0$ is an invalid reply, $g_1^0$ is pebbled on $\lstrucSBb$ and so the move $p \to g_0^1$ respects $E$. Furthermore, the $\LHS$ board is not a model of $R_s$ since both $g_0^0, g_1^0$ are pebbled. So Duplicator's response is valid. Since $g_1^0$ is pebbled on $\lstrucSBb$ we know by (1) that no element of the form $l_j^0$ is pebbled on $\lstrucSBb$, so the induction hypotheses are maintained. So suppose instead that Duplicator plays $g_0^1$ because some $r_i^0$ is pebbled on $\lstrucSBb$. Then $g_1^1$ is not pebbled on $\lstrucSBb$ by (2), therefore the response respects $E$. Also by (4) no element of the form $l_j^0$ is pebbled on $T_{\ell}$. Finally since $r_j^0$ and $g_0^0$ are pebbled on the $\LHS$, it does not model $R_s$, so we again see that Duplicator's response is valid. This concludes the case $x=g_0^0$, the case $x=g_1^0$ is similar.

Next, suppose that $x= l_i^0$. Then Duplicator plays $l_i^1$ on $\lstrucSBb$ whenever this is a valid response; such a response trivially respects the induction hypotheses. We claim that if $(g_0^0, l_1^0, \dots, l_k^0)$ is covered on the $\LHS$ then $l_i^1$ is  a valid response. To see suppose note that if $g_0^0$ is covered on $\lstrucSBb$ any Duplicator response is valid, as $\{g_0^0\} \otimes S(B)^k \subset R_s^{\strucSB}$. Otherwise, $g_0^1$ is pebbled on $\lstrucSBb$ and so by (1) no element of the form $l_j^0$ is pebbled. Therefore, after Duplicator replies with $l_i^1$, $(g_0^1,  l_1^1, \dots, l_k^1) \in R_s^{\strucSB}$ is covered. The claim follows. So suppose $l_i^1$ is not valid, so that Duplicator plays $l_i^0$. Therefore, we must be an initial position, but then $l_i^0$ is a valid response, since $\{l_i^0\} \otimes (S(B) \setminus \{g_0^1\})^k \subset R_s^{\strucSB}$. Also, since we are in an initial position, the induction hypotheses are maintained. The case $x=r_i^0$ is similar.
\end{proof}
 
\end{document}